\documentclass[10pt]{amsart}
\usepackage{amsmath}
  \usepackage{graphics} 
\usepackage{graphicx}  
  \usepackage{epsfig} 
 \usepackage[colorlinks=true]{hyperref}
  \usepackage{float}
  \usepackage{epstopdf}
  \usepackage{caption}
\usepackage{amsmath}
\usepackage{wrapfig}
\usepackage{graphicx}
\usepackage{subcaption}
\usepackage{amssymb}

\newtheorem{theorem}{Theorem}[section]

\newtheorem{lemma}[theorem]{Lemma}

\newtheorem{conjecture}{Conjecture}

\theoremstyle{definition}

\newtheorem{remark}{Remark}

\usepackage{enumitem}
\usepackage[toc]{appendix}
\title[Fear effect]
    {The effect of ``fear" on two species competition}

\begin{document}

\maketitle
\centerline{\scshape  Vaibhava Srivastava $^{1}$, Eric M. Takyi $^{2}$ and Rana D. Parshad $^{1}$}
\medskip
{\footnotesize
\centerline{ 1)Department of Mathematics,}
 \centerline{Iowa State University,}
   \centerline{Ames, IA 50011, USA.}
  \medskip
\centerline{ 2)Department of Mathematics and Computer Science,}
 \centerline{Ursinus College,}
   \centerline{Collegeville, PA 19426, USA.}
 
}

%
%
%


\begin{abstract}
Non-consumptive effects such as fear of depredation, can strongly influence predator-prey dynamics. These effects have not been as well studied in the case of purely competitive systems, despite ecological and social motivations for the same. In this work we consider the classic two species ODE and PDE Lokta-Volterra competition models, where \emph{one} of the competitors is ``fearful" of the other. We find that the presence of fear can have several interesting dynamical effects on the classical scenarios of weak and strong competition, and competitive exclusion. Notably, for fear levels in certain regimes, we show bi-stability between interior equilibrium and boundary equilibrium is possible - contrary to the classical strong competition situation where bi-stability is only possible between boundary equilibrium. Furthermore, in the spatially explicit setting, the effects of several spatially heterogeneous fear functions are investigated. 
In particular, we show that under certain $\mathbb{L}^{1}$ restrictions on the fear function, a weak competition type situation can change to competitive exclusion. Applications of these results to ecological as well as sociopolitical settings are discussed, that connect to the ``landscape of fear" (LOF) concept in ecology. 
\end{abstract}


\section{Introduction}

Fear, is defined as,
 
\begin{quotation}
\emph{An unpleasant emotion caused by the belief that someone or something is dangerous} \cite{mer02}. 
\end{quotation}

It is a complex emotion, that is critical as a safety measure, and can trigger the ``fight or flight" response \cite{Brac04} - in particular it can change the way one acts, even when there is no threat present \cite{Cress11}. In predator-prey systems, this is most naturally observed among prey, due to their perceived threat of depredation \cite{Sher20}. This perception can lead to non-consumptive effects or trait-mediated interactions, which are behavioral, morphological or physiological changes in prey phenotype, due to this threat \cite{Peac13, Sher20}. Such effects are known to strongly influence predator-prey dynamics \cite{Peck08}. 
From a mathematical viewpoint, the effects of fear in predator-prey systems has been intensely investigated since the seminal work of Brown et. al. \cite{Brown99}, where optimal foraging theory is extended to consider a game theoretic setup, played out by predator and prey, exhibiting stealth and fear, in which an animal follows a map or a ``landscape of fear" (LOF), which describes its predation risk while it navigates the physical landscape. In recent work, Wang et. al. \cite{Wang16}, model fear of depredation, as a (predator) density dependent effect, that negatively effects the prey population. In essence, the prey's growth rate is modeled as a monotonically decreasing function of predator density. Dynamically, a key finding in \cite{Wang16} is that under the parametric restrictions of a Hopf bifurcation, an increase in the fear parameter (and prey's birth rate parameter) can alter the direction of a Hopf bifurcation from supercritical to subcritical. Thus, fear enables both supercritical and subcritical Hopf bifurcations, contrary to only the supercritical bifurcations found in classical predator-prey systems. In essence, the fear effect can change the fundamental cylical patterns of predator-prey dynamics, leading to large scale ecological consequences \cite{Ab00}. 

These results have since initiated a host of activities in diverse ecological scenarios such as when refuges are present \cite{Wang19, Zhang19}, when the prey has tendencies to avoid predators \cite{Wang17}, or when the predators responses are influenced by interference pressures, for instance, via a Beddington-DeAngelis functional response \cite{Palb19}. Various works have considered the fear effect in case of group defense by the prey \cite{Das_Sam21, Das21}.
 It has been investigated in the context of cooperative and competitive systems within the larger predator-prey context. These include the fear effect when predators are cooperating \cite{Pal19} in the hunting process, or when they are hunting for competing prey \cite{Yousef22}. These effects have been investigated in the three and multi-species settings as well \cite{Pandey18, Panday19} where fear can damp population explosions \cite{Verma21}. Various authors have considered the fear effect in a stochastic setting \cite{Das18} as well as a spatially explicit setting, in the context of taxis type movements, as well as pattern formation \cite{Wu18, Wang18}. It can also lead to chaotic dynamics \cite{Duan19}. However, the effect of fear has been far less investigated in classical monotone systems, such as purely cooperative or competitive two species systems - that are outside the predator-prey setting.

 Competition among two species, typically modeled via the Lotka–Volterra competition model and its variants have been intensely investigated in the last few decades. These models take into account growth and inter/intraspecific competition \cite{Lou2008}, and predict well-observed states in biology of co-existence, competitive exclusion of one competitor, and bi-stability, and find diverse applications in ecology and invasion science \cite{Boer86, Chess00, O89, S97, A06}.
 There are several ecological motivations for competitors being fearful of each other. This is perhaps most naturally seen to occur with intraguild predation - a widespread phenomenon in many food webs, where competitors will kill and consume each other \cite{Pol89}. Recent evidence of non-consumptive effects exerted by intraguild predator mites (\emph{Blattisocius dentriticus}) on their competitor (\emph{Neoseiulus cucumeris}) show this can be an important factor in determining food web dynamics in biological control \cite{Gu22, V96}. However, there is strong evidence for fear in purely competitive two species systems \emph{without} predatory effects. Barred owls (\emph{Strix varia}) are a species of owl, native to eastern North America. They have expanded their range westward over the last century and are considered invasive in western North America. Currently, their range overlaps with the spotted owl (\emph{Strix occidentalis}), which is native to the north west and western North America. This has resulted in intense competition between the two species \cite{Long19}. Barred owls exert a strong negative influence on spotted owls, threatening their possible competitive exclusion \cite{Van11}. Field observations report frequent barred owl attacks on spotted owls, and even on surveyors imitating spotted owl calls \cite{Gut04}. There is also evidence of barred owls aggressively chasing spotted owls out of shared habitat - but \emph{not} the opposite \cite{Wiens14}. Such evidence clearly motivates considering fear type dynamics into a purely competitive two species model where one of the competitors is fearful of the other.

There are also several socio-economic-political settings, where pure competitors may be fearful of each other. Small/new businesses may be fearful of large businesses, due to their already large market share \cite{Byun20}. But large business may also be fearful of small local businesses, due to their familiarity with local nuances, that may yield competitive advantage at a small local scale \cite{Moen99}. Fear is also conceivable among two competing political parties, where the weaker party on a national scale, may have a stronger voter bank at a regional scale \cite{Lof16}. Or perhaps two warring drug cartels, where the weaker cartel has certain local/territorial strongholds \cite{Wain16, Sul02} - within which they might be able to induce fear among the stronger cartel \cite{Wain16}. Such phenomenon becomes even more interesting in the spatially explicit case where this fear could be heterogeneous in the spatial domain of interest. This connects back to the LOF concept, where the fear function is essentially the map that describes how the fear levels change as a species disperses over a physical landscape. 
 
Motivated by all of the affore mentioned sociopolitical, economic as well as ecological settings, the current manuscript considers the effect of fear in a competitive two species system. We restrict our analysis to the case where only \emph{one} of the competitors is fearful of the other. Our investigations show that:
 
 \begin{itemize}
 
 \item Sufficiently large fear can change a situation of competitive exclusion, to a strong competition type scenario, where there is bi-stability between boundary equilibrium. See Fig.~\ref{fig:Fig1} (C). Dynamically, this occurs via a transcritical bifurcation. This is shown via Lemma \ref{thm:tc}, see Fig.~\ref{fig:trans}.
 
 \item Fear in a certain parametric regime can change a situation of competitive exclusion to bi-stability between boundary equilibrium and interior equilibrium, see Fig.~\ref{fig:ode_two_postive} and Fig.~\ref{fig:Fig1} (B). Dynamically, this occurs via a saddle-node bifurcation. This is shown via Lemma \ref{thm:sad}, see Fig.~\ref{fig:sad_node}. This is in sharp contrast with classical competition theory, where bi-stability occurs only between boundary equilibriums.
 
 \item Sufficiently large fear can change a situation of weak competition to a competitive exclusion type scenario. This is shown via Lemma \ref{lem:ce_ode_2}, see Fig.~\ref{fig:ode_wc}.
 
 \item Fear cannot qualitatively change a strong competition type scenario. This is shown via Lemma \ref{ce_strong}, see Fig.~\ref{fig:ode_st_comp}.
 Also, fear cannot produce periodic orbits. This is demonstrated via Lemma~\ref{lem:dc1}.
 
 \item In the spatially explicit setting, comparison theory is used to determine point-wise restrictions on the fear functions such that competitive exclusion or strong competition type dynamics abounds. These are shown via Theorem \ref{thm:ce_1_pde}, Theorem \ref{thm:ce_2_pde} and Theorem \ref{st_pde}, see Figs.~[\ref{fig:pde_ce1_pde1},\ref{fig:pde_ce2_pde1},\ref{fig:pde_sc1_pde1}].

 \item In the spatially explicit setting, fear can change a situation of weak competition to a competitive exclusion type scenario, for fear functions with certain $\mathbb{L}^{1}$ restrictions. This is shown via Theorem \ref{thm:cotoce_1_pde} and Lemma \ref{lem:ff1}, see Figs.~[\ref{fig:remark_5},\ref{fig:remark_51},\ref{fig:remark_52}]. In particular the fear functions need not lie uniformly above the critical fear levels derived in the ODE case via Lemma \ref{lem:ce_ode_2}. 
 
 \item Various heterogeneous fear functions are constructed to demonstrate these results numerically, see Fig.~\ref{fig:fear_plot}. Applications of these to ecological as well as socio-political settings are discussed in section \ref{disc}.

 \end{itemize}

\section{The ODE case}
\subsection{Model formulation}

Consider the classical two species Lotka-Volterra ODE competition model,

\begin{equation}\label{eq:GeneralEquation}
\left\{ \begin{array}{ll}
\dfrac{du }{dt} &~ = u (a_{1}-b_{1}u - c_{1}v) ,\\[2ex]
\dfrac{dv }{dt} &~ =  v (a_{2}-b_{2}v - c_{2}u),
\end{array}\right.
\end{equation}
%
%
where $u$ and $v$ are the population densities of two competing species, $a_1$ and $a_2$ are the intrinsic (per capita) growth rates, $b_1$ and $b_2$ are the intraspecific competition rates, $c_1$ and $c_2$ are the interspecific competition rates. All parameters considered are positive. The dynamics of this system are well studied \cite{Murray93}. We recap these briefly,

\begin{itemize}
    \item $E_0 = (0,0)$ is always unstable.
    \item $E_u = (\frac{a_1}{b_1},0)$ is globally asymptotically stable if $\dfrac{a_{1}}{a_{2}} > \max\left\lbrace\dfrac{b_{1}}{c_{2}},\dfrac{c_{1}}{b_{2}}\right\rbrace$. Herein $u$ is said to competitively exclude $v$.
    \item $E_v = (0,\frac{a_2}{b_2})$ is globally asymptotically stable if $\dfrac{a_{1}}{a_{2}}<\min\left\lbrace\dfrac{b_{1}}{c_{2}},\dfrac{c_{1}}{b_{2}}\right\rbrace$. Herein $v$ is said to competitively exclude $u$. 
    \item $E^* = \Big(\frac{a_1b_2-a_2c_1}{b_1b_2-c_1c_2},\frac{a_2b_1-a_1c_2}{b_1b_2-c_1c_2}\Big)$ exists when $b_1b_2-c_1c_2 \neq 0$. The positivity of the equilibrium holds if $\frac{c_2}{b_1}<\frac{a_2}{a_1}<\frac{b_2}{c_1}$ and is globally asymptotically stable if $b_1b_2-c_1c_2>0$. This is said to be the case of weak competition.
   
   \item If $b_1b_2-c_1c_2<0$, then $E^* = \Big(\frac{a_1b_2-a_2c_1}{b_1b_2-c_1c_2},\frac{a_2b_1-a_1c_2}{b_1b_2-c_1c_2}\Big)$ is unstable as a saddle. In this setting, one has initial condition dependent attraction to either $E_u(\frac{a_1}{b_1},0)$ or $E_v(0,\frac{a_2}{b_2})$. This is the case of strong competition.
\end{itemize}

%
%
%
%
%
We proceed by considering the effects of fear on the classical model \eqref{eq:GeneralEquation}, when \emph{one} of the competitors is fearful of the other.

%
%

\subsection{The case of $v$ fearing $u$}

We consider the case of the competitor $v$ being fearful of $u$. Thus in the classical model \eqref{eq:GeneralEquation}, we model the fear effect as in \cite{Wang16}, where the growth rate of the fearful competitor $v$, is not constant but rather density dependent. Essentially, the growth rate
is decreased by a factor $\approx \frac{1}{1+k u}$, where $k \geq 0$ is a fear coefficient. Thus a higher density of the competitor $u$ increases the fear in $v$. When $k=0$, the assumption is there is no fear and one recovers the classical model \eqref{eq:GeneralEquation}. If fear is present, we obtain the following ODE model for two competing species $u$ and $v$, where $v$ is fearful of $u$.

\begin{align}\label{eq:ODE2}
	\begin{split}
		\dfrac{du}{dt} &= a_1 u -b_1 u^2 -c_1 uv, \\
		\dfrac{dv}{dt} &=  \dfrac{a_2 v}{1+ku}  -b_2 v^2 -c_2 uv.
	\end{split}
\end{align}

\subsubsection{Existence}

	The nullclines associated with the problem \eqref{eq:ODE2} are
	\[ u(a_1 -b_1 u -c_1 v)=0 \quad \text{and} \quad  v \Big( \dfrac{a_2 }{1+ku}  -b_2 v -c_2 u\Big).\]
	Hence, the boundary equilibrium points are obtained by substituting $u=0$ and $v=0$ in the above equations of the nullclines, respectively. Denote the boundary equilibrium points as 
	$\widehat{E}_1=(0,0)$, $\widehat{E}_2=(\frac{a_1}{b_1},0)$ and  $\widehat{E}_3=(0,\frac{a_2}{b_2})$.
	
	For the interior equilibrium, substitute $u^*=\frac{a_1}{b_1} - \frac{c_1}{b_1} v^*$ in the second nullcline equation, i.e.,
	
	\begin{align*}
		\begin{split}
			\dfrac{a_2 }{1+k \Big( \frac{a_1}{b_1} - \frac{c_1}{b_1} v^*\Big) }  -b_2 v^* -c_2 \Big(\frac{a_1}{b_1} - \frac{c_1}{b_1} v^*\Big) =0.
		\end{split}
	\end{align*}
	
	On simplification, we have that $v^*$ solves a quadratic equation of the form
	$A(v^*)^2+Bv^*+C=0$, where 
	\begin{equation}\label{quad_param}
		\begin{split}
			A&=c_1 k(b_2 b_1 -c_1 c_2),\\
			B&=b_1(c_1 c_2-b_2 b_1)-a_1k(b_2 b_1-2c_1 c_2),\\
			C&=b_1(a_2 b_1 -a_1 c_2)-a_1^2 c_2 k.
		\end{split}
	\end{equation}
	Let 
	\begin{equation}\label{quad_roots}
			v_{1,2}^* = \frac{-B \pm \sqrt{B^2-4AC}}{2A}
	\end{equation}
   be the two roots of above qudratic equation. WLOG assume $v_1^*<v_2^*.$  Moreover, consider the following parametric restriction
	 \begin{align}\label{pos_u_null}
		b_1(a_1 c_2-a_2 b_1)+a_1^2 c_2 k < \Big[ (b_2 b_1 -c_1 c_2)(2 a_1k-b_1) -a_1kc_1 c_2 \Big] \dfrac{a_1}{c_1}.
	\end{align}
	
	We can prove the existence of a positive equilibrium point $\widehat{E}_4$ with the choice of specific parameters. Let us use Descartes's rule of sign to establish some sufficient conditions for the existence of one or two positive equilibrium points.
	
	\textit{Two positive equilibrium points:} Under the assumption $A>0,B<0$ and $C>0$, i.e., $b_2b_1>2c_1c_2>c_1c_2$ and $k<\dfrac{1}{a_1^2c_2} \Big( b_1^2a_2-a_1c_2b_1\Big),$ we have two positive roots. In order to claim that these two roots correspond to two positive interior equilibria, we need some extra assumption given by:
	\[ v_1^*<v_2^*:=\dfrac{-B + \sqrt{B^2-4AC}}{2A} < \dfrac{a_1}{c_1} \implies -C < (A \dfrac{a_1}{c_1}+B)\dfrac{a_1}{c_1}.\]
	Hence, if $b_1(a_1 c_2-a_2 b_1)+a_1^2 c_2 k < [ (b_2 b_1 -c_1 c_2)(2 a_1k-b_1) -a_1kc_1 c_2] \frac{a_1}{c_1},$ we have two positive interior equilibrium points $\widehat{E}_4=(u_i^*,v_i^*)$ for $i=1,2$.
	
	\textit{One positive equilibrium point:} Under the assumption $A>0,B<0$, i.e., $b_2b_1>2c_1c_2>c_1c_2$, we have at least one positive root of the quadratic equation. If $C<0$, which is $k>\dfrac{1}{a_1^2c_2} \Big( b_1^2a_2-a_1c_2b_1\Big)$ along with $\eqref{pos_u_null}$ gives existence of one positive equilibrium point $\widehat{E}_4=(u^*,v^*)$. Moreover, if $C>0$, which is $k<\dfrac{1}{a_1^2c_2} \Big( b_1^2a_2-a_1c_2b_1\Big)$, and $v_1^*<\frac{a_1}{c_1}$ and $v_2^*>\frac{a_1}{c_1}$, then we have existence of one positive equilibrium point $\widehat{E}_4=(u^*,v^*)$.
	
	We formulate all these restrictions as an existence theorem:

\begin{theorem}\label{thm:exist}
	For the given ODE system \eqref{eq:ODE2}, we always have three boundary equilibrium points, namely  $\widehat{E}_1:=(0,0),\widehat{E}_2:=(\frac{a_1}{b_1},0)$ and $\widehat{E}_3:=(0,\frac{a_2}{b_2})$. 
	For the case of two positive interior equilibrium points $\widehat{E}_4=\Big(u_i^*,v_i^*\Big)_{i=1,2}$, we have the following parametric restrictions:
	\begin{align*}
		b_2b_1>&2c_1c_2>c_1c_2, \quad k<\frac{1}{a_1^2c_2} \Big[ b_1^2a_2-a_1c_2b_1\Big] \quad
		\& \quad \eqref{pos_u_null} \hspace{0.1in} \text{holds true}.
	\end{align*}
	Lastly, for the case of one positive interior equilibrium point $\widehat{E}_4=(u^*,v^*)$, we have either one of the following choices of parametric restrictions:
	\begin{enumerate}
		\item $b_2b_1>2c_1c_2>c_1c_2, \quad k>\dfrac{1}{a_1^2c_2} \Big( b_1^2a_2-a_1c_2b_1\Big)$  and $\eqref{pos_u_null}$ holds true.
		\item $b_2b_1>2c_1c_2>c_1c_2, \quad k<\dfrac{1}{a_1^2c_2} \Big( b_1^2a_2-a_1c_2b_1\Big)$, and $v_1^*<\frac{a_1}{c_1}$ and $v_2^*>\frac{a_1}{c_1}$.
		\item $b_2b_1<c_1c_2<2c_1c_2,\quad k>\dfrac{1}{a_1^2c_2} \Big( b_1^2a_2-a_1c_2b_1\Big)$, and $v_1^*<\frac{a_1}{c_1}$ and $v_2^*>\frac{a_1}{c_1}$,
	\end{enumerate}
	where $v_i^*$ are the roots of the quadratic equation defined as $\eqref{quad_roots}.$
\end{theorem}

We now provide several lemmas, so that we can compare the effect of fear to the classical competition case.

\begin{lemma}
\label{lem:cotce}
Consider the given ODE system \eqref{eq:ODE2}, with $k=0$, s.t we are in the weak competition setting with $b_2b_1>2c_1c_2$.
Then for a fear coefficient $k$ s.t. $k > k_{c} = \dfrac{1}{a_1^2c_2} \Big( b_1^2a_2-a_1c_2b_1\Big)$, $(\frac{a_1}{b_1},0)$ is globally asymptotically stable. That is, $u$ will competitively exclude $v$.
\end{lemma}


\begin{remark}
	For the existence of one positive interior equilibrium point for the given system \eqref{eq:ODE2}, we have either one of the following choices of parametric restrictions: For weak competition, 
	\begin{align}\label{one_post_weak}
		b_2b_1>2c_1c_2>c_1c_2 \quad \& \quad \dfrac{c_2}{b_1}<\dfrac{a_2}{a_1} <\dfrac{b_2}{c_1},
	\end{align}
	and for strong competition,
	\begin{align}\label{one_post_strong}
		b_2b_1<c_1c_2<2c_1c_2 \quad \& \quad \dfrac{c_2}{b_1}>\dfrac{a_2}{a_1} >\dfrac{b_2}{c_1}.
	\end{align}
	
\end{remark}

\subsubsection{Linear Stability Analysis}

We next perform stability analysis on the equilibrium points of system \eqref{eq:ODE2}. The Jacobian matrix of system \eqref{eq:ODE2} is given by
\begin{equation}\label{jacob_second}
	\widehat{J}^*(u^*,v^*)=
	\left(
	\begin{array}{cc}
		a_1-2b_1u^*-c_1v^* & -c_1 u^* \\
		-\dfrac{a_2 k v^*}{(k u^*+1)^2}-c_2 v^* & \dfrac{a_2}{k u^*+1}-2 b_2 v^*-c_2 u^* \\
	\end{array}
	\right).
\end{equation}

We state the following lemmas.
\begin{lemma}
	$\widehat{E}_1$ is locally unstable.
\end{lemma}
\begin{proof}
	On evaluating Eq.$(\ref{jacob_second})$ at $\widehat{E}_1$, we have
	\begin{equation*}
		\widehat{J}^*(\widehat{E}_1)=
		\left(
		\begin{array}{cc}
			a_1 & 0 \\
			0 & a_2 \\
		\end{array}
		\right).
	\end{equation*}
	Being a triangular matrix, we know that the above matrix has two positive eigenvalues $a_1$ and $a_2$. Hence, the equilibrium point $\widehat{E}_1$ is locally unstable.
\end{proof}

\begin{lemma}\label{lem:ce_ode_1}
	$\widehat{E}_2$ is locally stable iff $k>\dfrac{a_2 b_1^2-c_2a_1b_1}{a_1^2c_2}$.
\end{lemma}
\begin{proof}
	We again evaluate Eq.$(\ref{jacob_second})$ at $\widehat{E}_2$ and obtain
	\begin{equation*}
		\widehat{J}^*(\widehat{E}_2)=
		\left(
		\begin{array}{cc}
			-a_1 & -\dfrac{c_1a_1}{b_1} \\
			0 &  \dfrac{a_2b_1}{b_1 + k a_1} -\dfrac{c_2 a_1}{b_1} \\
		\end{array}
		\right).
	\end{equation*}
	Being a triangular matrix, the above matrix has two eigenvalues, $\lambda_1=-a_1$ and $\lambda_2=\frac{a_2b_1}{b_1 + k a_1} -\frac{c_2 a_1}{b_1}$. As $\lambda_1$ is always negative, if we can show that $\lambda_2$ is negative, we are done. We make the assumption that,
	\[ k>\dfrac{a_2 b_1^2-c_2a_1b_1}{a_1^2c_2}  \iff \lambda_2 =\dfrac{a_2b_1}{b_1 + k a_1} - \dfrac{c_2 a_1}{b_1}<0. \]
	Therefore, the boundary equilibrium point $\widehat{E}_2$ is locally stable.
\end{proof}
Local stability of $\widehat{E}_2$ actually implies global stability, we can see this via a simple geometric argument.

\begin{lemma}\label{lem:ce_ode_2}
$\widehat{E}_2$ is globally stable if $k > k_{c} = \dfrac{a_2 b_1^2-c_2a_1b_1}{a_1^2c_2}$.
\end{lemma}

\begin{proof}
	Consider the nullclines of $u$ and $v$, where,
	\[ v = \dfrac{1}{c_1}  (a_1 -b_1 u) \quad \& \quad v = \dfrac{1}{b_2} \Big( \dfrac{a_2}{1+ k u} -c_2 u \Big). \]
	In order to establish the global stability of $\widehat{E}_2$, via the geometry of the nullclines, it suffices to show that 
	\[  \Big[ \dfrac{a_1-b_1 u}{c_1} \Big] > \Big[ \dfrac{1}{b_2} \Big( \dfrac{a_2}{1+ k u} - c_2 u \Big)\Big]  \]
	when $u=\frac{a_1}{b_{1}}$,
	
	i.e., when $k>\frac{a_2 b_1^2-c_2a_1b_1}{a_1^2c_2}$. Herein, the $v$-nullcline lies completely below the $u$-nullcline, and via the convexity of the $v$-nullcline, it lies completely below the straight line connecting its $v$ and $u$ intercepts - which lies completely below the $u$-nullcline.
	Now, via the standard theory of competition and a comparison argument, where $v$ is compared to the $\tilde{v}$ that is a solution to the straight line nullcline connecting the $v$ and $u$ intercepts of the $v$-nullcline, we have the global stability of $\widehat{E}_2$.
\end{proof}

%

\begin{figure}[h]
	\begin{subfigure}[b]{.475\linewidth}
		\includegraphics[width=\linewidth,height=2in]{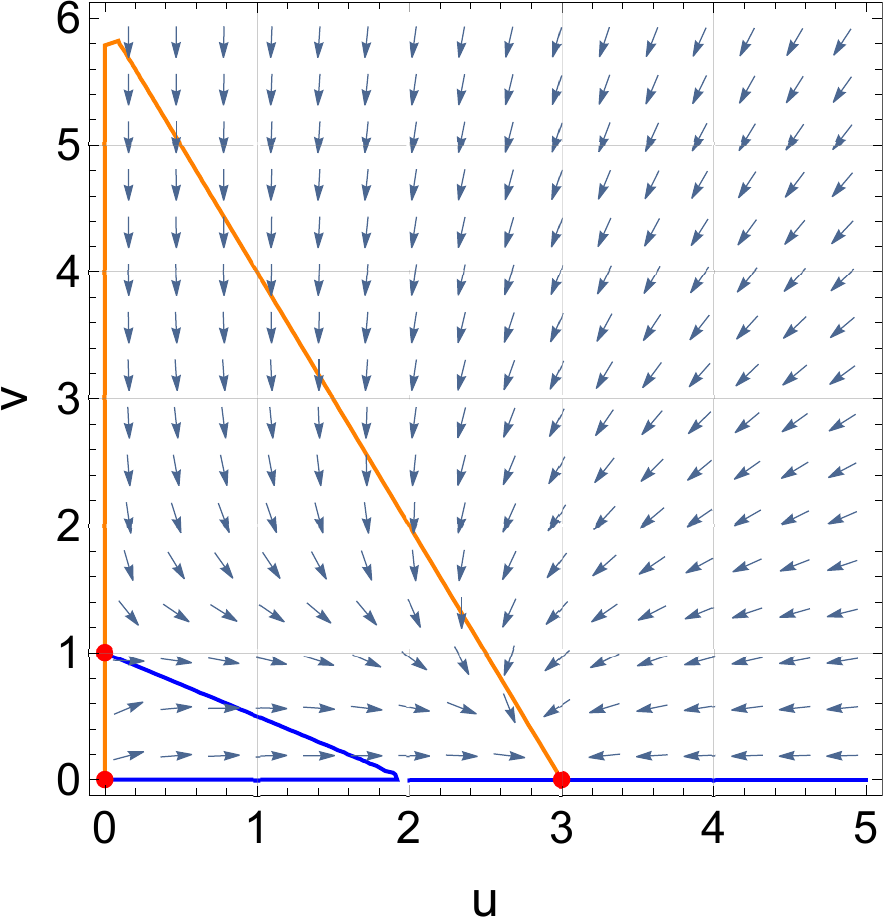}
		\caption{$k=0$ }
	\end{subfigure}
	\hfill
	\begin{subfigure}[b]{.49\linewidth}
		\includegraphics[width=\linewidth,height=2in]{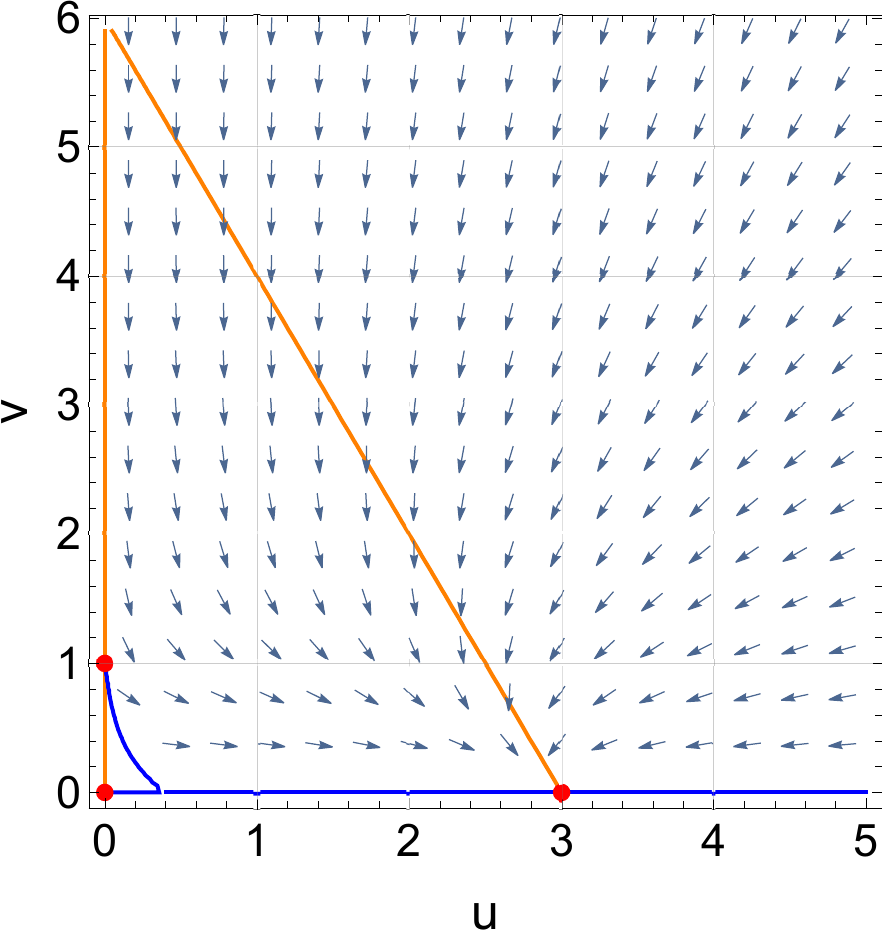}
		\caption{$k=10$}
	\end{subfigure}
	\caption{Phase plots showing dynamics under competition exclusion parametric restriction with $f=0$. The orange curve is the u-nullcline and blue curve is the v-nullcline. Here, $(u^*,0)$ always wins. Parameters used are $a_1=3, a_2=1,b_1=1, b_2=1, c_1=0.5, c_2=0.5$.}
	\label{fig:ode_ce}
\end{figure}

\begin{lemma}
	The equilibrium point $\widehat{E}_3$ is locally stable iff $a_1b_2<c_1a_2$.
\end{lemma}
\begin{proof}
	Similar evaluation of  Eq.$(\ref{jacob_second})$ at $\widehat{E}_3$ yields
	\begin{equation*}
		\widehat{J}^*(\widehat{E}_3)=
		\left(
		\begin{array}{cc}
			a_1- \dfrac{c_1a_2}{b_2} & 0 \\
			-\dfrac{ka_2^2}{b_2} -\dfrac{c_2a_2}{b_2}& -a_1 \\
		\end{array}
		\right).
	\end{equation*}
	Being a triangular matrix, the above matrix has two eigenvalues, $\lambda_1=-a_1$ and $\lambda_2=a_1- \dfrac{c_1a_2}{b_1}$. Under the assumed parametric restriction,
	\[ a_1b_2<c_1a_2 \iff \lambda_2<0.\]
	Hence, the equilibrium point $\widehat{E}_3$ is locally stable.
	
\end{proof}


\begin{lemma}\label{lem:co_ode_1}
The interior equilibrium $\widehat{E}_4$ exists and is locally stable if 
$k < \frac{1}{a_{2}} \left(\frac{b_{1}b_{2}}{c_{1}} - c_{2}\right).$
\end{lemma}

\begin{proof}
	On evaluating Eq.$(\ref{jacob_second})$ again at $\widehat{E}_4$, we have
	\begin{equation*}
		\widehat{J}^*(\widehat{E}_4)=
		\left(
		\begin{array}{cc}
			-b_1 u^* & -c_1 u^* \\
			-\dfrac{ka_2v^*}{(1+ku^*)^2} -c_2 v^*&  -b_2 v^* \\
		\end{array}
		\right).
	\end{equation*}
	For the local stability of $\widehat{E}_4$, it is enough to show that $Trace(\widehat{J}^*(\widehat{E}_4))<0$ and $Det(\widehat{J}^*(\widehat{E}_4))>0.$
	Simple computations yield
	\[  Trace(\widehat{J}^*(\widehat{E}_4)) = -b_1 u^* -b_2 v^*<0,\]
	and
	\[ Det(\widehat{J}^*(\widehat{E}_4) = b_1b_2 u^* v^* - c_1u^* \Big( \dfrac{ka_2v^*}{(1+ku^*)^2} +c_2 v^* \Big) =  u^*v^* \Big\{ b_1b_2 -  c_1 \Big( \dfrac{k a_2}{ ( 1 + ku^* )^2} +c_2 \Big) \Big\} .  \]
	Note that
	\[  k a_{2} > \frac{k a_{2}}{(1+ku^*)^2}.\]
	Therefore, if $k$ is chosen s.t, $k < \frac{1}{a_{2}} \left(\frac{b_{1}b_{2}}{c_{1}} - c_{2}\right)$, then,
	\[  b_{1} b_{2} > c_{1}(ka_{2} + c_{2})  > c_1 \Big[ \dfrac{k a_2}{ ( 1 + ku^* )^2} +c_2 \Big] \implies Det(\widehat{J}^*(\widehat{E}_4))>0,\]
	 and the result follows.
	
\end{proof}

\begin{figure}[h]
	\begin{subfigure}[b]{.49\linewidth}
		\includegraphics[width=\linewidth,height=2in]{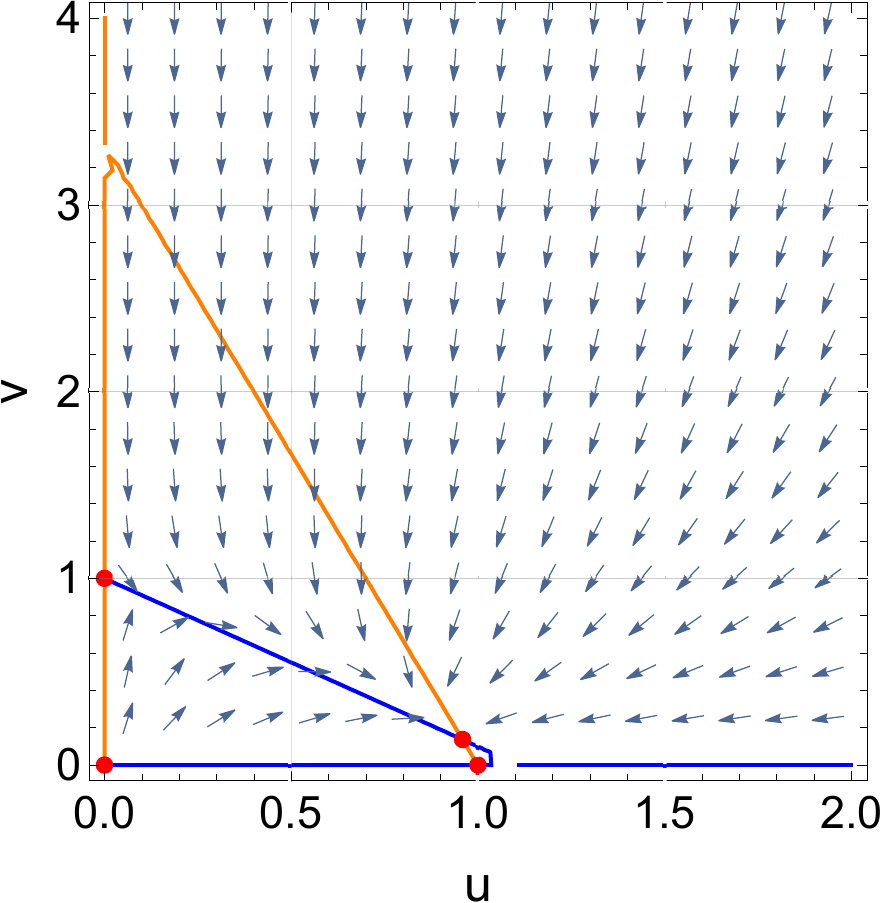}
		\caption{$k=0$}
	\end{subfigure}
	\hfil
	\begin{subfigure}[b]{.475\linewidth}
		\includegraphics[width=\linewidth,height=2in]{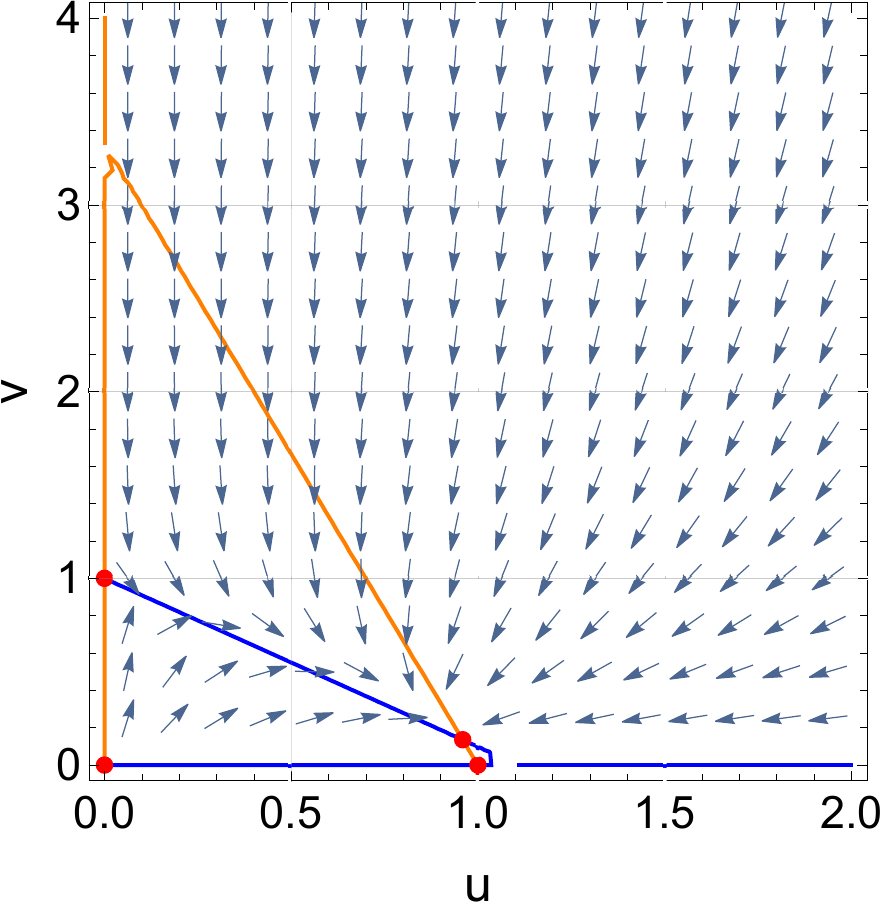}
		\caption{$k=10^{-3}$ }
	\end{subfigure}
	\newline
	\begin{subfigure}[b]{.49\linewidth}
		\includegraphics[width=\linewidth,height=2in]{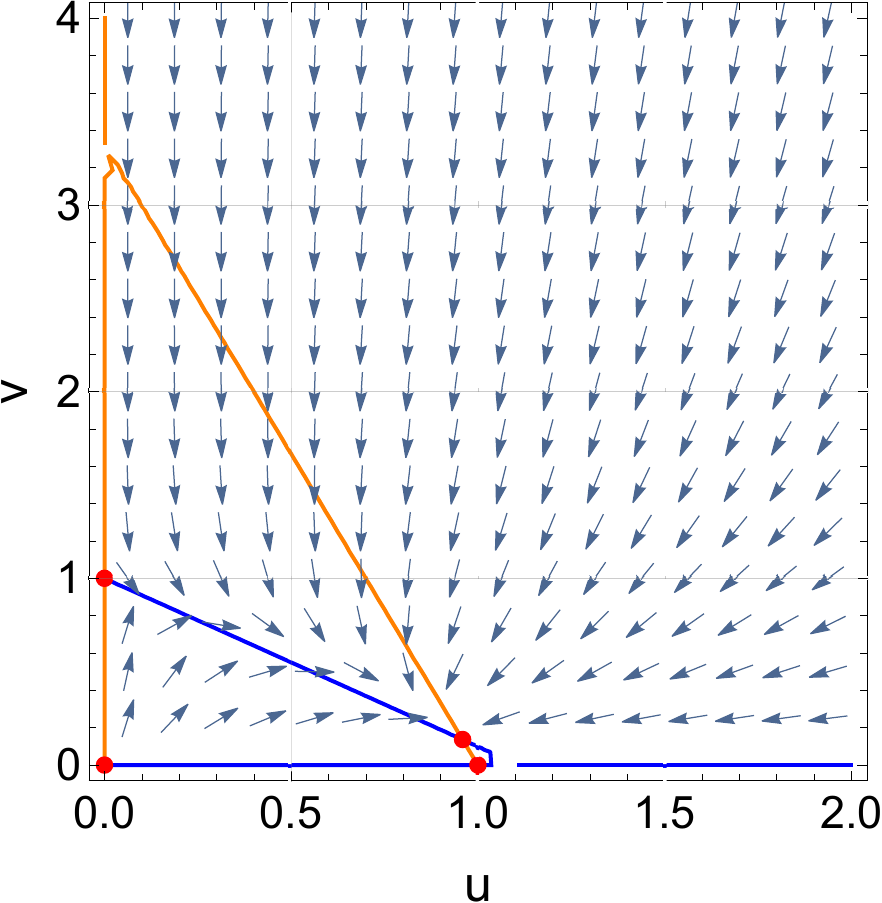}
		\caption{$k=10^{-4}$}
	\end{subfigure}
	\hfil
	\begin{subfigure}[b]{.49\linewidth}
		\includegraphics[width=\linewidth,height=2in]{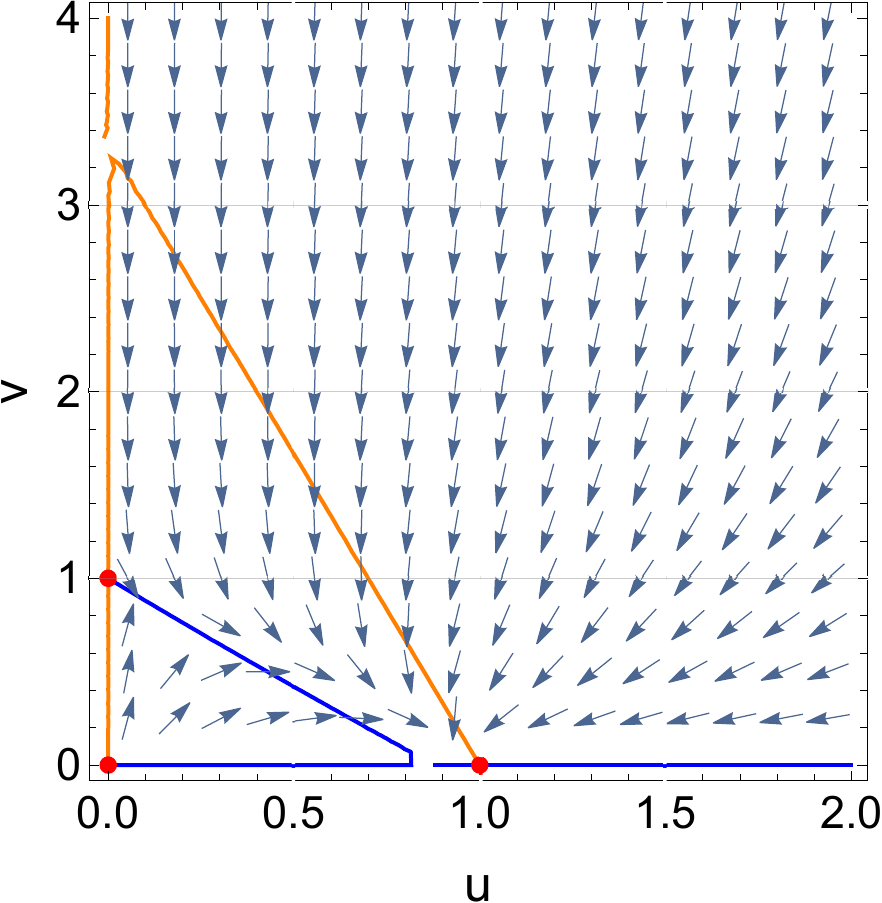}
		\caption{$k=0.3$}
	\end{subfigure}
	\caption{Phase plots showing various dynamics under weak competition parametric restriction with $f=0$. Here $(u^*,v^*)$ is a sink and both $(0,v^*)$ and $(u^*,0)$  are saddles in $(A),(B)$ and $(C)$. In $(D)$ there is competition exclusion and $(u^*,0)$ wins. The orange curve is the u-nullcline and blue curve is the v-nullcline. Parameters used are $a_1=1, a_2=2,b_1=1, b_2=2, c_1=0.3, c_2=1.8$.}
	\label{fig:ode_wc}
\end{figure}


\begin{lemma}\label{ce_strong}
	The interior equilibrium $\widehat{E}_4$ exists and is a saddle if 
	\[ \Big( \dfrac{b_2b_1}{c_1} -c_2\Big) <\frac{k a_2b_1^2}{(b_1+ka_1)^2}.\]
\end{lemma}

\begin{proof}
	In order to claim that the interior equilibrium $\widehat{E}_4$ is a saddle, it is enough to show that $Trace(\widehat{J}^*(\widehat{E}_4))<0$ and $Det(\widehat{J}^*(\widehat{E}_4))<0.$
	We have that
	\[  Trace(\widehat{J}^*(\widehat{E}_4)) = -b_1 u^* -b_2 v^*<0,\]
	
	and
	\[ Det(\widehat{J}^*(\widehat{E}_4) = b_1b_2 u^* v^* - c_1u^* \Big( \dfrac{ka_2v^*}{(1+ku^*)^2} +c_2 v^* \Big) = u^* v^* \Big\{ b_1b_2 - c_1 \Big( \dfrac{k a_2}{(1+ku^*)^2} + c_2\Big) \Big\}.  \]

	Under the assumption and density of reals, we have

	\begin{align*}
		\Big( \dfrac{b_2b_1}{c_1} -c_2\Big) &<\frac{k a_2}{(1+\frac{ka_1}{b_1})^2}  <\frac{k a_2}{(1+ku^*)^2} \implies Det(\widehat{J}^*(\widehat{E}_4)<0.
	\end{align*}
	
	Hence, $\widehat{E}_4$ is a saddle.
\end{proof}

\begin{figure}[h]
	\begin{subfigure}[b]{.475\linewidth}
		\includegraphics[width=\linewidth,height=2in]{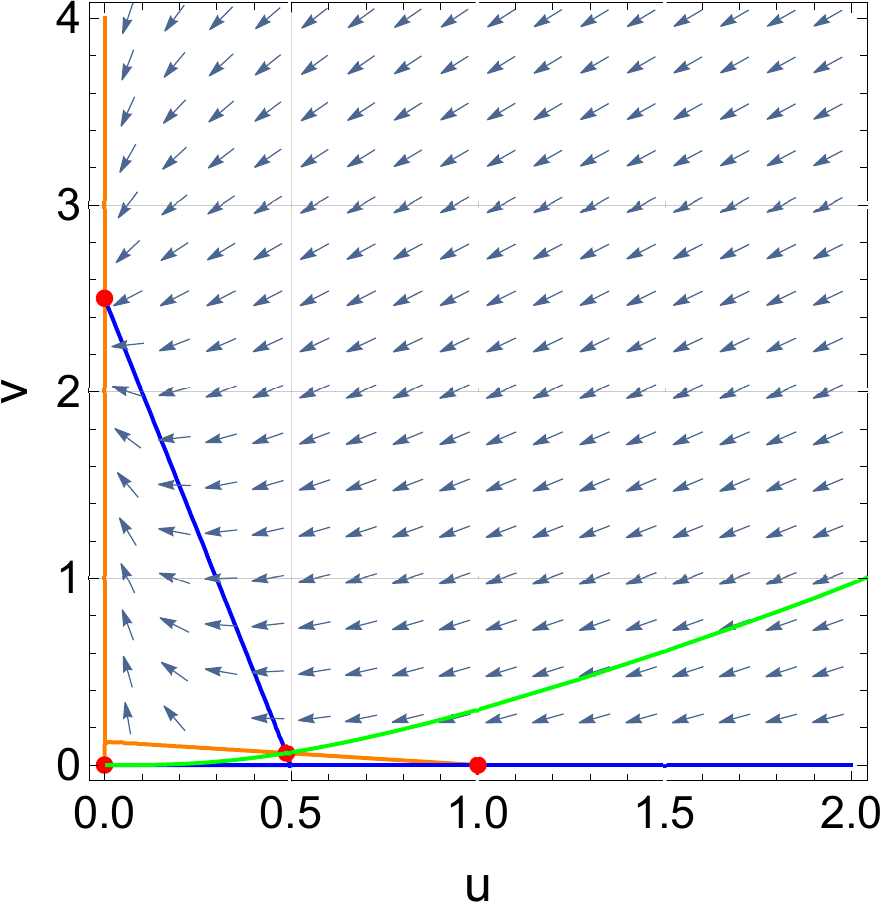}
		\caption{$k=0$ }
	\end{subfigure}
	\hfill
	\begin{subfigure}[b]{.49\linewidth}
		\includegraphics[width=\linewidth,height=2in]{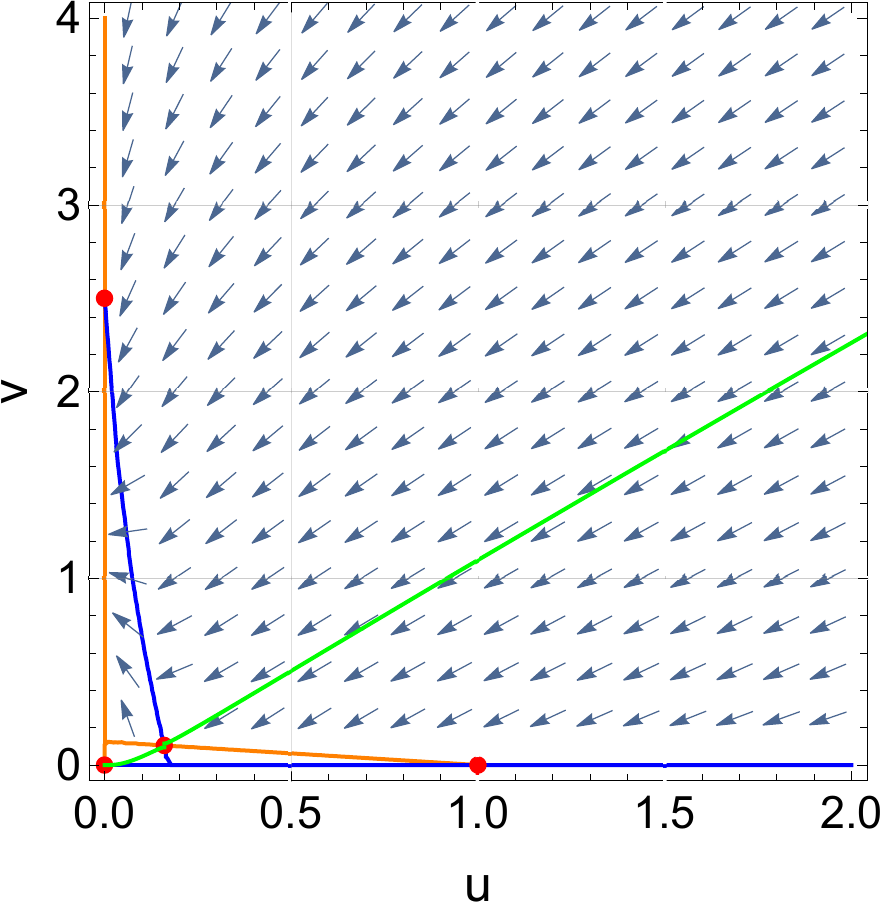}
		\caption{$k=11$}
	\end{subfigure}
	\caption{Phase diagrams showing dynamics under strong competition parametric restriction with $f=0$. The orange curve is the u-nullcline, blue curve is the v-nullcline and separatrix/stable manifold is in green. Here, $(u^*,v^*)$ is saddle. Parameters used are $a_1=0.5, a_2=2,b_1=0.5, b_2=0.8, c_1=4, c_2=4$.}
	\label{fig:ode_st_comp}
\end{figure}

%

\begin{lemma}\label{lem:two_post_ode}
	Consider the system \eqref{eq:ODE2}. For any given $k>0$ s.t.,
		 there exist two positive interior equilibria, a sink and a saddle.
\end{lemma}

\begin{proof}
	Theorem \ref{thm:exist} provides conditions under which two positive interior equilibria exists. Since stability has to alternate by standard theory \cite{perko2013differential} for planar systems, one of the equilibrium is stable while the other is unstable. Instability as a source is impossible due to the lack of periodic dynamics in the system via Lemma \ref{lem:dc1}. Thus the unstable equilibrium must be a saddle. This proves the lemma.
\end{proof}

\begin{lemma}\label{lem:dc1}
	Consider the ODE system \eqref{eq:ODE2}. There do not exist any periodic orbits for the system, for any values of the fear parameter $k$.
\end{lemma}

\begin{proof}
	Consider the function $\phi(u,v) = \frac{1}{u v}$ where $u$ and $v$ are both non-zero. Let,
	\begin{equation*}\label{eq:F}
	\begin{split}
		F_1(u,v) &= a_1 u -b_1 u^2 -c_1 uv, \\
		F_2(u,v)&=  \dfrac{a_2 v}{1+ku}  -b_2 v^2 -c_2 uv.
	\end{split}
\end{equation*}
		
	 Then we have
	\begin{eqnarray}
	\dfrac{\partial (F_1 \phi)}{\partial u} + \dfrac{\partial (F_2 \phi)}{\partial v}  
	&=& \dfrac{\partial}{\partial u} \left( \dfrac{1}{u v} (a_1 u -b_1 u^2 -c_1 uv)  \right) + \dfrac{\partial}{\partial v} \left( \dfrac{1}{u v} \left(\dfrac{a_2 v}{1+ku}  -b_2 v^2 -c_2 uv\right)  \right),  \nonumber \\
	&=& - \frac{b_{1}}{v} - \frac{b_{2}}{u} < 0.\nonumber \\
	\end{eqnarray}
	
	The result follows by application of the Dulac criterion \cite{perko2013differential}.
\end{proof}

\begin{figure}[h]
	\begin{subfigure}[b]{.475\linewidth}\label{fig:ode_two_postive_base}
		\includegraphics[width=\linewidth,height=2in]{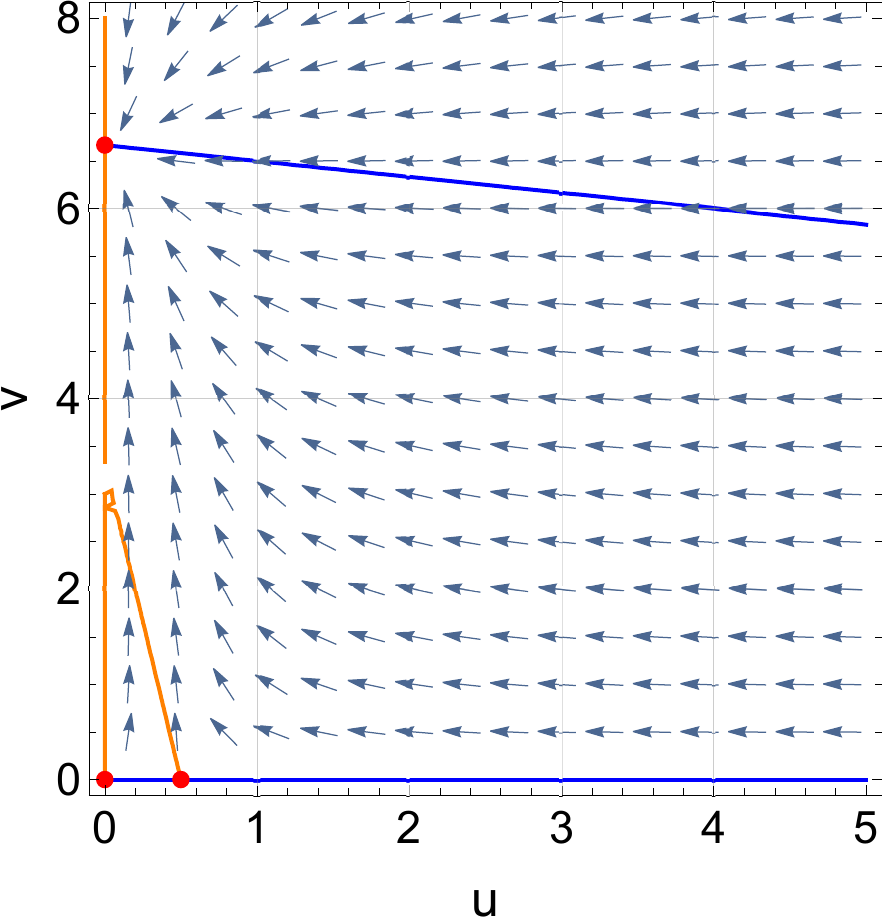}
		\caption{$k=0$ }
	\end{subfigure}
	\hfill
	\begin{subfigure}[b]{.475\linewidth}
		\includegraphics[width=\linewidth,height=2in]{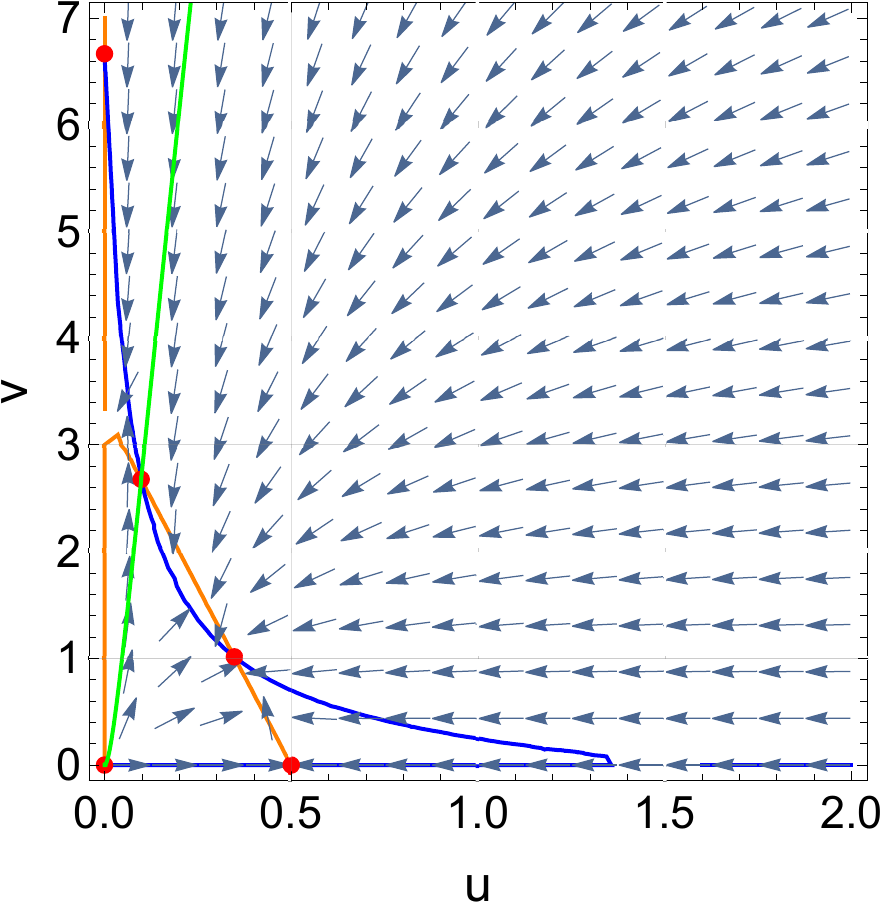}
		\caption{$k=15$ }
	\end{subfigure}
	\caption{Phase plot showing competition exclusion in $(A)$ when fear is absent in both competing species where $(0,v^*)$ wins. In $(B)$, we observe the occurrence of  two positive interior equilibria when $f=0$ and  $k=15$. The orange curve is the u-nullcline, blue curve is the v-nullcline and separatrix/stable manifold is in green. Parameters used are $a_1=1, a_2=2,b_1=2, b_2=0.3, c_1=0.3, c_2=0.05$.}
	\label{fig:ode_two_postive}
\end{figure}

\begin{remark}
	Some observations:
\begin{enumerate}
		\item From Fig.~$\ref{fig:ode_two_postive}$, we observe that when there is no fear, then $(0,v^*)$ is globally stable. For a sufficiently large level of  fear in species $v$, a bi-stability situation is created. That is, for a certain initial data, species $u$ is completely excluded by $v$ and initial data is attracted to the co-existence state (See Fig. ~\ref{fig:ode_two_postive}). 
	
		\item If $(0,v^*)$ is globally attracting,  a much higher level of fear in species $v$ ($\approx 200$) can change the dynamics to a strong competition type case. However most data in this setting is attracted to $(u^{*},0)$. For initial data $(u_{0},v_{0})$ to be attracted to 
	$(0,v^*)$, we would need $v_{0} >>1, u_{0} <<1$.	
\end{enumerate}
\end{remark}

\subsection{The case of $u$ fearing $v$}

In this subsection, we shall consider the case when the competitor $u$ is being fearful of $v$. For the modeling construct, we will follow the modeling approach of  fear effect as in the model  in \cite{Wang16}, where the growth rate of the fearful competitor $u$, is not constant but rather density dependent. Essentially, the growth rate is decreased by a factor $\approx \frac{1}{1+f v}$, where $f \geq 0$ is a fear coefficient. Thus a higher density of the competitor $v$ increases the fear in $u$. When $f=0$, the assumption is there is no fear, and one recovers the classical model \eqref{eq:GeneralEquation}. If fear is present, we obtain the following ODE model for two competing species $u$ and $v$, where $u$ is fearful of $v$.

\begin{align}\label{eq:ODE3}
	\begin{split}
		\dfrac{du}{dt}&= \dfrac{a_1 u}{1+fv}  -b_1 u^2 -c_1 uv,\\
	\dfrac{dv}{dt} &= a_2v -b_2 v^2 -c_2 uv.
	\end{split}
\end{align}

The system \eqref{eq:ODE3} possesses the following biologically feasible non-negative equilibria. These are 
\begin{itemize}
	\item $E_1=(0,0)$,
	\item $E_2=(\frac{a_1}{b_1},0 )$,
	\item $E_3=(0,\frac{a_2}{b_2} )$,
	\item $E_4=(u^*,v^*)$,
\end{itemize}
where $u^*$ is given by $(\ref{u_star})$ and $v^*$ is a positive root of the following third order polynomial in  Eq.$(\ref{three_post1})$.
\subsubsection{Existence}
The system \eqref{eq:ODE3} can have one or two positive interior equilibria, but not three. This is established via the following lemma,

\begin{lemma}
\label{lem:n3}
Consider the ODE system \eqref{eq:ODE3}, with $f=0$. If we are in the competitive exclusion, weak-competition, or strong-competition setting, then for any $f>0$, under the parametric restrictions in  $(\ref{three_post})$, it is impossible to find  three positive interior equilibria. 
\end{lemma}

\begin{proof}

If $k=0$ and $f>0$, Eq.(\ref{eq: vpoly}) reduces to the cubic equation $B(v^*)^3+C(v^*)^2+D(v^*)+E=0$, where 
\begin{equation}\label{three_post1}
	\begin{split}
		B&=b_1f^2(c_1c_2-b_1b_2),\\
		C&=a_2b_1^2f^2+2b_1f(c_1c_2-b_2b_1),\\
		D&=b_1f(-a_1c_2+2a_2b_1)+b_1(c_1c_2-b_2b_1),\\
		E&=b_1(-a_1c_2+a_2b_1).
	\end{split}
\end{equation}
We will require the following conditions to satisfy Descartes's rule of signs, so as to obtain three positive roots to the cubic equation. These are:

\begin{align}\label{three_post}
	\begin{split}
		c_1c_2-b_1b_2 &<0, \\
		a_2b_1f+2(c_1c_2-b_2b_1)& >0,\\
		f(-a_1c_2+2a_2b_1) & <-(c_1c_2-b_2b_1),\\
		-a_1c_2+a_2b_1 &>0.
	\end{split}
\end{align}
Let us prove this is an impossible claim by contradiction:
	First, assume the parameter set satisfies both competitive exclusion-state parametric restriction and $(\ref{three_post})$. We know that competitive exclusion-state is asymptotically stable if $\frac{a_1}{a_2} > \max \Big\{ \frac{b_1}{c_2},\frac{c_1}{b_2}  \Big \}$. On using these parametric restriction, $-a_1c_2 + a_2b_1<0$, which is a contradiction to the last inequality in $(\ref{three_post})$.
	
	If the parameter set satisfies both the strong competition state parametric restriction and $(\ref{three_post})$, then we have a contradiction because of the first inequality in $(\ref{three_post})$, as under the strong competition state parametric restriction that inequality should be positive.
	
	For the weak competition, recall the parametric restrictions:
	\[\frac{c_2}{b_1}<\frac{a_2}{a_1}<\frac{b_2}{c_1}.\]
	Let's re-write the third inequality in $(\ref{three_post})$,
	\[ -(c_1c_2-b_2b_1)-f(-a_1c_2+2a_2b_1)>0.\]
	On adding the second and third inequality in $(\ref{three_post})$, we have
	\[(c_1c_2-b_1b_2) + f(a_1c_2-a_2b_1)>0, \]
	which is a contradiction as the added inequality should be negative by the parametric restrictions of weak-competition. This proves the lemma.
	
\end{proof}

\subsubsection{Linear Stability Analysis}
The Jacobian matrix of system \eqref{eq:ODE3} is given by
\begin{equation}
	J^*=
	\left(
	\begin{array}{cc}
		\dfrac{a_1}{f v^*+1}-2 b_1 u^*-c_1 v^* & -\dfrac{a_1 f u^*}{(f v^*+1)^2}-c_1 u^* \\
		-c_2 v^* & a_2-2 b_2 v^*-c_2 u^* \\
	\end{array}
	\right).
\end{equation}

\begin{lemma}
	The trivial steady state $E_1$ is locally unstable.
\end{lemma}

\begin{proof}
	For proof details refer to \eqref{3eq1}
\end{proof}

\begin{figure}[h]
	\begin{subfigure}[b]{.475\linewidth}
		\includegraphics[width=\linewidth,height=2in]{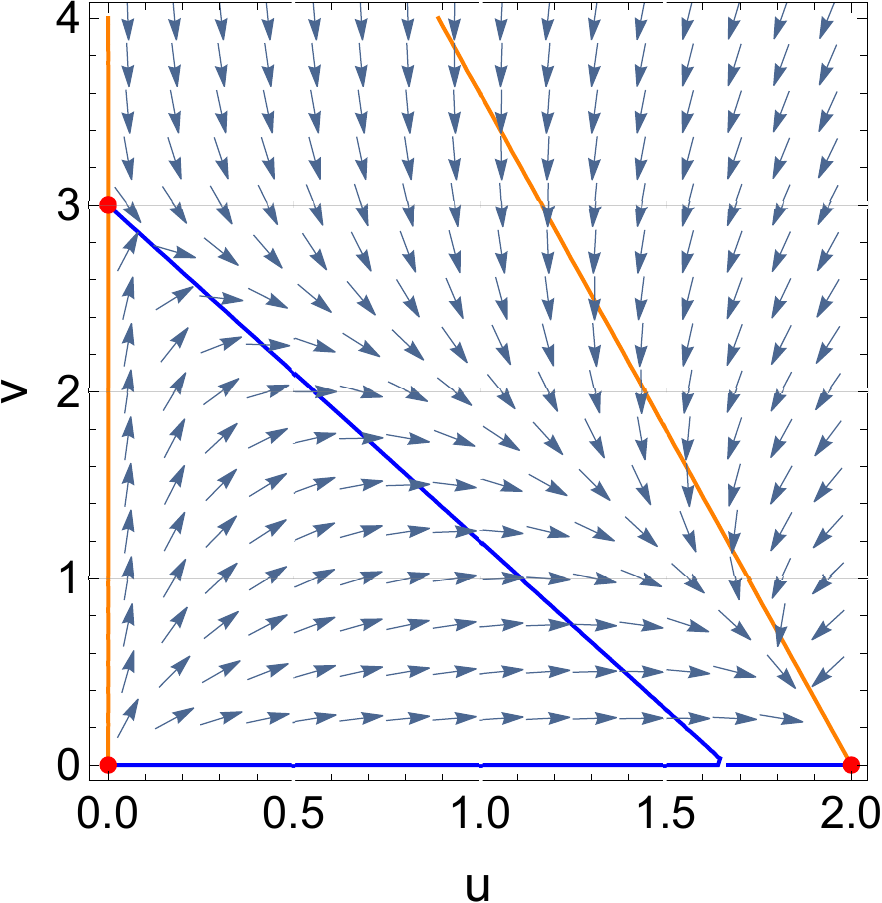}
		\caption{$f=0$ }
	\end{subfigure}
	\hfill
	\begin{subfigure}[b]{.49\linewidth}
		\includegraphics[width=\linewidth,height=2in]{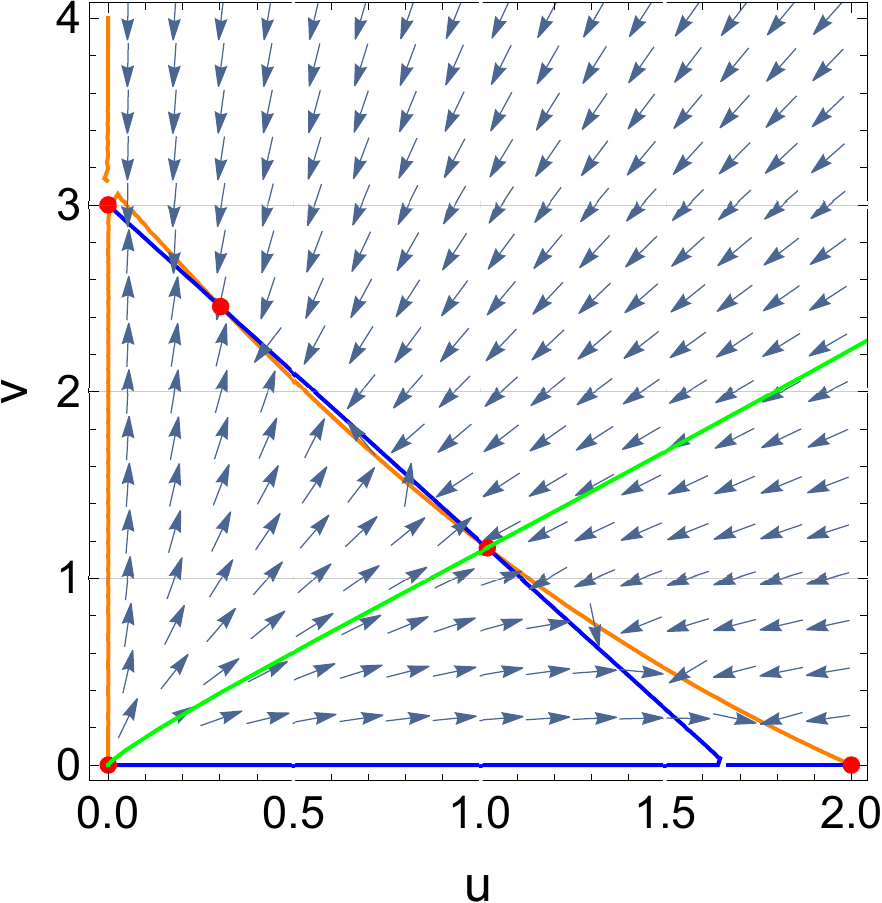}
		\caption{$f=0.42$}
	\end{subfigure}
	\newline
	\begin{subfigure}[b]{.475\linewidth}
		\includegraphics[width=\linewidth,height=2in]{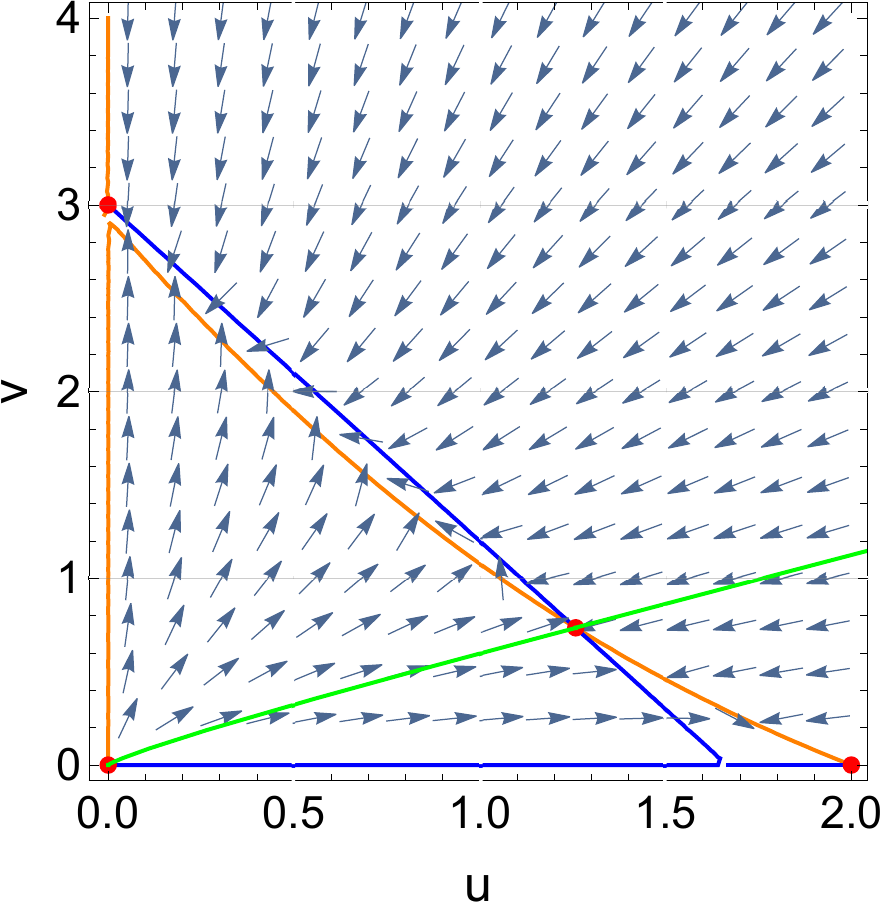}
	\end{subfigure}
	\caption{Phase plots showing dynamics under competition exclusion parametric restriction with  $k=0$ in each plot. In $(A)$, $(0,v^*)$ is saddle point. In $(B)$ we observe multiple coexistence points and $(0,v^*)$ and $(u^*,0)$ are saddle points. In $(C)$ a further increase in $f$ still leads to bi-stability but multiple coexistence states are lost. The parameters $a_1=3.6, a_2=3,b_1=1.8, b_2=1, c_1=0.5, c_2=1.8$ are used for each simulation. In each plot, the orange curve is the u-nullcline, blue curve is the v-nullcline and the green curve is the separatrix/stable manifold.}
	\label{fig:Fig1}
\end{figure}

\begin{lemma}
	The boundary equilibrium point $E_2$ is locally stable iff $a_2b_1<c_2a_1$.
\end{lemma}
\begin{proof}
	For proof details refer to \eqref{3eq2}
\end{proof}

\begin{lemma}
	The boundary equilibrium point $E_3$ is locally stable iff $f>\dfrac{a_1 b_2^2 - a_2 b_2 c_1}{a_2^2 c_1}$.
\end{lemma}

\begin{proof}
	For proof details refer to \eqref{3eq3}
\end{proof}

\begin{lemma}\label{lem:coexi_ode2}
	The interior equilibrium $E_4$ exists and is locally stable  if  
	\[ f < \frac{1}{a_1} \Big( \frac{b_1b_2}{c_2} - c_1 \Big).\]
\end{lemma}

\begin{proof}
	For proof details refer to \eqref{3eq4}
\end{proof}

\begin{lemma}\label{lem:strong_ode2}
	The interior equilibrium $E_4$ exists and is a saddle  if  
	\[\Big( \dfrac{b_1 b_2}{c_2} -c_1 \Big)< \dfrac{a_1 c_2^2 f}{(f a_2 +c_2)^2}.\]
\end{lemma}

\begin{proof}
	For proof details refer to \eqref{3eq5}
\end{proof}

\begin{lemma}\label{lem:dc2}
	Consider the ODE system \eqref{eq:ODE3}. There do not exist any periodic orbits for the system, for any values of the fear parameter $f$.
\end{lemma}

\begin{proof}
	The proof follows as in the proof of Lemma \ref{lem:dc1}.
\end{proof}

\begin{figure}[h]
	\begin{subfigure}[b]{.475\linewidth}
		\includegraphics[width=\linewidth,height=2in]{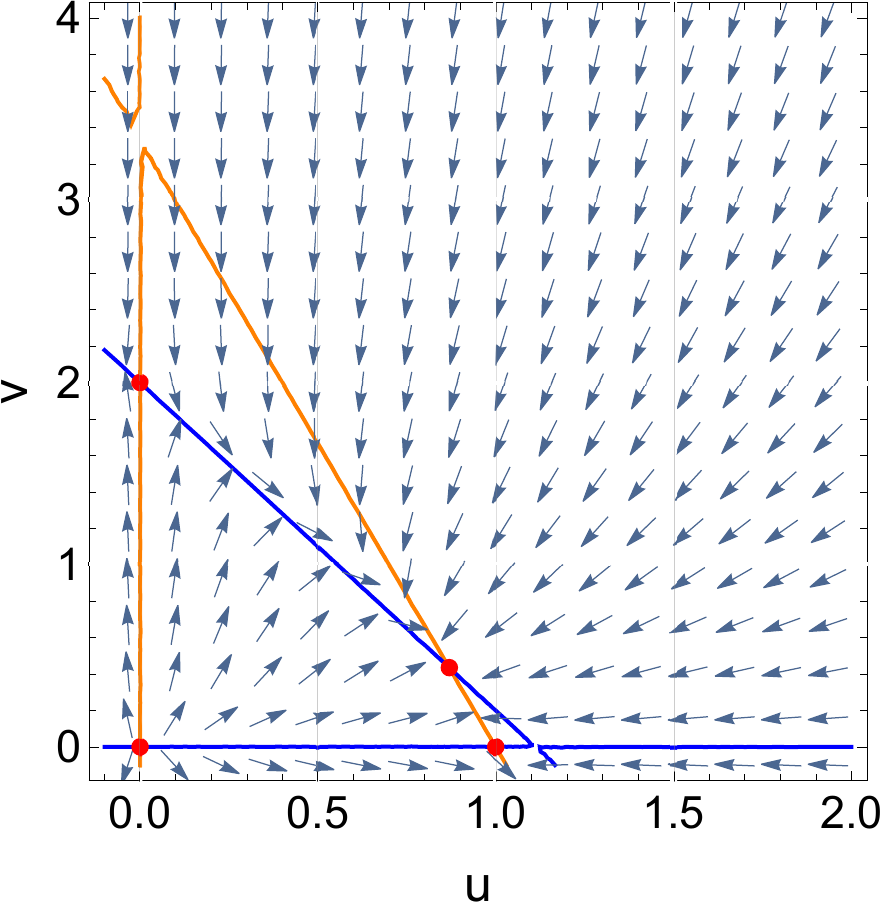}
		\caption{$f=0$ }
	\end{subfigure}
	\hfill
	\begin{subfigure}[b]{.49\linewidth}
		\includegraphics[width=\linewidth,height=2in]{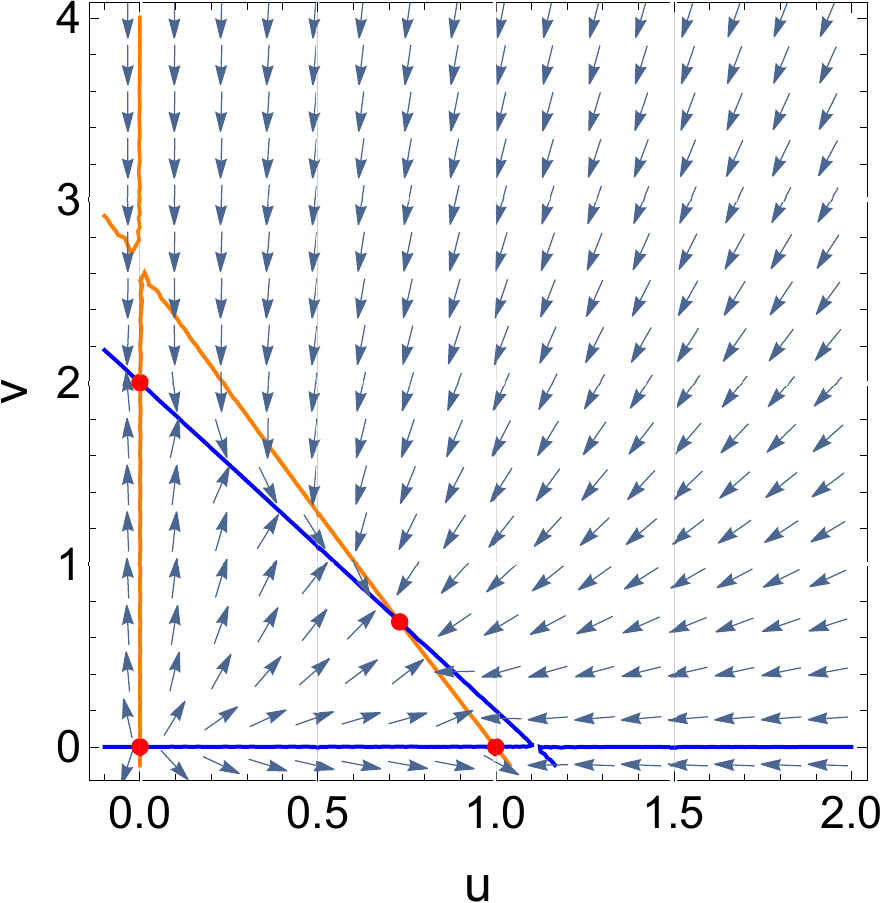}
		\caption{$f=0.1$}
	\end{subfigure}
	\newline
	\begin{subfigure}[b]{.475\linewidth}
		\includegraphics[width=\linewidth,height=2in]{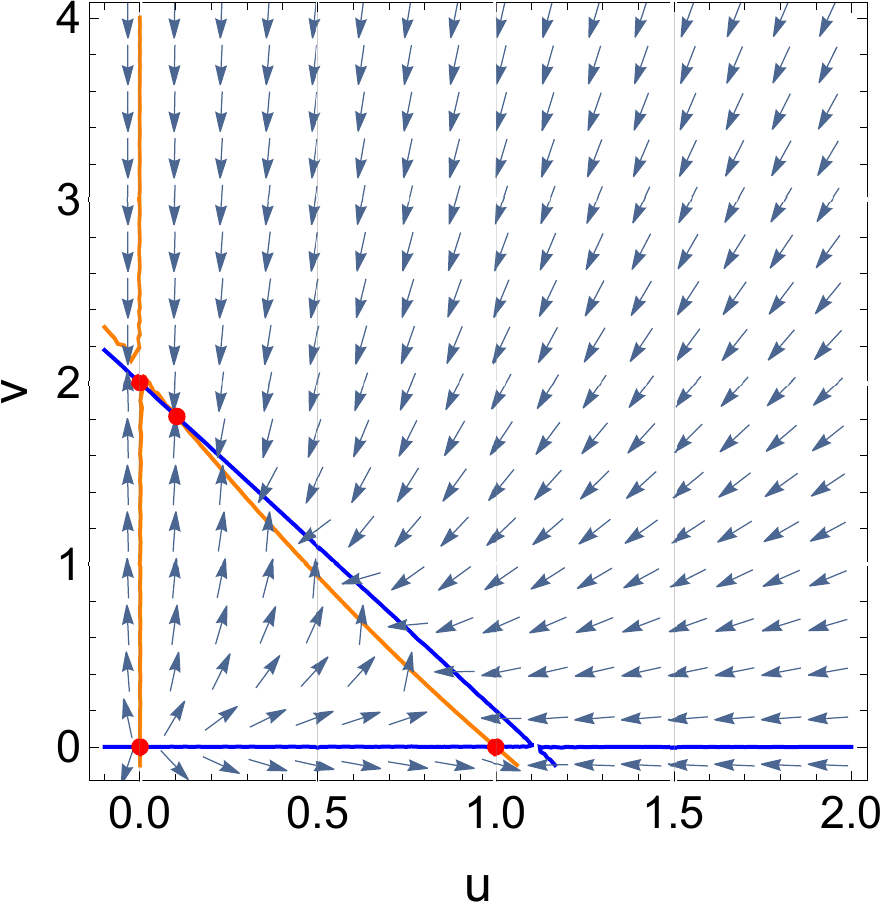}
		\caption{$f=0.3$ }
	\end{subfigure}
	\hfill
	\begin{subfigure}[b]{.49\linewidth}
		\includegraphics[width=\linewidth,height=2in]{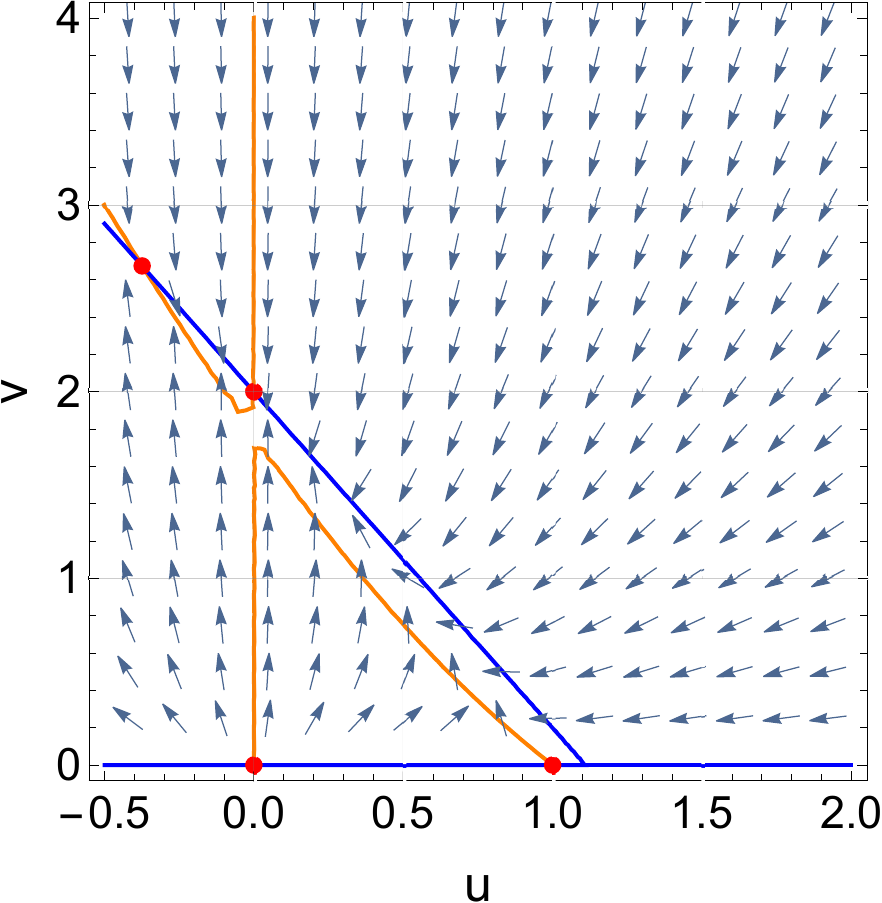}
		\caption{$f=0.5$}
	\end{subfigure}
	\caption{ Dynamics under weak competition parametric restriction with $k=0$ in each plot. Here $(u^*,v^*)$ is a sink and both $(0,v^*)$ and $(u^*,0)$  are saddles in $(A), (B)$ and $(C)$. $(0,0)$ is a source.  The orange curve is the u-nullcline and blue curve is the v-nullcline. In $(D)$,  $(0,v^*)$ changes from a saddle to a sink. Parameters used are $a_1=1, a_2=2,b_1=1, b_2=1, c_1=0.3, c_2=1.8$.}
	\label{fig:Fig2}
\end{figure}

\begin{figure}[h]
	\begin{subfigure}[b]{.475\linewidth}
		\includegraphics[width=\linewidth,height=2in]{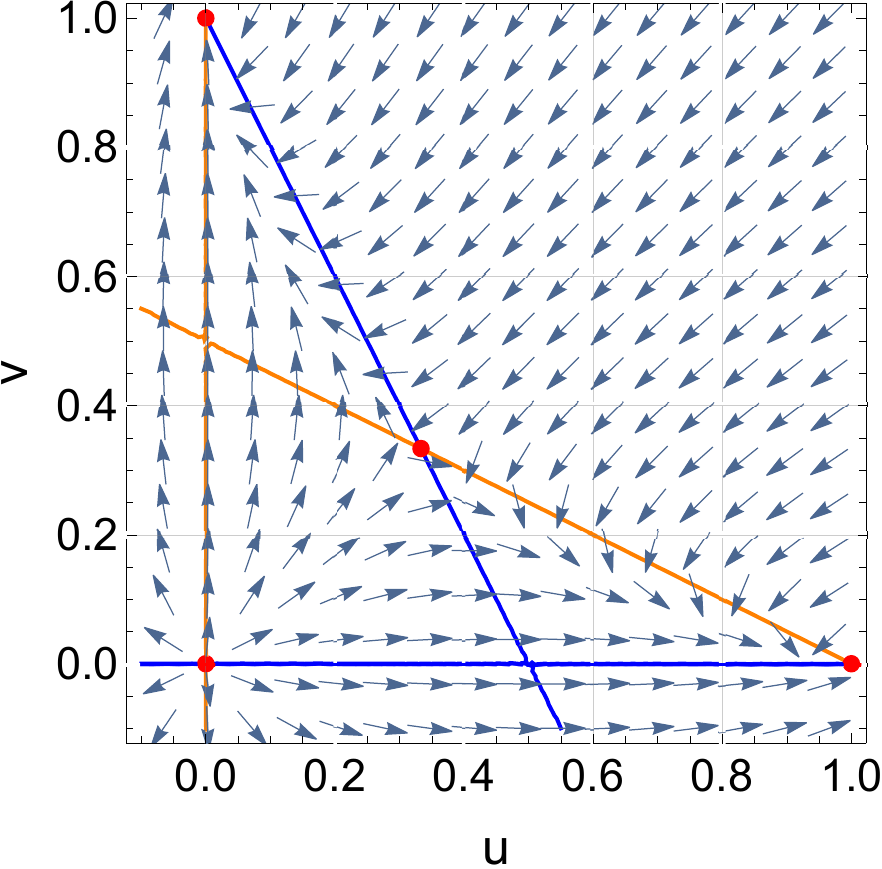}
		\caption{$ f=0$ }
	\end{subfigure}
	\hfill
	\begin{subfigure}[b]{.49\linewidth}
		\includegraphics[width=\linewidth,height=2in]{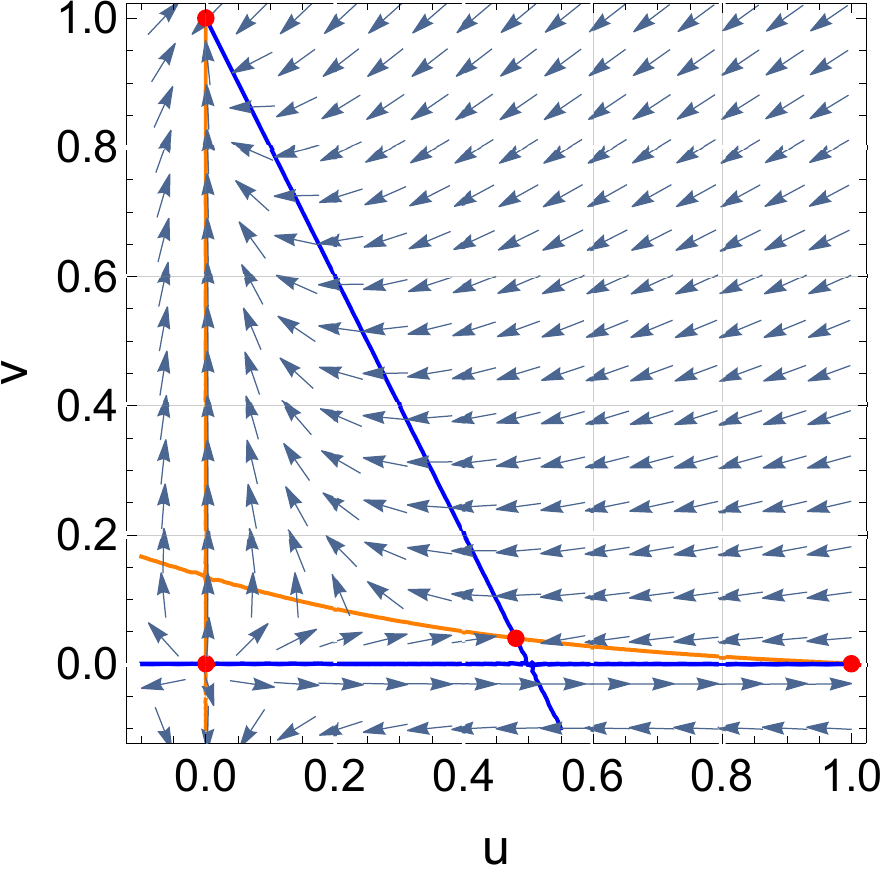}
		\caption{$f=20$}
	\end{subfigure}
	\caption{Phase plots showing dynamics under strong competition parametric restriction with $k=0$. Parameters used are $a_1=1, a_2=1,b_1=1, b_2=1, c_1=2, c_2=2$.}
	\label{Fig3}
\end{figure}

\subsection{Bifurcation Analysis}
A bifurcation is said to occur in a dynamical system when the behavior of solutions changes when a parameter is varied. Bifurcation analysis is useful in understanding and measuring these qualitative changes as the system switches from stable to unstable and vice-versa. 

\subsubsection{Saddle-node bifurcation}

The following theorem is connected to the existence of a saddle-node bifurcation for the growth rate $a_1$  when there is no fear effect in competitor $u$.
\begin{theorem}\label{thm:sad}
	The ODE system \eqref{eq:ODE2} undergoes a saddle-node bifurcation around $E_4^*$ at $a_1=a_1^*$ when the conditions $det(J^{*})=0$ and $tr(J^{*})<0$ hold for the system parameters.
\end{theorem}

\begin{proof}
	We shall use the Sotomayor's Theorem \cite{perko2013differential}  to show the occurrence of a saddle-node bifurcation at $a_1=a_1^*$. At $a_1=a_1^*$, we can have $det(J^{*})=0$ and $tr(J^{*})<0$. This implies that, $det(J^{*})$ admits a zero eigenvalue. Now let $P=(p_1,p_2)^T$ and $Q=(q_1,q_2)^T$ be the eigenvectors of $J^{*}$ and $J^{*T}$ corresponding to the zero eigenvalue respectively.
	
	We have that, $P=\left(\dfrac{-B}{A},1 \right)^T$ and $Q=\left(\dfrac{-A}{F},1 \right)^T$ \\where 
	$A=a_1-2b_1 u^*-c_1 v^*$, $B=-c_1u^*$, $F=-\dfrac{a_2 k v^*}{(1+ku^*)^2}-c_2 v^*$.
	
	Now, let $X=(X_1,X_2)^T$ where 
	\begin{equation*}
		\begin{split}
			X_1&= a_1u-b_1 u^2-c_1uv,\\
			X_2&= \dfrac{a_2 v}{1+ku}-b_2v^2-c_2 uv.
		\end{split}
	\end{equation*}
	
	Furthermore, $$Q^TX_{a_1}(E_3^*,a_1)=\left(\dfrac{-A}{F},1 \right) \left(u^*,0 \right)^T=-\dfrac{Au^*}{F}\neq 0$$ and 
	\begin{align*}
		Q^T[D^2X(E_3^*,a_1)(P,P)] &=\left(\dfrac{-A}{F},1 \right) \left(\dfrac{2B}{A}\left(b_1+c_1 \right), 2\left(\dfrac{a_2 k B}{A \left(1+ku^* \right)^2} \left[1-\dfrac{kv^*}{1+ku^*} \right ] -b_2 \right)  \right)^T\\
		&\neq 0.
	\end{align*}
	Therefore by the Sotomayor's theorem system \eqref{eq:ODE2} undergoes a saddle-node bifurcation at $a_1=a_1^*$ around $E_4^*$.
\end{proof}

\begin{remark}
	Consider the case of two interior equilibria, such as in Fig. \ref{fig:ode_two_postive}. Then decreasing the fear coefficient $k$ results in a saddle-node bifurcation, where the interior equilibrium (the saddle $E_{4}$ and the node $E_{5}$) collide and disappear, resulting in the boundary equilibrium $E_{2} = (u^{*},0)$ becoming globally asymptotically stable. Thus in this setting, a certain critical level of ``fear" can maintain a co-existence state, but fear less than this level takes the system back to a competitive exclusion type scenario. This transition occurs via a saddle-node bifurcation.
	This can be rigorously proven by adopting the methods of Theorem \ref{thm:sad}, to the parameter $k$ instead of using $a_{1}$.
\end{remark}

\subsubsection{Transcritical bifurcation}

\begin{theorem}
\label{thm:tc}
The ODE system \eqref{eq:ODE3} experiences a transcritical bifurcation around $E_3^*$ at $f=f^*=\dfrac{b_2 \left(a_1 b_2-a_2 c_1\right)}{a_2^2 c_1}$ and when $\left(\dfrac{a_1f}{\left(1+fv^*\right)^2}+c_1-\dfrac{b_1 b_2}{c_2} \right) \neq 0$.
\end{theorem}

\begin{proof}
An evaluation of the Jacobian matrix for system \eqref{eq:ODE3} at $E_3$ with $f^*=\dfrac{b_2 \left(a_1 b_2-a_2 c_1\right)}{a_2^2 c_1}$ yields
\begin{equation}\label{transf}
J_f^*=
\left(
\begin{array}{cc}
     0& 0 \\
 -\dfrac{a_2c_2}{b_2} & -a_2 \\
\end{array}
\right).
\end{equation}
A calculation of the  eigenvalues of the Jacobian matrix in Eq. (\ref{transf}) are $\lambda_1=0$ and $\lambda_2=-a_2$. Next, we let $G=(g_1,g_2)^T$ and $H=(h_1,h_2)^T$ denote the eigenvectors corresponding  to the zero eigenvalue of the matrices $J_f^{*}$ and $J_f^{*T}$ respectively.

We have $G=\left(-\dfrac{b_2}{c_2},1 \right)^T$ and $H=\left(1,0 \right)^T$.  Now, let $S=(S_1,S_2)^T$ where 
	\begin{equation*}
		\begin{split}
			S_1&= \dfrac{a_1u}{1+fv}-b_1 u^2-c_1uv,\\
			S_2&= a_2 v-b_2v^2-c_2 uv.
		\end{split}
	\end{equation*}
	The next step is to validate the transversality conditions using the Sotomayor's theorem \cite{perko2013differential}. Now,
	$$H^TR_{f^*}(E_3^*,f)=\left(1,0 \right) \left(0,0 \right)^T= 0.$$
	Also,
	\begin{align*}
H^{T}\left[DS_{f}\left(E_3,f^* \right)G\right] &= \left(
\begin{array}{cc}
 1 & 0 \\
\end{array}
\right) 
\left(
\begin{array}{ccc}
 -\dfrac{a_1 v^*}{\left(1+fv^* \right)^2}& 0  \\
  0 & 0 \\
\end{array}
\right)
\left(
\begin{array}{ccc}
 w_1 \\
  w_2 \\
\end{array}
\right) \\
&=\dfrac{a_1b_2v^*}{c_2\left(1+fv^* \right)^2}\neq 0.
\end{align*}

and 
\begin{equation*}
\begin{split}
H^{T}\left[D^2 S\left(E_3,f^* \right)(G,G)\right] &= \left(
\begin{array}{cc}
 1 & 0 \\
\end{array}
\right) 
\left(
\begin{array}{ccc}
 \dfrac{2b_2}{c_2}\left(\dfrac{a_1f}{\left(1+fv^*\right)^2}+c_1-\dfrac{b_1 b_2}{c_2} \right)  \\
  0 \\
\end{array}
\right) \\
& = \dfrac{2b_2}{c_2}\left(\dfrac{a_1f}{\left(1+fv^*\right)^2}+c_1-\dfrac{b_1 b_2}{c_2} \right)  \neq 0.
\end{split}
\end{equation*}

Therefore by the Sotomayor's theorem system \eqref{eq:ODE3} experiences a transcritical bifurcation at some $f=f^*=\dfrac{b_2 \left(a_1 b_2-a_2 c_1\right)}{a_2^2 c_1}$ around $E_3^*$.
\end{proof}

\begin{remark}
	Consider the case of one interior equilibrium, such as in Fig.~\ref{fig:Fig2} - that is, we are in the weak competition case, when there is no fear or $f=0$. Now increasing the fear coefficient $f$, results in a transcritical bifurcation, where the interior equilibrium (the node $E_{4}$) and the boundary saddle equilibrium $E_{3} = (0, v^{*})$ collide, exchange stability, after which $E_{4}$ now moves to the 2$^{nd}$ quadrant, while $E_{3}$
becomes globally asymptotically stable. Thus in this setting, a certain critical level of ``fear" can move the system from a weak competition setting of coexistence to a competitive exclusion type scenario. This transition occurs via a transcritical bifurcation.
	This is rigorously proven in Theorem \ref{thm:tc}.
\end{remark}

\begin{figure}[h]
	\begin{subfigure}[b]{.475\linewidth}
		\includegraphics[width=\linewidth,height=1.8in]{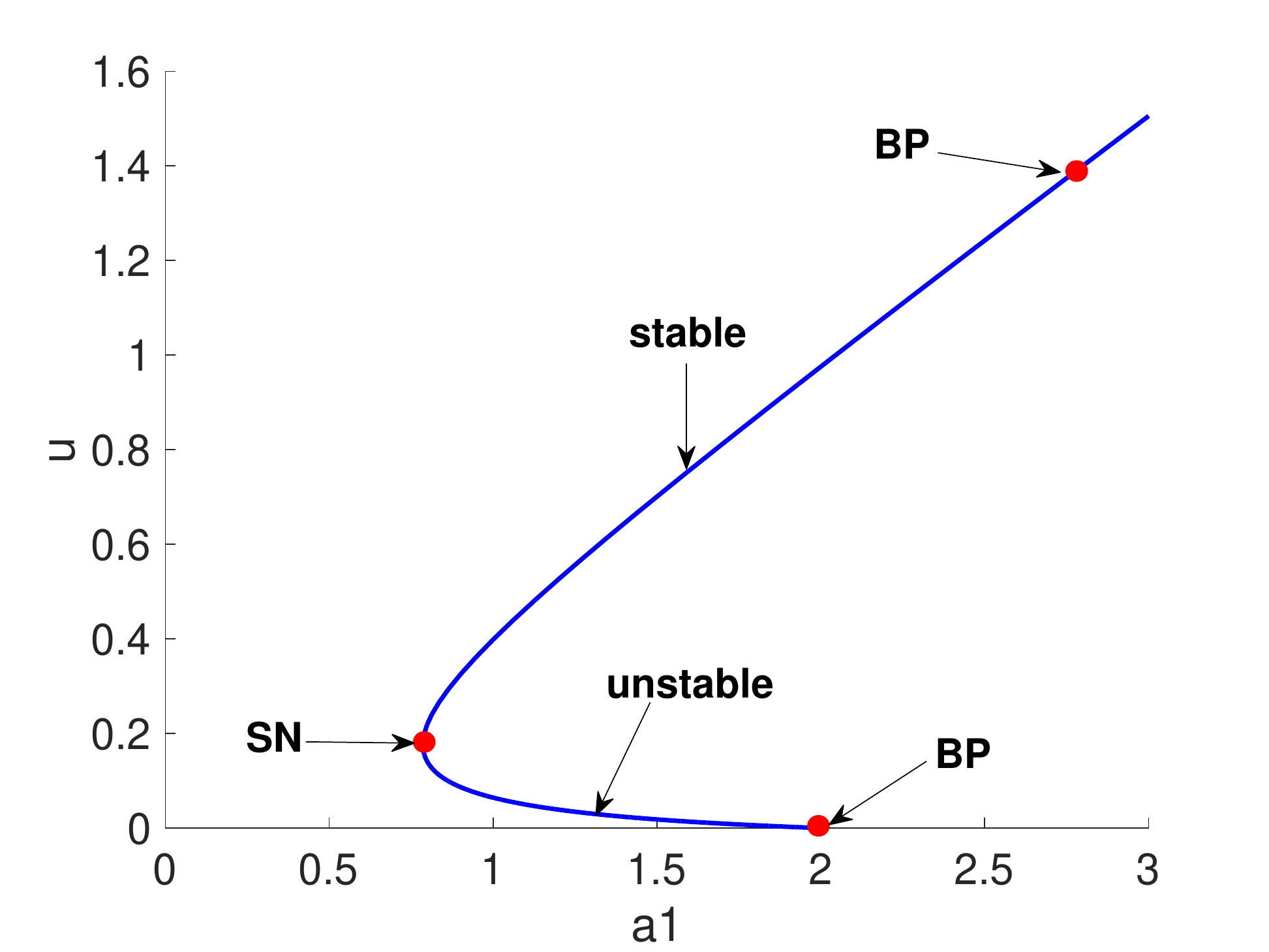}
		\caption{}
	\end{subfigure}
	\hfill
	\begin{subfigure}[b]{.475\linewidth}
		\includegraphics[width=\linewidth,height=1.8in]{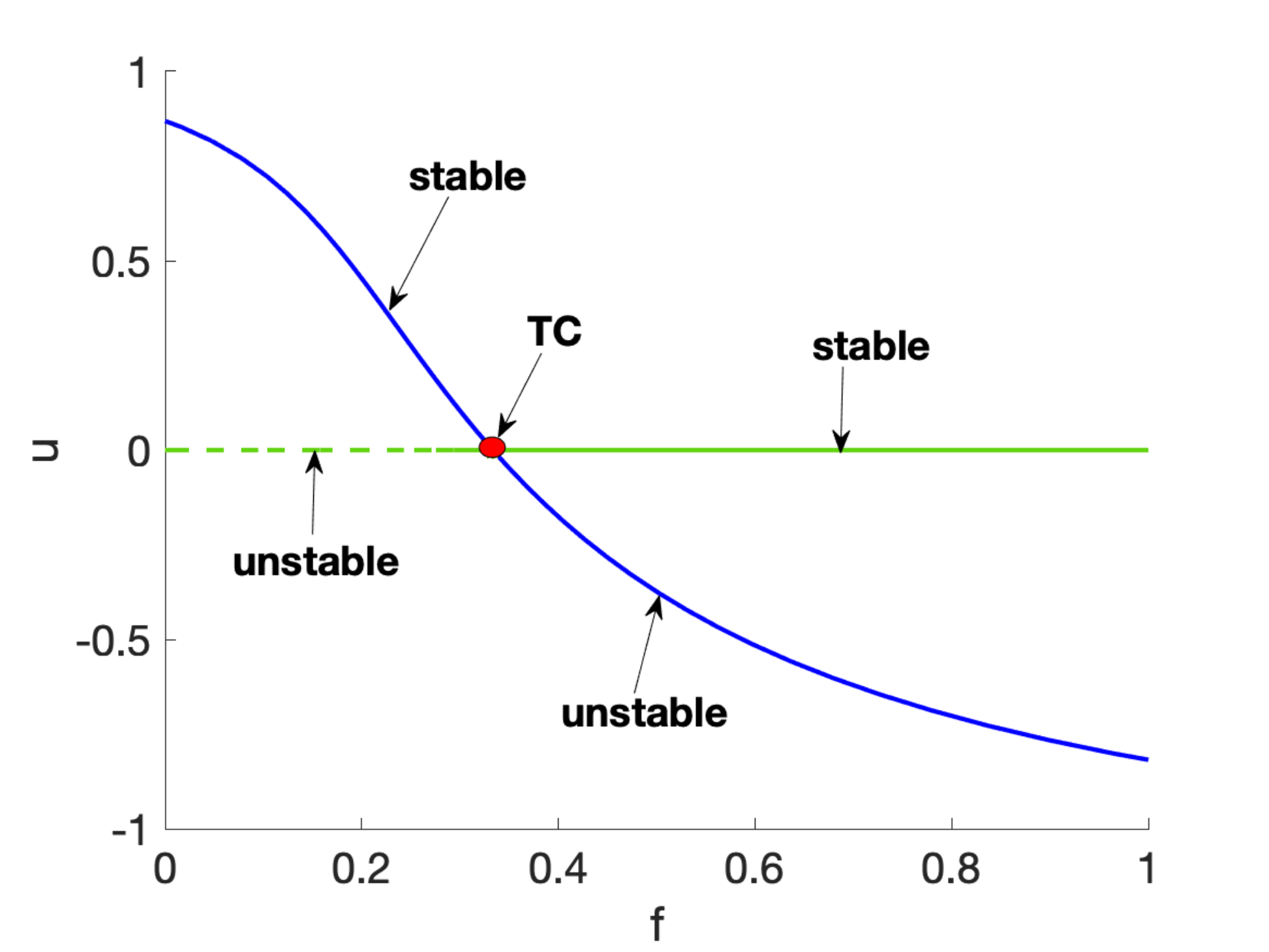}
		\caption{}
	\end{subfigure}
	\caption{Bifurcation diagrams showing the impacts of the intrinsic growth rate $a_1$ and the fear effect parameter $f$. In $(A)$ we observe the occurrence of a saddle-node bifurcation at $a_1=a_1^*=0.785676$. The parameters used are $a_2=2, b_1=2, b_2=0.3, c_1=0.3, c_2=0.05, f=0, k=20$. Initial condition was chosen as $(u_0,v_0)=(0.1,2.5)$. A transcritical bifurcation is also observed in $(B)$ at $f=f^*=0.3333$. The parameters are chosen as $a_1=1,a_2=2,b_1=1,b_2=1$ and $c_1=0.3,c_2=1.8,k=0$. Initial condition was chosen as $(u_0,v_0)=(1,1)$. (Note: TC=Transcritical point, SN=Saddle-Node point, BP=Branch Point.) }
	\label{fig:sad_node}
\end{figure}

\begin{figure}[h]
	\begin{subfigure}[b]{.475\linewidth}
		\includegraphics[width=\linewidth,height=1.8in]{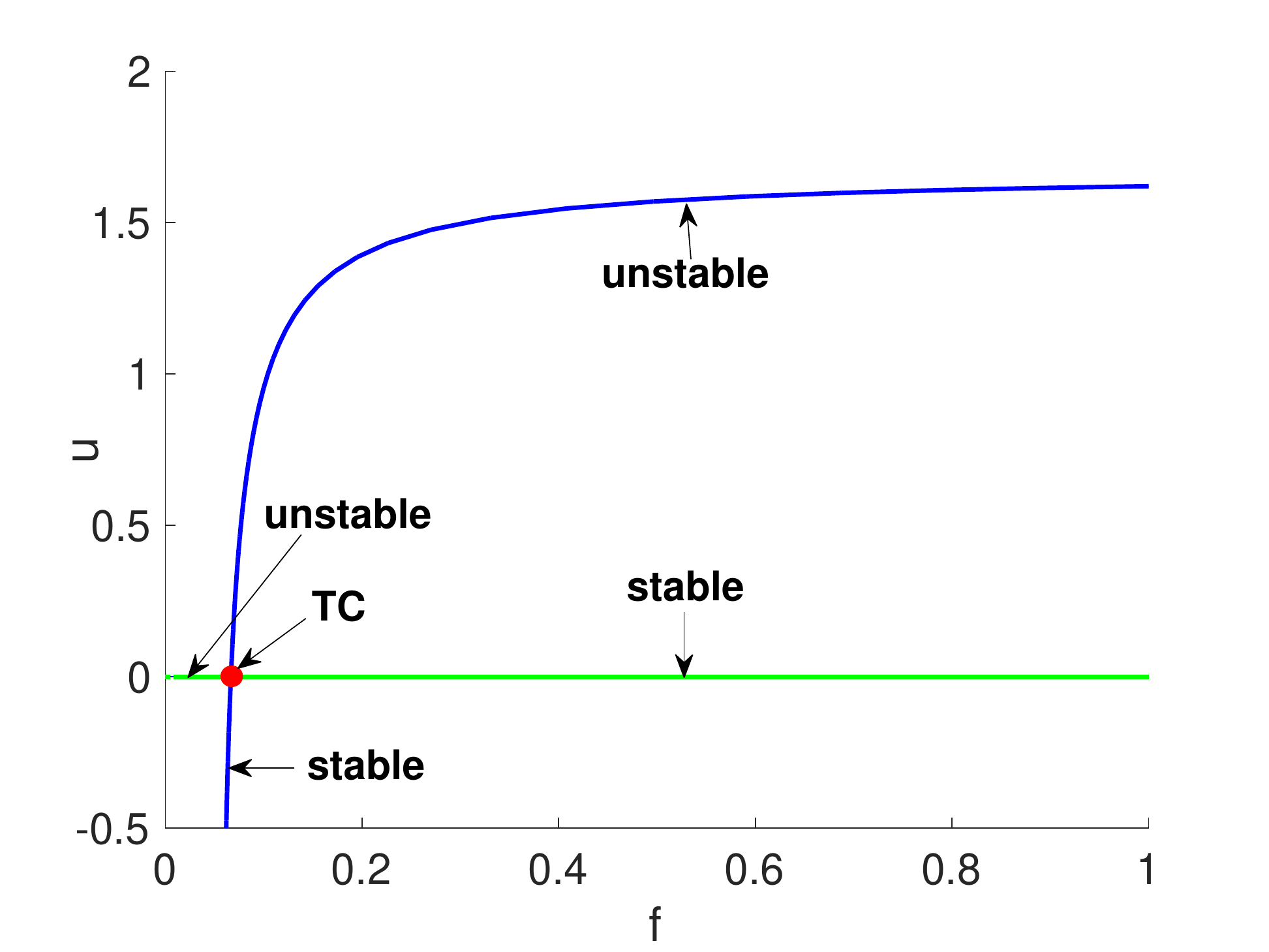}
		\caption{}
		\label{fig:transA}
	\end{subfigure}
	\hfill
	\begin{subfigure}[b]{.475\linewidth}
		\includegraphics[width=\linewidth,height=1.8in]{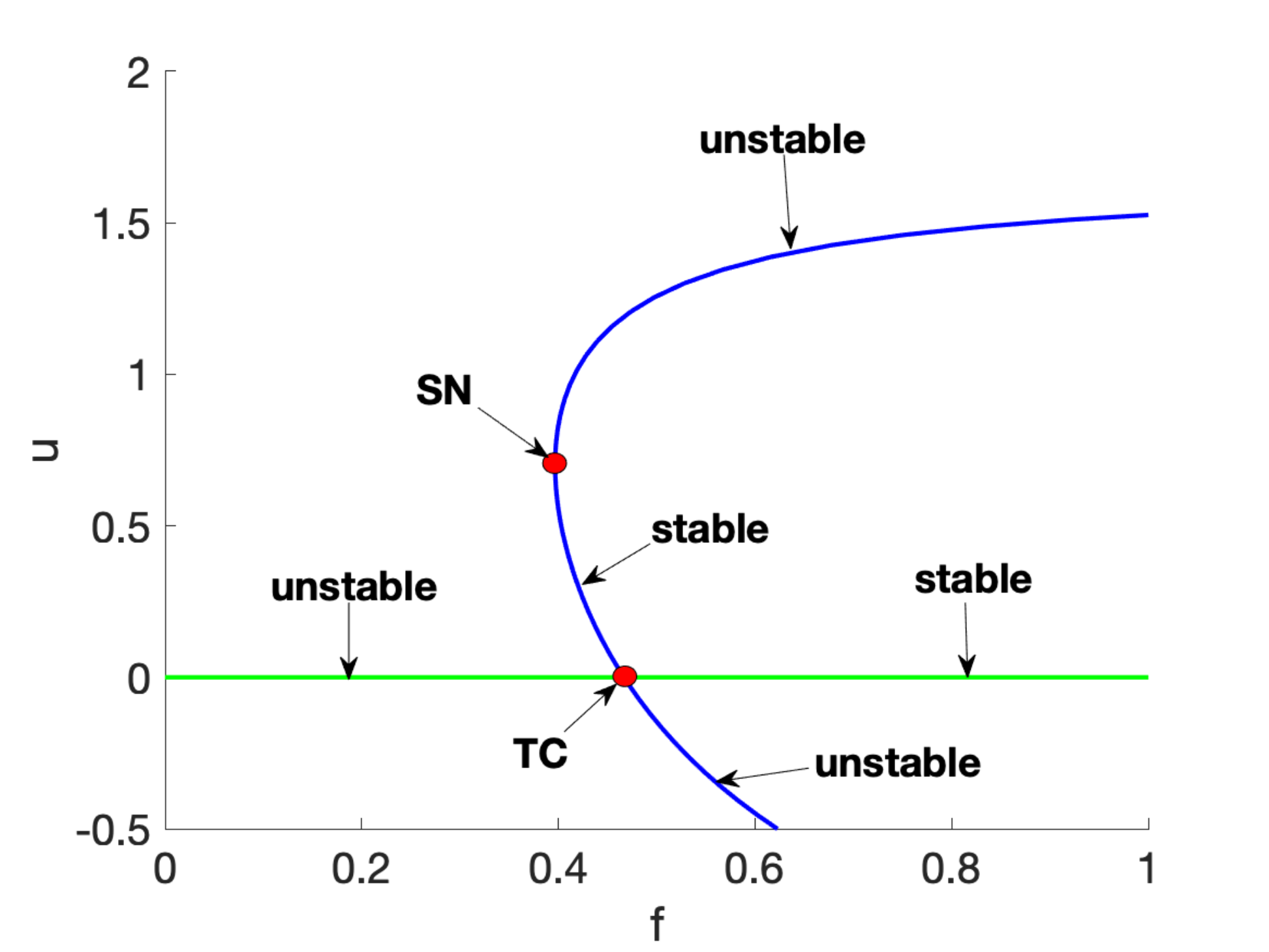}
		\caption{}
		\label{fig:transB}
	\end{subfigure}
	\caption{Bifurcation diagrams showing the impact of the fear parameter $f$ with different parameter sets. The parameters are chosen as $a_2=3, b_2=1, c_1=0.5, c_2=1.8, k=0$. In $(A)$, $a_1=1.8, b_1=1$ and a transcritical bifurcation is observed at $f=f^*=0.0666$. In $(B)$, $a_1=3.6, b_1=1.8$ and we observe the occurrence of a saddle-node at $f=f^*=0.3966$ and a transcritical bifurcation at $f=f^*=0.4666$. The initial condition was chosen as $(u_0,v_0)=(0.5,3)$ for both $(A)$ and $(B)$. (Note: TC=Transcritical point, SN=Saddle-Node point.) }
	\label{fig:trans}
\end{figure}



\begin{table}[H]\label{Table:3}
	\caption{Effect of fear on classical competition ODE dynamics for the case of $v$ fearing $u$}
	\scalebox{0.72}{
		\begin{tabular}{|c|c|c|}
			\hline
			& Classical case & $k>0, f=0$\\
			&  &       \\        \hline
			(i) & CE $(u^{*},0)$ &  Species $v$ is competitively excluded for both small and large $k$. \\
			&  &      \\        \hline
			(ii) & CE $(0, v^{*})$ &  (1) If $\frac{a_2 b_1^2 \mathbf{k}}{(b_1 + k a_1)^2}< \frac{b_1 b_2 - c_1 c_2}{c_1}$, then species $u$ is competitively excluded.   \\
			&  &  (2) Possibility of one interior saddle equilibrium, with a large level of $k$.    \\
			&  &  (3) Possibility of two positive interior equilibria, one sink and one saddle, with an intermediate level of $k$.     \\        \hline
			(iii) & weak competition   & (1) If $k< k_c = \frac{b_2b_1-c_2c_1}{a_2c_1}$ via Thm~$\ref{thm:exist}$ yields co-existence.  \\
			&  & (2) If $k>k_c,$ then species $v$ is competitively excluded.      \\        \hline
			(iv) & strong competition &  Interior equilibrium always exists and is a saddle.  \\
			&  &       \\        \hline
	\end{tabular}}
\end{table}

\begin{table}[H]\label{Table:4}
	\caption{Effect of fear on classical competition ODE dynamics for the case of $u$ fearing $v$}
	\scalebox{0.72}{
		\begin{tabular}{|c|c|c|}
			\hline
			& Classical case & $k=0, f>0$\\
			&  &       \\        \hline
			(i) & CE $(u^{*},0)$ &  (1) Species $v$ is competitively excluded, with a small level of $f$.  \\
			&  &  (2) Possibility of one interior saddle equilibrium, with a large level of $f$. \\
			&  & (3) Possibility of two positive interior equilibria, one sink and one saddle, with an intermediate level of $f$.   \\        \hline
			(ii) & CE $(0, v^{*})$ &  Species $u$ is competitively excluded for both small and large $f$.  \\
			&  &     \\  \hline
			(ii) & weak competition &  (1) For small $f$ we have co-existence.   \\
			&  &  (2) For large $f$, $u$ is competitively excluded.    \\  \hline
			iv) & strong competition &   \\
			&  &  Interior equilibrium  always exists and is a saddle.    \\        \hline
	\end{tabular}}
\end{table}

\section{The PDE Case}
Species diffusion is ubiquitous in spatial ecology \cite{Cantrell2003}. Species disperse to find mates, food and shelter \cite{O01}. Such movement is modeled often via reaction diffusion systems \cite{Cantrell2003}. The spatially explicit Lokta-Volterra model, particularly in the case of heterogeneity in spatial resources has been intensely investigated, \cite{Dockery1998, He2013a, Hastings, He2013b, He2016b, Ninomiya1995, Lam2012, Lou2006a, Lou2006b, Lou2008, Ni2012}. Herein, we consider the spatially explicit version of \eqref{eq:ODE2}, resulting in the following reaction diffusion system,

\begin{align}\label{eq:PDE-c}
	\begin{split}
		u_t &= d_1 \Delta u + a_1 u -b_1 u^2 -c_1 uv, \quad x\in \Omega, \\
		v_t &=  d_2 \Delta v +\dfrac{a_2 u}{1+ k v}  -b_2 v^2 -c_2 uv,\quad x\in \Omega,  \\
		\dfrac{\partial u}{\partial \nu} & = \dfrac{\partial v}{\partial \nu} =0, \quad \text{on} \quad \partial \Omega,\\
		u(x,0)&=u_0 (x), \quad v(x,0)=v_0(x).
	\end{split}
\end{align}
 Here $u(x,t), v(x,t)$ are the densities of two competing species, where $v$ is also fearful of $u$. The species diffuse in a bounded domain $\Omega \subset \mathbb{R}^{n}$, with dispersal speeds $d_{1}$ and  $d_{2}$ respectively. We impose no flux Neumann boundary conditions, modeling the effect that the species do not immigrate or emigrate from the domain $\Omega$. 
 We now proceed to study the dynamics of the above model, when various forms of fear are considered.

\subsubsection{Notations and preliminary observations}
\label{notat}

To prove global existence of solutions to \eqref{eq:PDE-c}, it suffices to derive uniform estimate on the $\mathbb{L}^{p}$ norms of the R.H.S. of \eqref{eq:PDE-c}, for some $p > \frac{n}{2}$. 
Classical theory will then yield global existence, \cite{henry}.
The usual norms in spaces $\mathbb{L}^{p}(\Omega )$, $\mathbb{L}^{\infty
}(\Omega ) $ and $\mathbb{C}\left( \overline{\Omega }\right) $ are
respectively denoted by

\begin{equation}
\label{(2.2)}
\left\| u\right\| _{p}^{p} = 
\int_{\Omega }\left| u(x)\right|^{p}dx, \ \left\| u\right\| _{\infty }\text{=}\underset{x\in \Omega }{max}\left|
u(x)\right| .  
\end{equation}

 To this end, we use standard techniques \cite{ Morgan89}. We first recall classical results guaranteeing non-negativity of solutions, local and global existence \cite{P10, Morgan89}:

\begin{lemma}\label{lem:class1}
Let us consider the following $m\times m$ - reaction diffusion system: for all $i=1,...,m,$ 
\begin{equation}
\label{eq:class1}
\partial_t u_i-d_i\Delta u_i=f_i(u_1,...,u_m)~in~ \mathbb{R}_+\times \Omega,~ \partial_\nu u_i=0~ \text{on}~ \partial \Omega, u_i(0)=u_{i0},
\end{equation}
where $d_i \in(0,+\infty)$, $f=(f_1,...,f_m):\mathbb{R}^m \rightarrow \mathbb{R}^m$ is $C^1(\Omega)$ and $u_{i0}\in L^{\infty}(\Omega)$. Then there exists a $T>0$ and a unique classical solution of \ref{eq:class1} on $[0,T).$ If $T^*$ denotes the greatest of these $T's$, then 
\begin{equation*}
\Bigg[\sup_{t \in [0,T^*),1\leq i\leq m} ||u_i(t)||_{L^{\infty}(\Omega)} < +\infty \Bigg] \implies [T^*=+\infty].
\end{equation*}
If the nonlinearity $(f_i)_{1\leq i\leq m}$ is moreover quasi-positive, which means 
$$\forall i=1,..., m,~~\forall u_1,..., u_m \geq 0,~~f_i(u_1,...,u_{i-1}, 0, u_{i+1}, ..., u_m)\geq 0,$$
then $$[\forall i=1,..., m, u_{i0}\geq 0]\implies [\forall i=1,...,m,~ \forall t\in [0,T^*), u_i(t)\geq 0].$$
\end{lemma}

\begin{lemma}\label{lem:class2}
Using the same notations and hypotheses as in Lemma \ref{lem:class1}, suppose moreover that $f$ has at most polynomial growth and that there exists $\mathbf{b}\in \mathbb{R}^m$ and a lower triangular invertible matrix $P$ with nonnegative entries such that  $$\forall r \in [0,+\infty)^m,~~~Pf(r)\leq \Bigg[1+ \sum_{i=1}^{m} r_i \Bigg]\mathbf{b}.$$
Then, for $u_0 \in L^{\infty}(\Omega, \mathbb{R}_+^m),$ the system (\ref{eq:class1}) has a strong global solution.
\end{lemma}

Under these assumptions, the following local existence result is well known, see D. Henry \cite{henry}, 

\begin{theorem}
\label{thm:class3}
The system \eqref{eq:PDE} admits a unique, classical solution $(u,v)$ on $%
[0,T_{\max }]\times \Omega $. If $T_{\max }<\infty $ then 
\begin{equation}
\underset{t\nearrow T_{\max }}{\lim }\Big\{ \left\Vert u(t,.)\right\Vert
_{\infty }+\left\Vert v(t,.)\right\Vert _{\infty } \Big\} =\infty ,  
\end{equation}%
where $T_{\max }$ denotes the eventual blow-up time in $\mathbb{L}^{\infty }(\Omega ).$
\end{theorem}

The next result follows from the application of standard theory \cite{kish88},

\begin{theorem}
\label{thm:km1}
	Consider the reaction diffusion system \eqref{eq:PDE-c}. For spatially homogenous initial data $u_{0} \equiv c, v_{0} \equiv d$, with $c,d>0$, then the dynamics of \eqref{eq:PDE} and its resulting kinetic (ODE) system, when $d_{1}=d_{2}=0$ in \eqref{eq:PDE-c}, are equivalent.
\end{theorem}

\subsubsection{Spatially Heterogeneous Fear}

Our objective now is to consider the case of a fear function that may be heterogeneous in space. A motivation for this comes from several ecological and sociological settings. For example it is very common for prey to be highly fearful closer to a predators lair, but less fearful in a region of refuge \cite{Zhang19}, or in regions of high density due to group defense \cite{Samsal20}. Furthermore, a conceivably weaker drug cartel, could have certain localized strongholds, within which they would be more feared by stronger groups. To these ends, it is conceivable that the fear coefficient $k$ is not a constant, but actually varies in the spatial domain $\Omega$, so $k=k(x)$, which could take different forms depending on the application at hand. This is also in line with the LOF concept \cite{Brown99}. Thus we consider the following spatially explicit version of \eqref{eq:ODE2}, with heterogeneous fear function $k(x)$, resulting in the following reaction diffusion system,

\begin{align}\label{eq:PDE}
	\begin{split}
		u_t &= d_1 (u)_{xx} + a_1 u -b_1 u^2 -c_1 uv, \quad x\in \Omega, \\
		v_t &=  d_2 (v)_{xx} +\dfrac{a_2 u}{1+k(x) v}  -b_2 v^2 -c_2 uv,\quad x\in \Omega,  \\
		\dfrac{\partial u}{\partial \nu} & = \dfrac{\partial v}{\partial \nu} =0, \quad \text{on} \quad \partial \Omega.\\
		u(x,0)&=u_0 (x) \equiv c > 0, \quad v(x,0)=v_0(x) \equiv d > 0,
	\end{split}
\end{align}
where $\Omega \subset \mathbb{R}^{n}$. We assume no flux Neumann boundary conditions. Also we prescribe spatially homogeneous (flat) initial conditions
$u(x,0)=u_0 (x) \equiv c > 0, \quad v(x,0)=v_0(x) \equiv d > 0.$ Furthermore, we impose the following restrictions on the fear function $k(x)$,

\begin{align}\label{eq:as1}
\begin{split}
	&(i) \quad k(x)  \in C^{1}(\Omega),
	\\
	&(ii) \quad k(x) \geq 0,
	\\
	& (iii)\quad  \mbox{If} \  k(x) \equiv 0 \ \mbox{on}  \ \Omega_{1} \subset \Omega, \ \mbox{then} \ |\Omega_{1}| = 0.
\\
& (iv)\quad  \mbox{If} \  k(x) \equiv 0 \ \mbox{on}  \ \cup^{n}_{i=1}\Omega_{i} \subset \Omega, \ \mbox{then} \ \Sigma^{n}_{i=1}|\Omega_{i}| = 0.
\end{split}
\end{align}

\begin{remark}
If $k(x) \equiv 0$ on $\Omega_{1} \subset \Omega$, with $|\Omega_{1}| > \delta > 0$, or $k(x) \equiv 0$ on $\cup^{n}_{i=1}\Omega_{i} \subset \Omega$, with $\Sigma^{n}_{i=1}|\Omega_{i}| > \delta > 0$, that is, on non-trivial parts of the domain, the analysis is notoriously difficult, as one now is dealing with a \emph{degenerate} problem. See \cite{Du02a, Du02b} for results on this problem. This case is not in the scope of the current manuscript. 
\end{remark}

Since the nonlinear right hand side of \eqref{eq:PDE} is continuously
differentiable on $\mathbb{R}^{+}\times $ $\mathbb{R}^{+}$, then for any
initial data in $\mathbb{C}\left( \overline{\Omega }\right) $ or $\mathbb{L}%
^{p}(\Omega ),\;p\in \left( 1,+\infty \right) $, it is standard to 
estimate the $\mathbb{L}^{p}-$norms of the solutions and thus deduce global existence. Standard theory will apply even in the case of a bonafide fear function $k(x)$, because due to our assumptions on the form of $k$,  standard comparison arguments will apply \cite{Gil77}. Thus applying the classical methods above, via Theorem \ref{thm:class3}, and Lemmas \ref{lem:class1}-\ref{lem:class2}, we can state the following lemmas,

\begin{lemma}
\label{lem:pos1}
Consider the reaction diffusion system \eqref{eq:PDE}, for $k(x)$ s.t the assumtions via \eqref{eq:as1} hold. Then solutions to \eqref{eq:PDE} are non-negative, as long as they initiate from positive initial conditions.
\end{lemma}

\begin{lemma}
\label{lem:cl1}
Consider the reaction diffusion system \eqref{eq:PDE}. For $k(x)$ s.t the assumtions via \eqref{eq:as1} hold. The solutions to \eqref{eq:PDE} are classical. That is for $(u_{0},v_{0}) \in \mathbb{L}^{\infty }(\Omega )$,  $(u,v) \in C^{1}(0,T; C^{2}(\Omega))$, $\forall T$.
\end{lemma}

Our goal in this section is to investigate the dynamics of \eqref{eq:PDE}. Herein we will us the comparison technique, and compare to the ODE cases of classical competition, or the constant fear function case, where the dynamics are well known. 

\begin{remark}
	The analysis in this section are primarily focused on the choice of spatially homogenous (flat) initial data.
\end{remark}

We begin by defining the following systems of PDEs,

\begin{align}\label{eq:lv_model}
	\begin{split}
		\overline{u}_t &= d_1 (\overline{u})_{xx} + a_1 \overline{u} -b_1 \overline{u}^2 -c_1 \overline{u}\overline{v}, \\
		\overline{v}_t &=  d_2 (\overline{v})_{xx} +a_2 \overline{v}  -b_2 \overline{v}^2 -c_2 \overline{u}\overline{v},
	\end{split}
\end{align}

\begin{align}\label{eq:upper}
	\begin{split}
		\widehat{u_t} &= d_1 (\widehat{u_t})_{xx} + a_1 \widehat{u} -b_1 \widehat{u}^2 -c_1 \widehat{u}\widehat{v}, \\
		\widehat{v}_t &=  d_2 (\widehat{v})_{xx} +\dfrac{a_2 \widehat{v}}{1+\mathbf{\widehat{k}} \widehat{u}}  -b_2 \widehat{v}^2 -c_2 \widehat{u}\widehat{v},
	\end{split}
\end{align}

\begin{align}\label{eq:lower}
	\begin{split}
		\widetilde{u} _t &= d_1 ( \widetilde{u})_{xx} + a_1  \widetilde{u} -b_1  \widetilde{u}^2 -c_1  \widetilde{u}\widetilde{v},\\
		\widetilde{v}_t &=  d_2 (\widetilde{v})_{xx} +\dfrac{a_2  \widetilde{v}}{1+\mathbf{\widetilde{k}} \widetilde{u}}  -b_2 \widetilde{v}^2 -c_2  \widetilde{u}\widetilde{v},
	\end{split}
\end{align}

\begin{align}\label{eq:lowest}
	\begin{split}
		\tilde{u} _t &= d_1 ( \tilde{u})_{xx} + a_1  \tilde{u} -b_1  \tilde{u}^2 -c_1  \tilde{u}\widetilde{v},\\
		\tilde{v}_t &=  d_2 (\tilde{v})_{xx} +\dfrac{a_2  \tilde{v}}{1+\mathbf{\tilde{k}} \frac{a_{1}}{b_{1}}}  -b_2 \tilde{v}^2 -c_2  \tilde{u}\tilde{v},
	\end{split}
\end{align}

where
\begin{align}\label{eq:lowest1}
 \mathbf{\widehat{k}} = \min_{x\in \Omega} k(x), \quad \, \quad   \mathbf{\widetilde{k} }= \max_{x\in \Omega} k(x).
\end{align}

We assume no flux Neumann boundary conditions for all of the reaction diffusion systems \eqref{eq:lv_model} - \eqref{eq:lowest}. Also in each of the systems we prescribe spatially homogenous (flat) initial conditions
$u(x,0)=u_0 (x) \equiv c > 0, \quad v(x,0)=v_0(x) \equiv d > 0.$

We now state the following lemma,

\begin{lemma}
\label{lem:com1}
Consider the reaction diffusion system \eqref{eq:PDE}, as well as the reaction diffusion systems \eqref{eq:lv_model} - \eqref{eq:lowest}. Then the following point wise comparison holds,

\begin{align*}
	 \tilde{v} \leq	\widetilde{v}	 \leq   v  \leq     \widehat{v}	          \leq \overline{v}.
	\end{align*}

\end{lemma}

\begin{proof}

Note via positivity of solutions to \eqref{eq:PDE}, \eqref{eq:lv_model} - \eqref{eq:lowest}, the definitions via  \eqref{eq:lowest1}, and the upper bound on $u$, which is a solution to \eqref{eq:PDE} of $\frac{a_{1}}{b_{1}}$, (derived via comparison to the logistic equation), we have
\begin{align}\label{eq:com1}
	\dfrac{a_{2}}{1 +  \mathbf{\widetilde{k}}  \frac{a_{1}}{b_{1}}}  \le    
	 \dfrac{a_{2}}{1 +  \mathbf{\widetilde{k}}  u(x)} \le 
	 \dfrac{a_{2}}{1 + k(x) u(x)} \le 
	 \dfrac{a_{2}}{1 +  \mathbf{\widehat{k}}u(x)} \le 
	 a_{2}, 
	 \quad \forall x\in \Omega.
\end{align}
	
Thus the result follows via standard comparison theory \cite{Gil77}.
\end{proof}

\subsection{The Competitive Exclusion Case}

\begin{theorem}\label{thm:ce_1_pde}
	Consider the reaction diffusion system \eqref{eq:PDE}, for a fear function $k(x)$, s.t the assumptions via \eqref{eq:as1} are met, and
	\[ \mathbf{\widetilde{k} }>\dfrac{a_2 b_1^2-c_2a_1b_1}{a_1^2c_2} \quad \text{and} \quad \dfrac{a_1}{a_2} > \max \Big\{ \dfrac{b_1}{c_2},\dfrac{c_1}{b_2}  \Big\}. \]
	Then the solution $(u,v)$ to \eqref{eq:PDE} converges uniformly to the spatially homogenous state $(\frac{a_{1}}{b_{1}},0)$ as $t \to \infty$.
\end{theorem}

\begin{proof}
%
%
	From the classical theory of competition \cite{Murray93}, we know the dynamics for \eqref{eq:lv_model}, that is in the competitive exclusion case, when
	\[  \dfrac{a_1}{a_2} > \max \Big\{ \dfrac{b_1}{c_2},\dfrac{c_1}{b_2}  \Big \} , \]
	we have 
	\[ (\overline{u},\overline{v}) \to \Big(\dfrac{a_1}{b_1},0\Big).\]
	Moreover, under the assumption 
	\[ \mathbf{\widetilde{k} }>\dfrac{a_2 b_1^2-c_2a_1b_1}{a_1^2c_2}, \]
	and making use of Lemma \ref{lem:ce_ode_2}, along with the use of Theorem \ref{thm:km1}, we have
	$(\widetilde{u},\widetilde{v}) \to \Big(\dfrac{a_1}{b_1},0\Big)$. Now using Lemma \ref{lem:com1} we have,
	\begin{align*}
	 \widetilde{v}	 \leq   v    \leq \overline{v},
	\end{align*}
	 which entails,
	\begin{align*}
	\lim_{t \rightarrow \infty}(\widetilde{u}, \widetilde{v})	 \leq  \lim_{t \rightarrow \infty} (u,v)    \leq \lim_{t \rightarrow \infty} (\overline{u},\overline{v}),
	\end{align*}
	subsequently,
	\begin{align*}
		\left(\frac{a_{1}}{b_{1}},0\right)	 \leq  \lim_{t \rightarrow \infty}   (u,v)    \leq 	\left(\frac{a_{1}}{b_{1}},0\right).
	\end{align*}
	Now using a squeezing argument, in the limit that $t \rightarrow \infty$, we have uniform convergence of solutions of \eqref{eq:PDE}, i.e.,
	\[ (u,v) \to \Big(\dfrac{a_1}{b_1},0\Big)  \] as $t \rightarrow \infty$.
\end{proof}

Using the positivity of solutions the requirement on $k(x)$ and so in turn on $\mathbf{\widetilde{k} }$, can be weakened to derive a stronger result,

\begin{theorem}\label{thm:ce_1_pde22}
	Consider the reaction diffusion system \eqref{eq:PDE}, for a fear function $k(x)$, s.t the assumptions via \eqref{eq:as1} are met, and $\dfrac{a_1}{a_2} > \max \Big\{ \dfrac{b_1}{c_2},\dfrac{c_1}{b_2}  \Big\}$.
	Then the solution $(u,v)$ to \eqref{eq:PDE} converges uniformly to the spatially homogeneous state $(\frac{a_{1}}{b_{1}},0)$ as $t \to \infty$.
\end{theorem}

\begin{proof}

	From the classical theory of competition \cite{Murray93}, we know the dynamics for \eqref{eq:lv_model}, that is in the competitive exclusion case, when
	\[  \dfrac{a_1}{a_2} > \max \Big\{ \dfrac{b_1}{c_2},\dfrac{c_1}{b_2}  \Big \} , \]
	we have 
	\[ (\overline{u},\overline{v}) \to \Big(\dfrac{a_1}{b_1},0\Big).\]
	 Now using Lemma \ref{lem:com1}, we have
	\begin{align*}
	 v    \leq \overline{v}.
	\end{align*}
	 Using the non negativity of solutions to \eqref{eq:PDE} via Lemma \ref{lem:pos1}, entails,
	\begin{align*}
	 0 \leq \lim_{t \rightarrow \infty} v    \leq \lim_{t \rightarrow \infty} (\overline{v}) = 0,
	\end{align*}
	subsequently,
	\begin{align*}
		  \lim_{t \rightarrow \infty}   (u,v)    \rightarrow 	\left(\frac{a_{1}}{b_{1}},0\right).
	\end{align*}
	
\end{proof}

\begin{theorem}\label{thm:ce_2_pde}
	Consider the reaction diffusion system \eqref{eq:PDE}, for a fear function $k(x)$, s.t the assumptions via \eqref{eq:as1} are met, and
	\[ \dfrac{a_2 b_1^2 \mathbf{\widetilde{k}}}{(b_1 + \mathbf{\widetilde{k}} a_1)^2}< \dfrac{b_1 b_2 - c_1 c_2}{c_1} \quad \text{and} \quad \dfrac{a_1}{a_2} < \min \Big\{ \dfrac{b_1}{c_2},\dfrac{c_1}{b_2}  \Big\}. \]
Then the solution $(u,v)$ to \eqref{eq:PDE}  converges uniformly to the spatially homogeneous state $(0, \frac{a_{2}}{b_{2}})$ as $t \to \infty$.
\end{theorem}

\begin{proof}
From the classical theory of competition \cite{Murray93}, we know the dynamics for \eqref{eq:lv_model}, that is in the competitive exclusion case, when
\[  \dfrac{a_1}{a_2} < \min \Big\{ \dfrac{b_1}{c_2},\dfrac{c_1}{b_2}  \Big\} , \]
we have 
\[ (\overline{u},\overline{v}) \to \Big(0,\dfrac{a_2}{b_2}\Big).\]
Consider the nullclines of \eqref{eq:lower}.  If
\[ \dfrac{d}{du} \Big[ \dfrac{a_1-b_1 u}{c_1} \Big] < \dfrac{d}{du} \Big[ \dfrac{1}{b_2} \Big( \dfrac{a_2}{1+ \mathbf{\widetilde{k}} u} -c_2 u \Big)\Big], \]
it follows via the geometry of the nullclines, that the $v$-nullcline remains above the $u$-nullcline, and thus
 $(\widetilde{u},\widetilde{v}) \to \Big(0, \dfrac{a_2}{b_2}\Big)$.
On simplification, and using the upper bound estimate for $u$, we have
\[ \dfrac{a_2 b_1^2 \mathbf{\widetilde{k}}}{(b_1 + k a_1)^2} < \dfrac{a_2 \mathbf{\widetilde{k}}} {(1+\mathbf{\widetilde{k}} u)^2}< \Big( \dfrac{b_1b_2 -c_1 c_2}{c_1} \Big). \]
Now using Lemma \ref{lem:com1} we have,
\begin{align*}
	\widetilde{v}	 \leq   v    \leq \overline{v},
\end{align*}
which entails,
\begin{align*}
	\lim_{t \rightarrow \infty}(\widetilde{u}, \widetilde{v})	 \leq  \lim_{t \rightarrow \infty} (u,v)    \leq \lim_{t \rightarrow \infty} (\overline{u},\overline{v}),
\end{align*}
subsequently,
\begin{align*}
	\left(0,\frac{a_{2}}{b_{2}}\right)	 \leq  \lim_{t \rightarrow \infty}   (u,v)    \leq 	\left(0,\frac{a_{2}}{b_{2}}\right).
\end{align*}
Now using a squeezing argument, in the limit that $t \rightarrow \infty$, we have uniform convergence of solutions of \eqref{eq:PDE}, i.e.,
\[ (u,v) \to \Big(0,\dfrac{a_2}{b_2}\Big) \] as $t \rightarrow \infty$.
\end{proof}

\begin{figure}[h]
	\begin{subfigure}[b]{.475\linewidth}
		\includegraphics[width=\linewidth,height=1.8in]{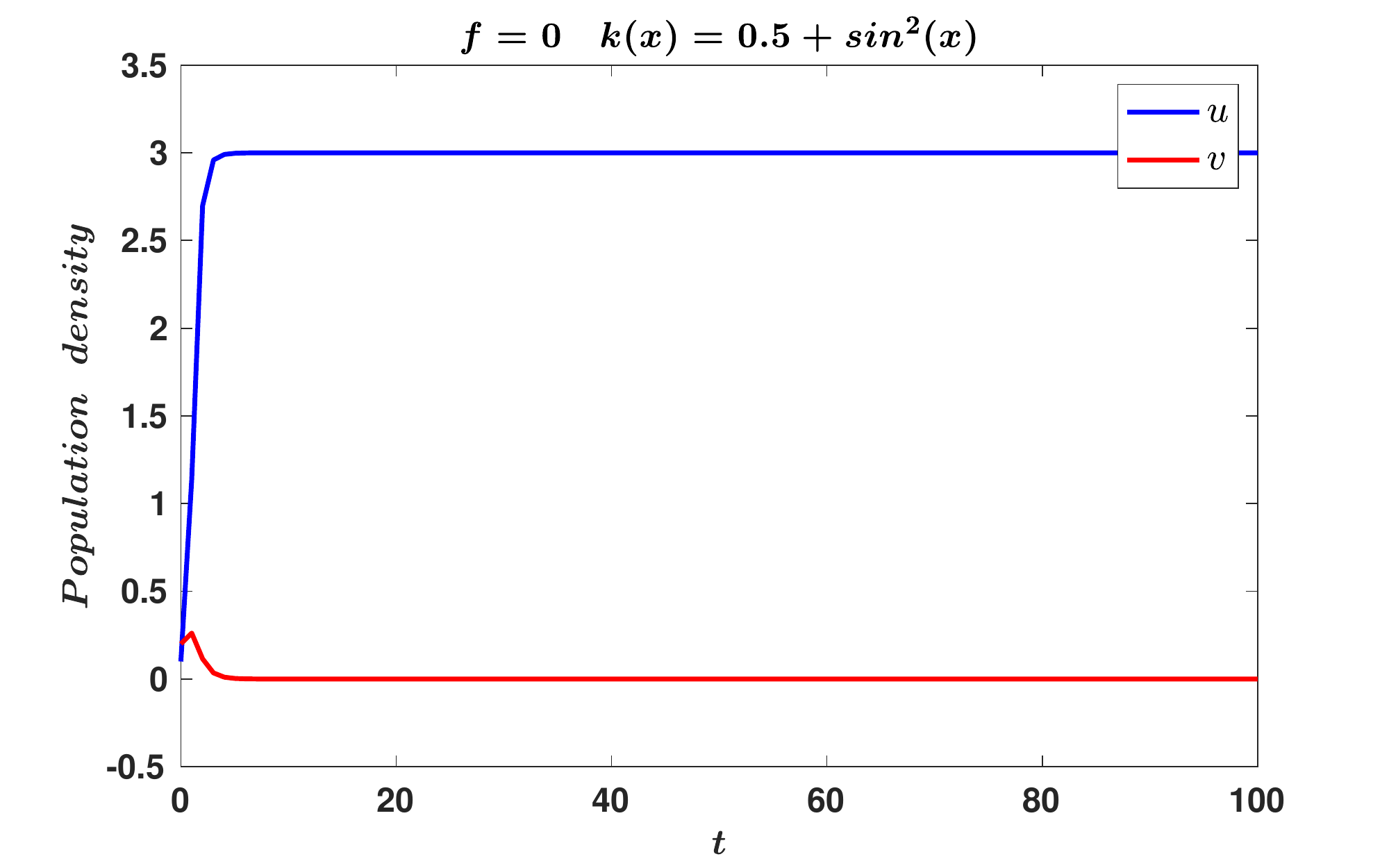}
		\caption{$f=0$ and $k=0.5+\sin^2 (x)$}
	\end{subfigure}
	\hfill
	\begin{subfigure}[b]{.475\linewidth}
		\includegraphics[width=\linewidth,height=1.8in]{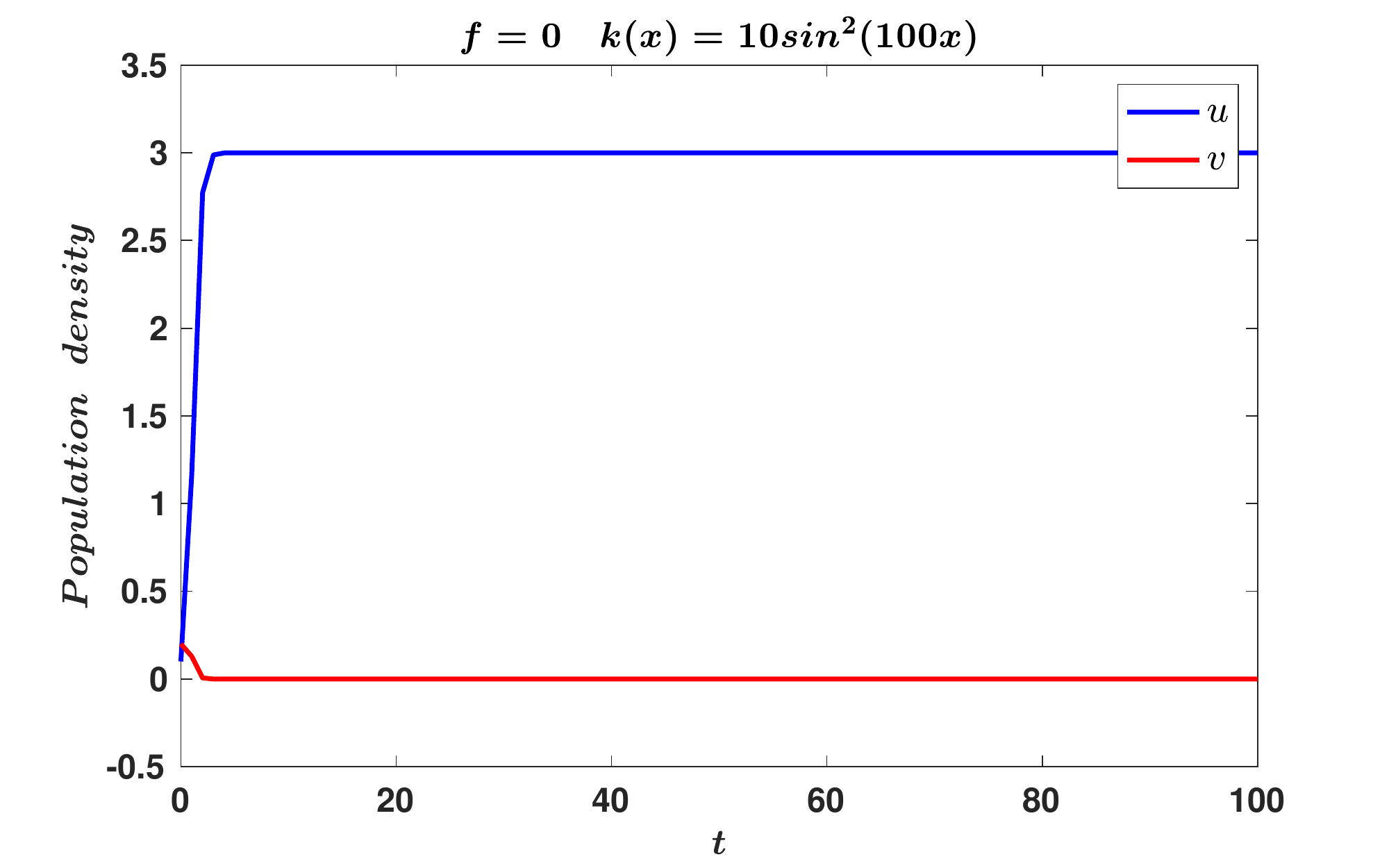}
		\caption{$f=0$ and $k=10+\sin^2 (x)$}
	\end{subfigure}
	\newline
	\begin{subfigure}[b]{.475\linewidth}
		\includegraphics[width=\linewidth,height=1.8in]{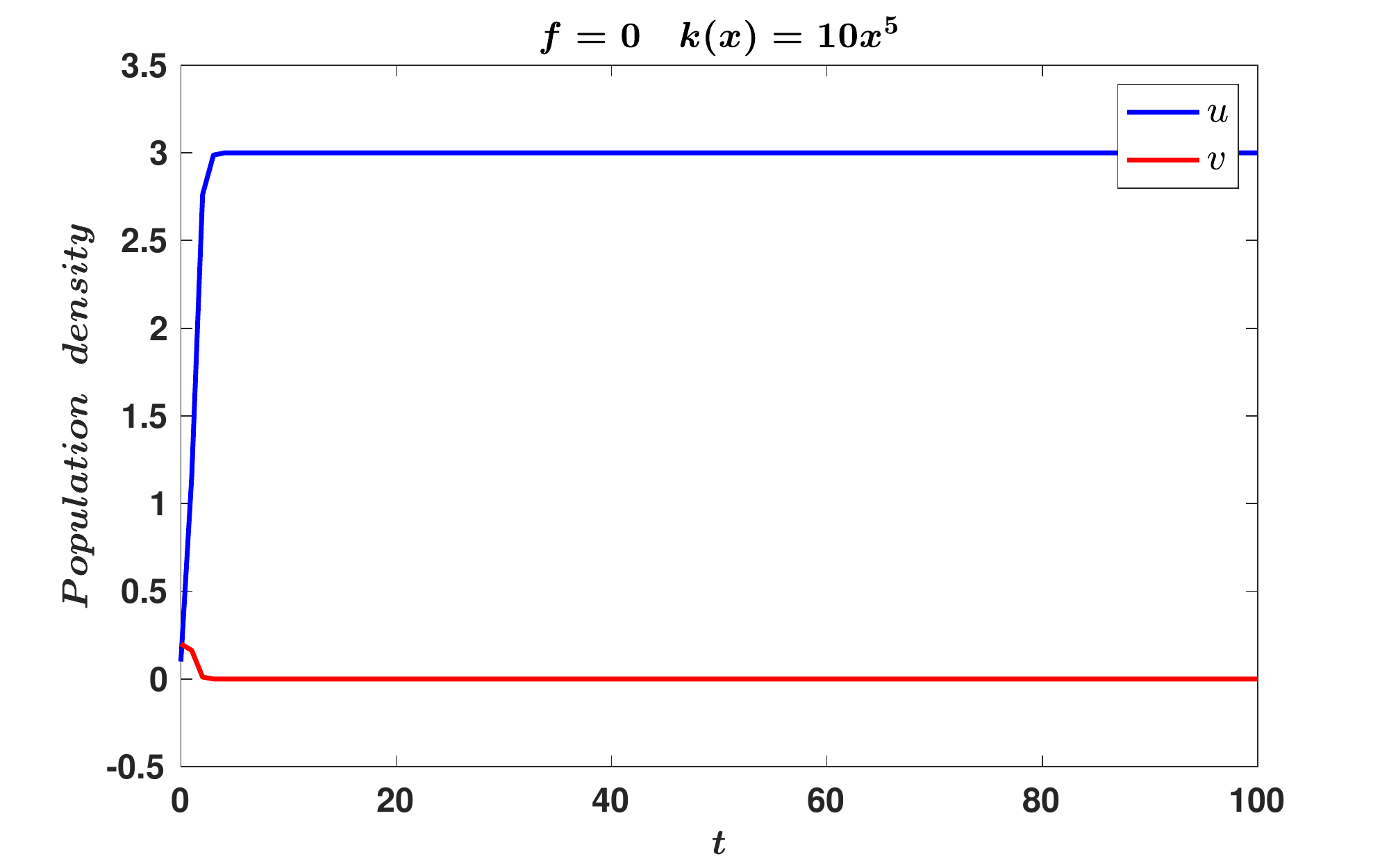}
		\caption{$f=0$ and $k=10x^5$}
	\end{subfigure}
	\hfill
	\begin{subfigure}[b]{.475\linewidth}
		\includegraphics[width=\linewidth,height=1.8in]{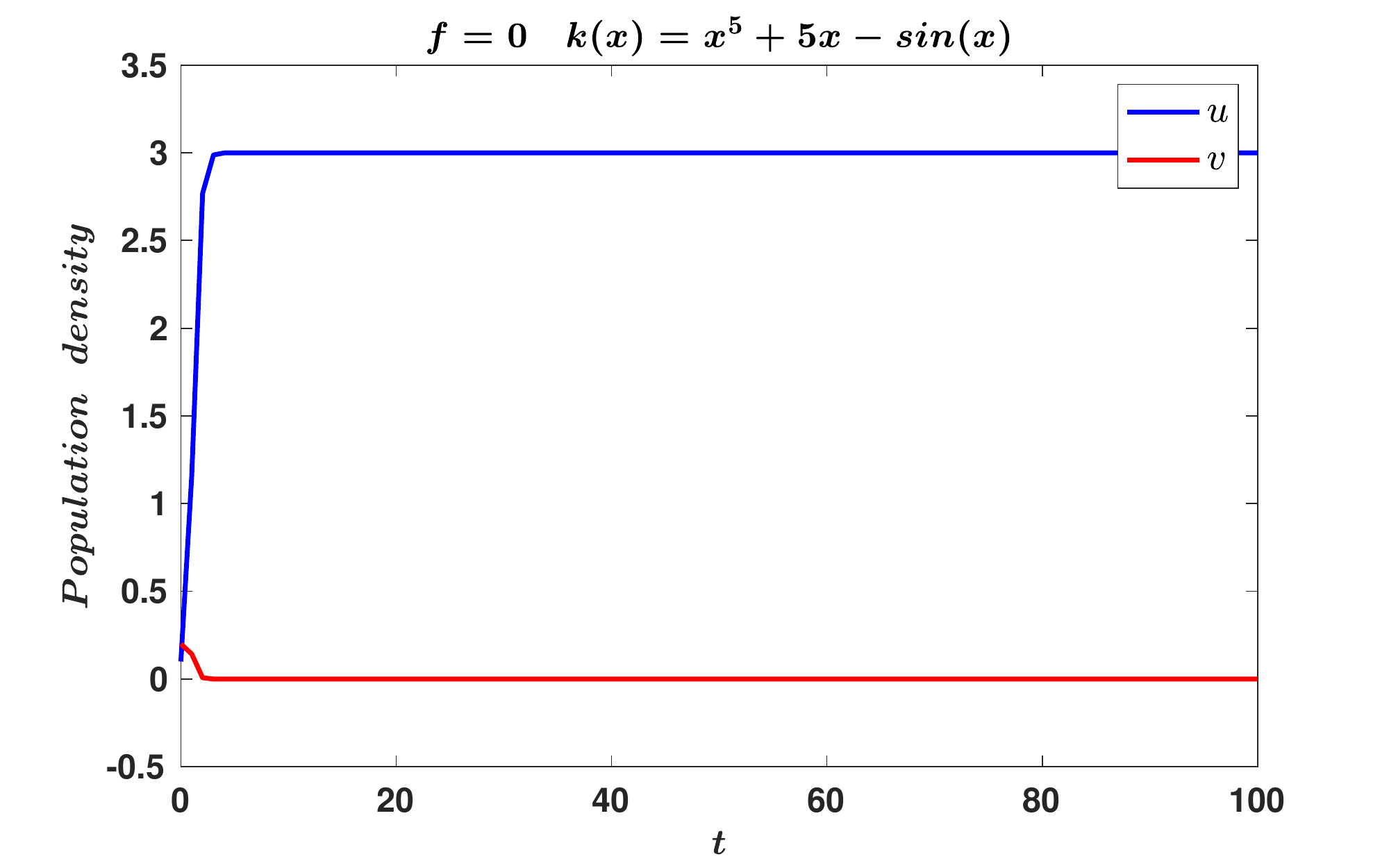}
		\caption{$f=0$ and $k=x^5+5x-\sin(x)$}
	\end{subfigure}
	\caption{Numerical simulation of $(\ref{eq:PDE})$ for the case of competition exclusion in $\Omega=[0,1]$. The parameters are chosen as $[u_0,v_0]=[0.1,0.2],d_1=1,d_2=1,a_1=3,a_2=1,b_1=b_2=1$ and $c_1=c_2=0.5$.}
	\label{fig:pde_ce1_pde1}
\end{figure}

\begin{figure}[h]
	\begin{subfigure}[b]{.475\linewidth}
		\includegraphics[width=\linewidth,height=1.8in]{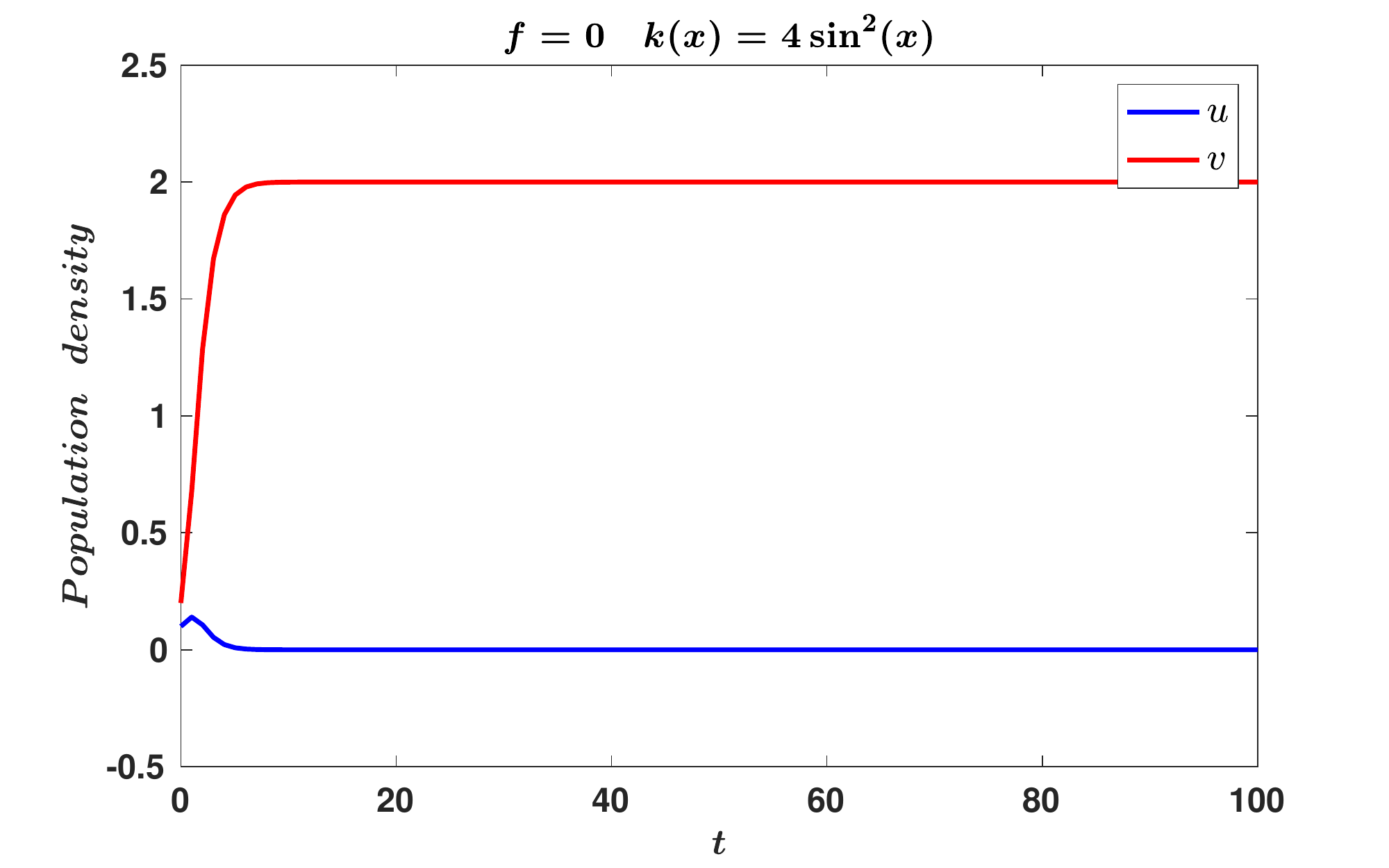}
		\caption{$f=0$ and $k=4\sin^2 (x)$}
	\end{subfigure}
	\hfill
	\begin{subfigure}[b]{.475\linewidth}
		\includegraphics[width=\linewidth,height=1.8in]{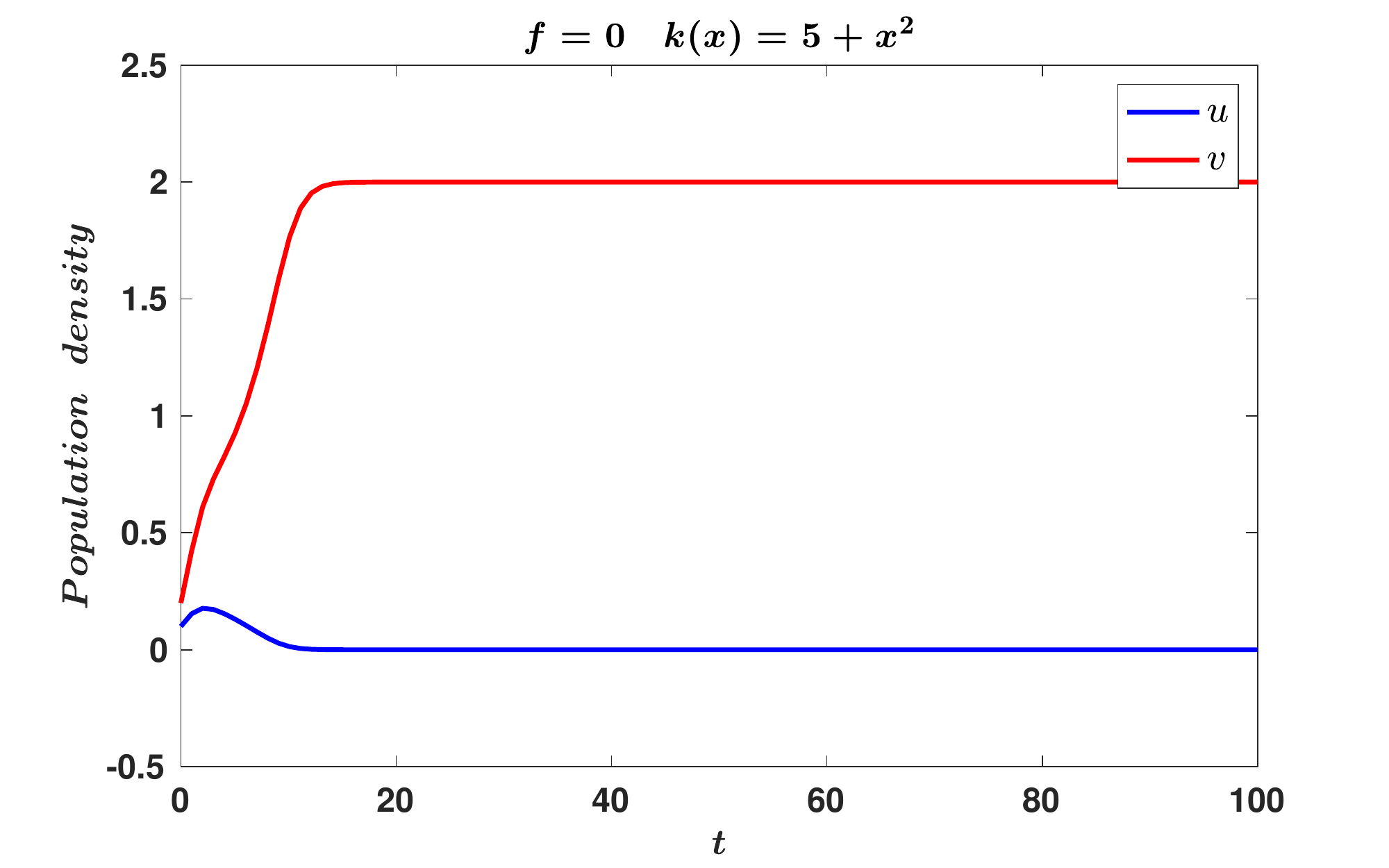}
		\caption{$f=0$ and $k=4+x^2$}
	\end{subfigure}
	\caption{Numerical simulation of $(\ref{eq:PDE})$ for the case of competition exclusion in $\Omega=[0,1]$. The parameters are chosen as $[u_0,v_0]=[0.1,0.2],d_1=1,d_2=1,a_1=1,a_2=2,b_1=2, b_2=1,c_1=1$ and $c_2=1$.}
	\label{fig:pde_ce2_pde1}
\end{figure}

\subsection{The Weak Competition Case}

We state the following result,
\begin{lemma}\label{lem:ce_1_pde-l}
	Consider the reaction diffusion system \eqref{eq:PDE}, for a fear function $k(x)$, s.t the assumptions via \eqref{eq:as1} are met, with  
	\[ \mathbf{\widetilde{k} }>\dfrac{a_2 b_1^2-c_2a_1b_1}{a_1^2c_2} \quad \text{and} \quad b_{1}b_{2} > 2 c_{1} c_{2}. \]
	Then the solution $(u,v)$ to \eqref{eq:PDE}  converges uniformly to the spatially homogeneous state $(u^*,0)$ as $t \to \infty$.
\end{lemma}

\begin{proof}
Via Lemma \ref{lem:cotce}, and the parametric restrictions assumed we have,

\begin{align*}
	\widetilde{v}	 \leq   v    \leq  \widehat{v},
	\end{align*}
	
	and the result follows via similar analysis as in Theorem \ref{thm:ce_1_pde}.
\end{proof}

%

	
%

\begin{lemma}\label{lem:ce_1_pde-l}
	Consider the reaction diffusion system \eqref{eq:PDE}, s.t. we are in the weak competition case when $k(x) \equiv 0$. Then given $0< \epsilon << 1$, there exists a fear function $k_{\epsilon}(x)$, for which the assumptions via \eqref{eq:as1} are met,
s.t. the solution $(u,v)$ to \eqref{eq:PDE} with the fear function $k_{\epsilon}(x)$, converges uniformly to a spatially homogeneous state $(u^*,v^{*})$ as $t \to \infty$.
\end{lemma}

\begin{proof}
	Given $0< \epsilon << 1$, we can always construct a $k_{\epsilon}$, s.t $k_{c} - \epsilon \leq \widehat{k}_{\epsilon} $, whereas $\widetilde{k}_{\epsilon} \leq k_{c} + \epsilon$. Thus via Lemma \ref{lem:com1} we have,
	
	\begin{align*}
		\widetilde{v}	 \leq   v_{\epsilon}    \leq \widehat{v}.
	\end{align*}
	
	Lemma \ref{lem:co_ode_1} ensures we have $(\widetilde{u},\widetilde{v}) \to (u^*,v^{*})$ and $(\widehat{u},\widehat{v}) \to (u^{**},v^{**})$, where the spatially homogeneous solutions may be different.
	Hence, via squeezing argument, we can take $\epsilon \rightarrow 0$, to yield the uniform convergence of solutions, i.e.,
	\[ \lim_{\epsilon \rightarrow 0} \lim_{t \rightarrow \infty} (u_{\epsilon},v_{\epsilon}) \to (u^*,v^{*}). \] 
	This proves the lemma.
\end{proof}

\begin{figure}[h]
	\begin{subfigure}[b]{.475\linewidth}
		\includegraphics[width=\linewidth,height=1.8in]{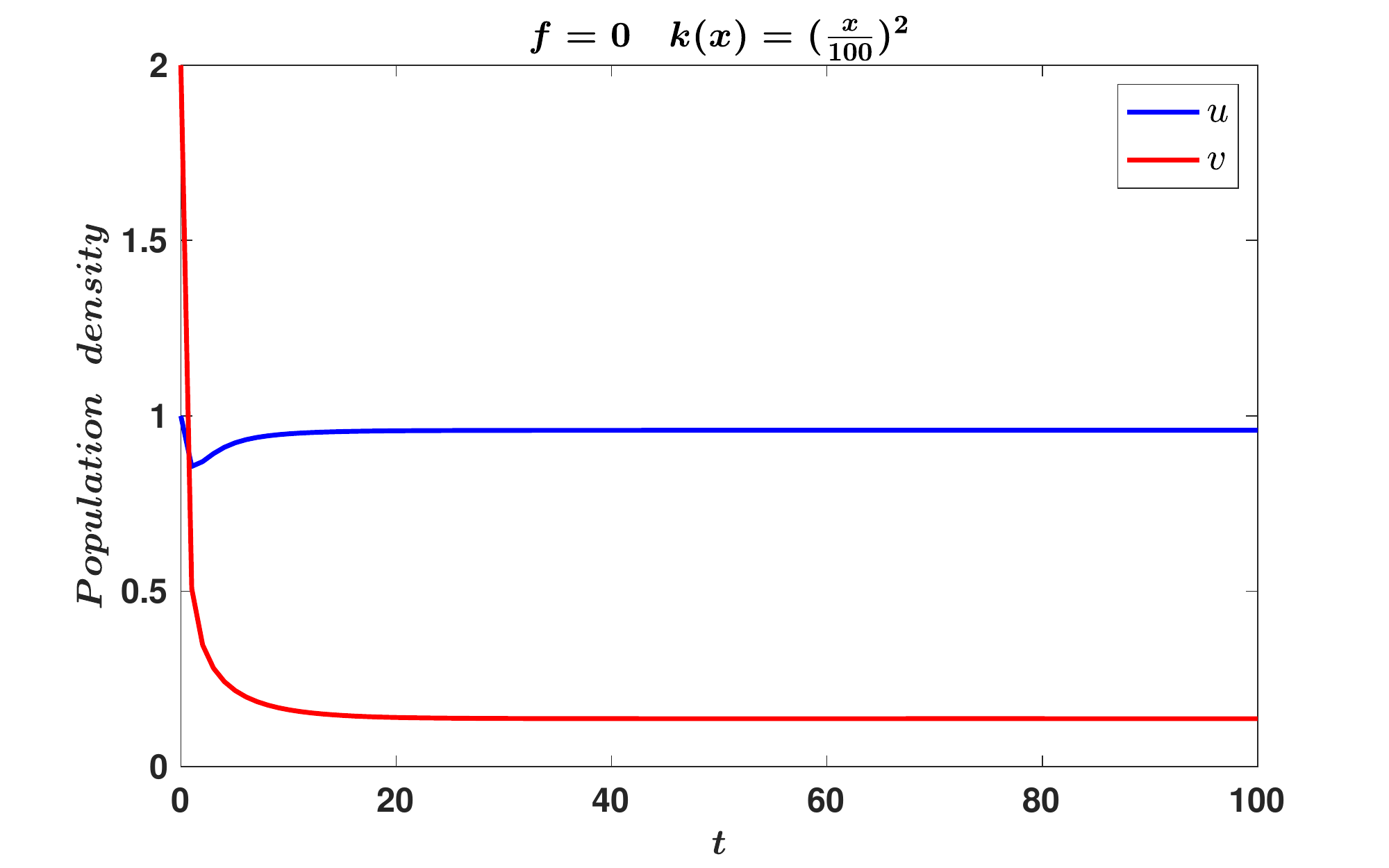}
		\caption{$[u_0,v_0]=[1,2]$ }
	\end{subfigure}
	\hfill
	\begin{subfigure}[b]{.475\linewidth}
		\includegraphics[width=\linewidth,height=1.8in]{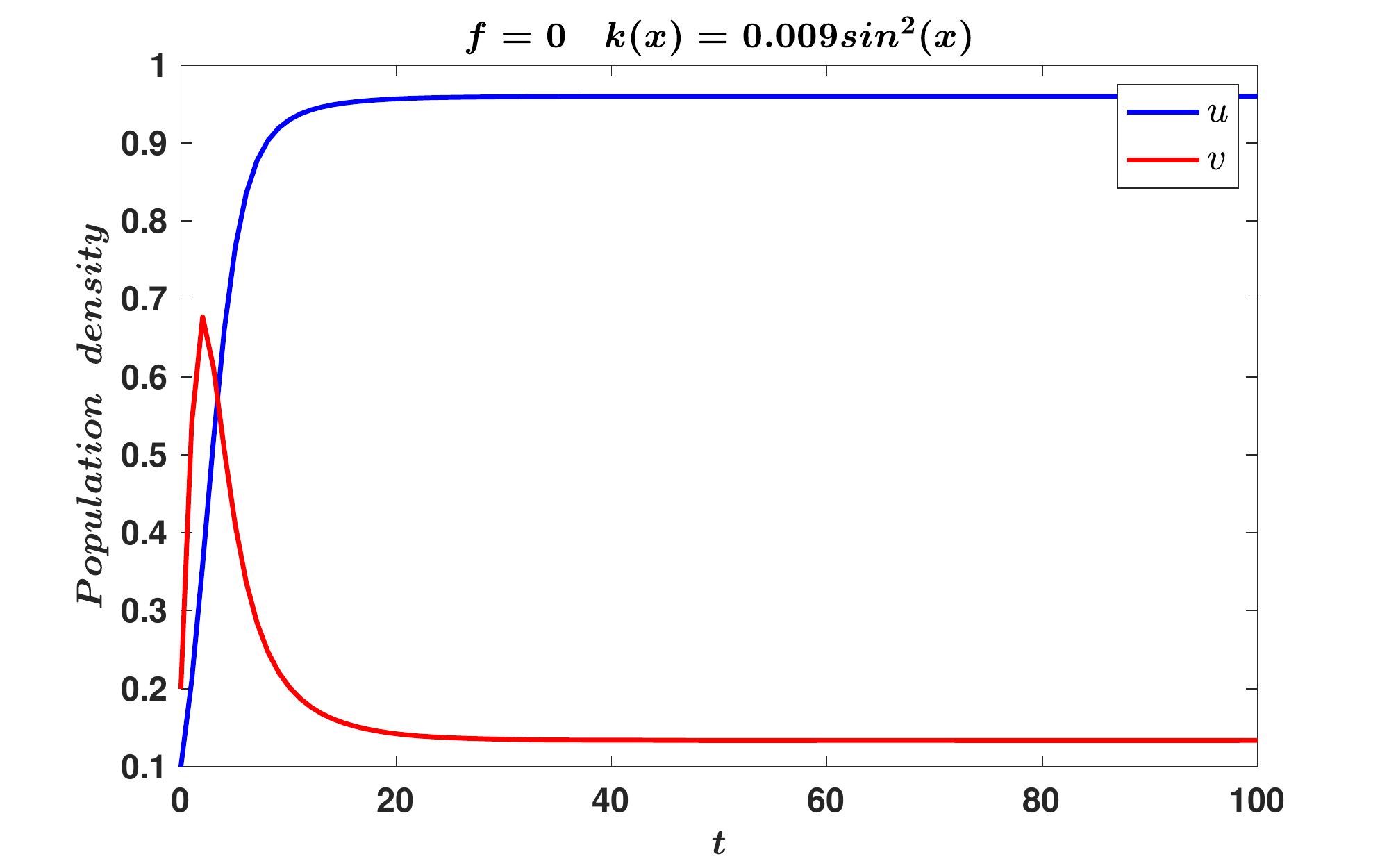}
		\caption{$[u_0,v_0]=[0.1,0.2]$}
	\end{subfigure}
	\newline
	\begin{subfigure}[b]{.475\linewidth}
		\includegraphics[width=\linewidth,height=1.8in]{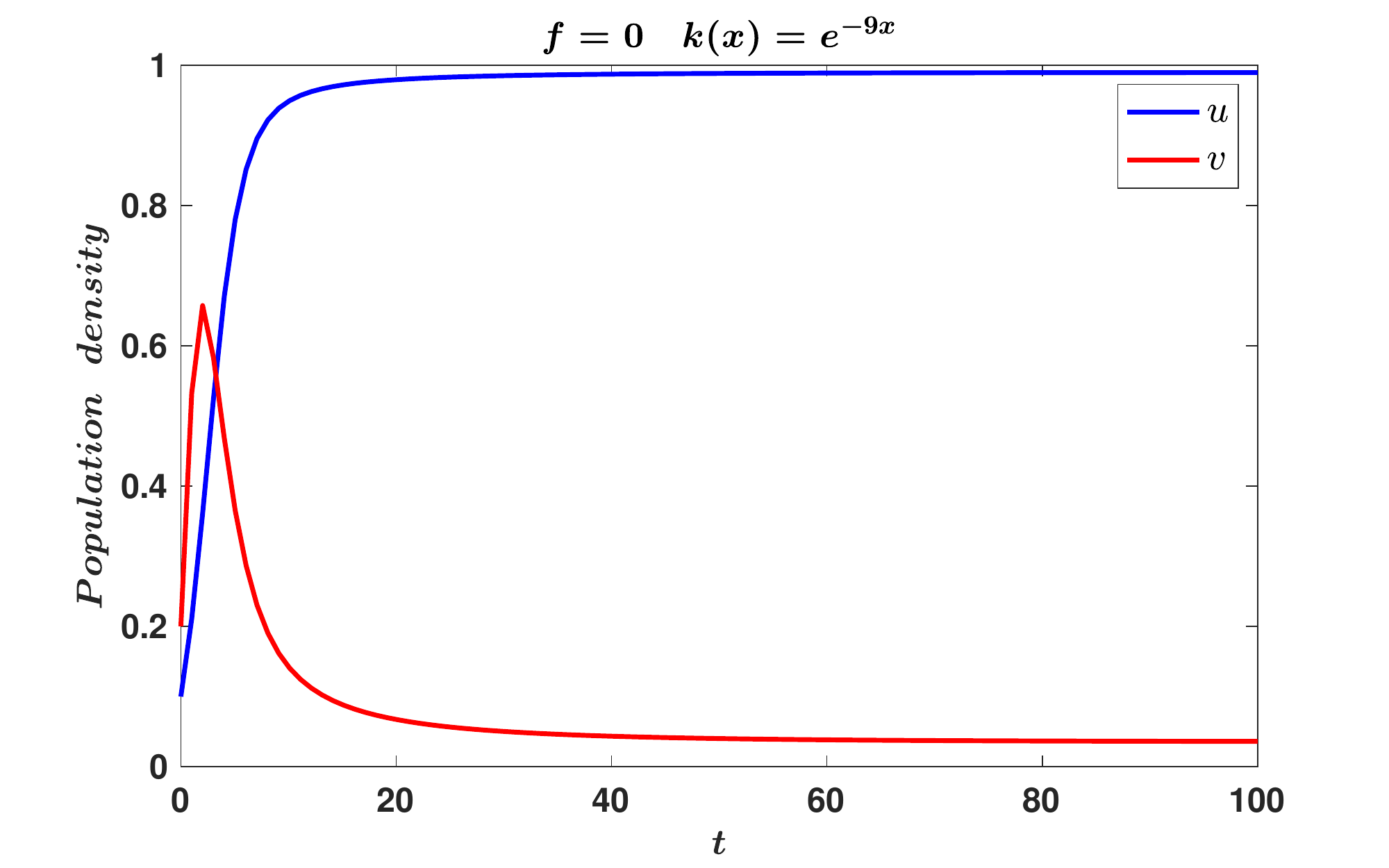}
		\caption{$[u_0,v_0]=[0.1,0.2]$}
	\end{subfigure}
	\hfill
	\begin{subfigure}[b]{.475\linewidth}
		\includegraphics[width=\linewidth,height=1.8in]{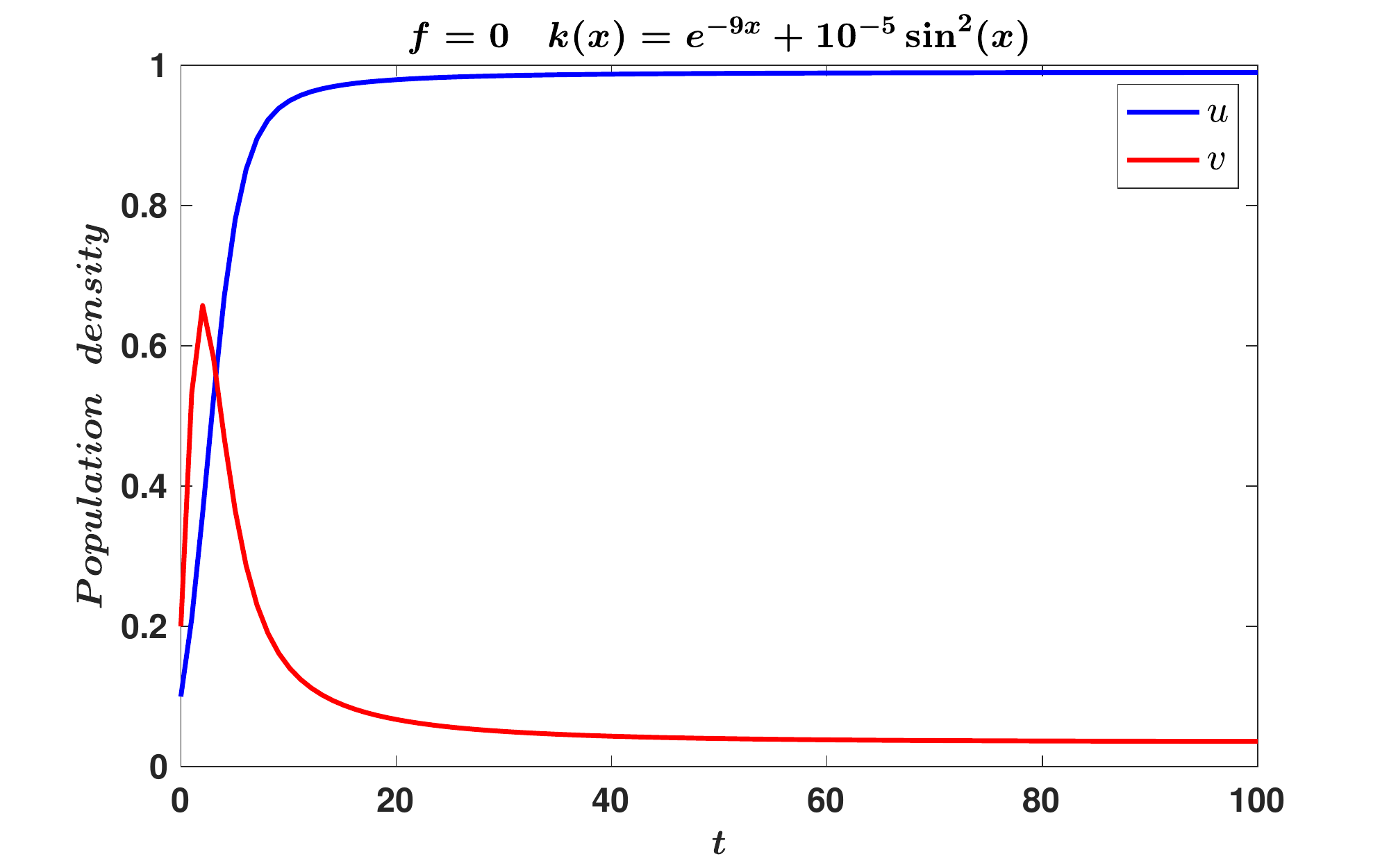}
		\caption{$[u_0,v_0]=[0.1,0.2]$}
	\end{subfigure}
	\caption{Numerical simulation of $(\ref{eq:PDE})$ for the case of weak-competition in $\Omega=[0,1]$. The parameters are chosen as $d_1=1,d_2=1,a_1=1,a_2=2,b_1=1,b_2=2,c_1=0.3$ and $c_2=1.8$.}
	\label{fig:pde_co1_pde1}
\end{figure}

The numerical simulations above motivate the following conjecture,
\begin{conjecture}\label{conj:co_1_pde3}
	Consider the reaction diffusion system \eqref{eq:PDE}, for a fear function $k(x)$, s.t the assumptions via \eqref{eq:as1} are met, with 
	\[ \mathbf{\widetilde{k}} < \frac{1}{a_{2}} \left(\frac{b_{1}b_{2}}{c_{1}} - c_{2}\right),\]
	and the parameters follow $(\ref{one_post_weak})$ and Theorem~$\ref{thm:exist}$, then the solution $(u,v)$ to \eqref{eq:PDE}  converges uniformly to the spatially homogenous state $(u^*,v^{*})$ as $t \to \infty$.

\end{conjecture}

\begin{remark}
We see via Lemma \ref{lem:ce_ode_2}, and Lemma \ref{lem:co_ode_1}, that in the ODE case, if we are in the weak competition setting without fear, then a critical amount of fear $k_{c}$ is \emph{both} sufficient and necessarily required to change the system's dynamics to a competitive exclusion type scenario. In the PDE case, where the fear function $k(x)$ can be spatially heterogeneous, this requirement is certainly sufficient, as seen via Theorem \ref{thm:cotoce_1_pde}, but \emph{not} necessary, in a point wise sense. 
\end{remark}

This result is stated next,


\begin{theorem}\label{thm:cotoce_1_pde}
	Consider the reaction diffusion system \eqref{eq:PDE}, with initial data $u_{0}(x) \geq 1$, s.t $u \nearrow u^{*}$, s.t. we are in the weak competition setting when $k(x) \equiv 0$. For a fear function $k(x)$, s.t the assumptions via \eqref{eq:as1} are met, the solution $(u,v) $ converges uniformly to the spatially homogeneous state $(\frac{a_{1}}{b_{1}},0)$ as $t \to \infty$, if the following condition holds,
	
	\begin{align}\label{eq:com1}
	\frac{C_{1}}{|\Omega|}\int_{\Omega} \frac{1}{1+k(x)} dx  <  \left(\frac{1}{1+k_{c} \left(\frac{a_{1}}{b_{1}}\right)} \right),
\end{align}
	
	where $k_{c}$ is as defined in Lemma \ref{lem:ce_ode_2}.
\end{theorem}

\begin{proof}
Consider \eqref{eq:PDE}, integrating the equation over $\Omega$ yields,

\begin{eqnarray}
&& \frac{d}{dt} \int_{\Omega} v dx \nonumber \\
&& = \int_{\Omega} \left(\dfrac{a_2 v}{1+k(x) u}  - b_2 v^2 - c_2 uv \right) dx\nonumber \\
&& \leq    \int_{\Omega} \left(\dfrac{a_2 v}{1 + k(x)u_{0}}  - b_2 v^2 - c_2 uv \right) dx \nonumber \\
&& =   \frac{1}{1+k(x^{*})u_{0}}\int_{\Omega} a_{2} v dx  - \int_{\Omega}  \left( b_2 v^2 + c_2 uv \right) dx \nonumber \\
&& \leq   \frac{1}{1+k(x^{*})}\int_{\Omega} a_{2} v dx  - \int_{\Omega}  \left( b_2 v^2 + c_2 uv \right) dx \nonumber \\
&&\leq   \left(C_{1}\int_{\Omega} \frac{1}{1+k(x)} dx \right) \int_{\Omega} a_{2} v dx  - \int_{\Omega}  \left( b_2 v^2 + c_2 uv \right) dx \nonumber \\
&& \leq   \left( \frac{1}{1 + k_{c}\frac{a_{1}}{b_{1}}} \right)\int_{\Omega} a_{2}v dx  - \int_{\Omega}  \left( b_2 v^2 + c_2 uv \right) dx \nonumber \\
&&\leq  \int_{\Omega} \left(\frac{a_{2}}{1+k_{c}\frac{a_{1}}{b_{1}}}v  - b_2 v^2 - c_2 uv \right) dx. \nonumber \\
\end{eqnarray}

This follows via the mean value theorem for integrals. We can now compare,

\begin{align}\label{eq:com133}
	\int_{\Omega} v dx  < \int_{\Omega} \tilde{v} dx,
\end{align}
where $\tilde{v}$ solves $\frac{d \tilde{v}}{dt} = \left(\frac{a_{2}}{1+k_{c}\frac{a_{1}}{b_{1}}}\tilde{v}  - b_2 \tilde{v}^2 - c_2 u\tilde{v} \right).$
 Thus using Theorem \ref{thm:km1}, Lemma \ref{lem:cotce}
and the positivity of solutions, this entails
\begin{align}\label{eq:com134}
0 \leq \lim_{t \rightarrow \infty}   \int_{\Omega} v dx     \leq \lim_{t \rightarrow \infty} 	\int_{\Omega} \tilde{v} dx = 0.
\end{align}

The bounds on $v$, standard Lebesgue convergence theorems, and squeezing argument entail,

\begin{align}\label{eq:com133}
	\int_{\Omega} (\lim_{t \rightarrow \infty}  v ) dx  \rightarrow 0 ,
\end{align}

which implies the uniform convergence,
\begin{align*}
\lim_{t \rightarrow \infty} (u,v) \rightarrow \left(\frac{a_{1}}{b_{1}}, 0\right).
\end{align*}

\end{proof}

\begin{remark}
The $C_{1}$ in Theorem \ref{thm:cotoce_1_pde} is a pure constant, that will depend on the size of the domain $\Omega$, and the other problem parameters, but not on the initial data or the spatial variable $x$.
\end{remark}

\begin{remark}
Clearly $k(x)$ could be chosen s.t. it lies below $k_{c}$ for a portion of the domain, and above $k_{c}$ on some portion of the domain - thus $k_{min}$ does not lie uniformly above $k_{c}$. Yet via Theorem \ref{thm:cotoce_1_pde}, we see that, one can change the system's dynamics and bring it to a competitive exclusion type scenario, from a coexistence situation.
\end{remark}

The next result gives a lower estimate for the fear function $k$,

\begin{lemma}
\label{lem:ff1}
Consider the fear function $k(x)$ in Theorem \ref{thm:cotoce_1_pde}, then we have,

\begin{align*}
	\int_{\Omega} k(x) dx  \geq  |\Omega|\left(1 - \frac{1}{C_{1}}\left(\frac{1}{1+k_{c} \left(\frac{a_{1}}{b_{1}}\right)} \right) \right).
\end{align*}

\end{lemma}

\begin{proof}
We have that for any fear function $k(x)$ satisfying \eqref{eq:as1},

\begin{align}\label{eq:com1}
	\frac{1}{1+k(x)} > 1 - k(x)  \implies \int_{\Omega} \frac{1}{1+k(x)} dx > |\Omega| - \int_{\Omega} k(x) dx.
\end{align}

Now using Theorem \ref{thm:cotoce_1_pde} the result follows.
\end{proof}

\subsection{The Strong Competition Case}

\begin{theorem}\label{st_pde}
	Consider the reaction diffusion system \eqref{eq:PDE}, for a fear function $k(x)$, s.t the assumptions via \eqref{eq:as1} are met, with $b_{1}b_{2} < c_{1}c_{2}$. Then there exists sufficiently large positive initial data $[u_0,v_0]$ for which the solution $(u,v)$ to \eqref{eq:PDE}  converges uniformly to the spatially homogenous state $(0,\frac{a_{2}}{b_{2}})$ as $t \to \infty$, while there also exists sufficiently small positive initial data $[u_1,v_1]$ for which solution $(u,v)$ to \eqref{eq:PDE}  converges uniformly to the spatially homogenous state $(\frac{a_{1}}{b_{1}},0)$ as $t \to \infty$.
\end{theorem}

\begin{proof}
Consider the system \eqref{eq:lv_model}. From the classical strong competition parametric restrictions, $b_{1}b_{2} < c_{1}c_{2}$ that are assumed, as well as Lemma \ref{lem:com1}, we can make use of standard competition theory and use the stable manifold theorem, i.e., $\exists W_{s}(E_{4}) \in C^{1}$ separatrix, such that for initial data $(\overline{u}_0,\overline{v}_0)$ chosen above $W_{s}(E_{4})$ the solution $(\overline{u},\overline{v}) \to (0,v^*)$ and for initial data chosen below $W_{s}(E_{4})$, $(\overline{u},\overline{v}) \to (u^*,0)$. Here $E_{4}$ is the interior saddle equilibrium to the kinetic (ODE) system for \eqref{eq:lv_model}.
	Moreover, since $\dfrac{a_{2}}{1 + \mathbf{\widetilde{k}}  \frac{a_{1}}{b_{1}}} \le a_{2}$, and $b_{1}b_{2} < c_{1}c_{2}$, we have that for \eqref{eq:lowest1} we still remain in the strong competition case, and via standard theory again, $W^{1}_{s}(E^{*}_{4}) \in C^{1}$ separatrix, such that for initial data $(\widehat{u}_0,\widehat{v}_0)$ chosen above $W^{1}_{s}(E^{*}_{4})$ the solution $(\widehat{u},\widehat{v}) \to (0,v^*)$ and for initial data chosen below $W^{1}_{s}(E_{4})$, $(\widehat{u},\widehat{v})\to (u^*,0)$. Here $E^{*}_{4}$ is the interior saddle equilibrium to the kinetic (ODE) system for \eqref{eq:lowest1}.
Now since $\dfrac{a_{2}}{1 + \mathbf{\widetilde{k}}  \frac{a_{1}}{b_{1}}} \le a_{2}$, the $v$ component of $E^{*}_{4}$ is higher than the $v$ component of $E_{4}$. Now using the $C^{1}$ property of the separatricies $W^{1}_{s}(E^{*}_{4}) , W_{s}(E_{4})$, we have the existence of a wedge $\nu$ emanating from $E_{4}$, s.t within $\nu$ we have $W^{1}_{s}(E^{*}_{4}) \geq W_{s}(E_{4})$. Note via Lemma \ref{lem:com1} we have $ \tilde{v} \leq v \leq \widehat{v}$.	Let us consider positive initial data $(u_0,v_0)$ chosen large enough, within $\nu$ s.t. $(u_0,v_0) >> W^{1}_{s}(E^{*}_{4}) > W_{s}(E_{4})$, we will have 
	
	\begin{align*}
	\Big\{  (0,v^*) \Big\} \le \Big\{ (u,v) \Big\} \le \Big\{ (0,v^*) \Big\}.
	\end{align*}
	
	On the other hand,  for positive initial data $(u_1,v_1)$ chosen small enough s.t. $(u_1,v_1) << W^{1}_{s}(E_{4}) < W_{s}(E_{4})$, we will have 
	\begin{align*}
	\Big\{  (u^*,0)\Big\} \le \Big\{ (u,v) \Big\} \le \Big\{ (u^*,0) \Big\}.
	\end{align*}
	
	This proves the theorem.

\end{proof}

%

\begin{figure}[h]
	\begin{subfigure}[b]{.475\linewidth}
		\includegraphics[width=\linewidth,height=1.8in]{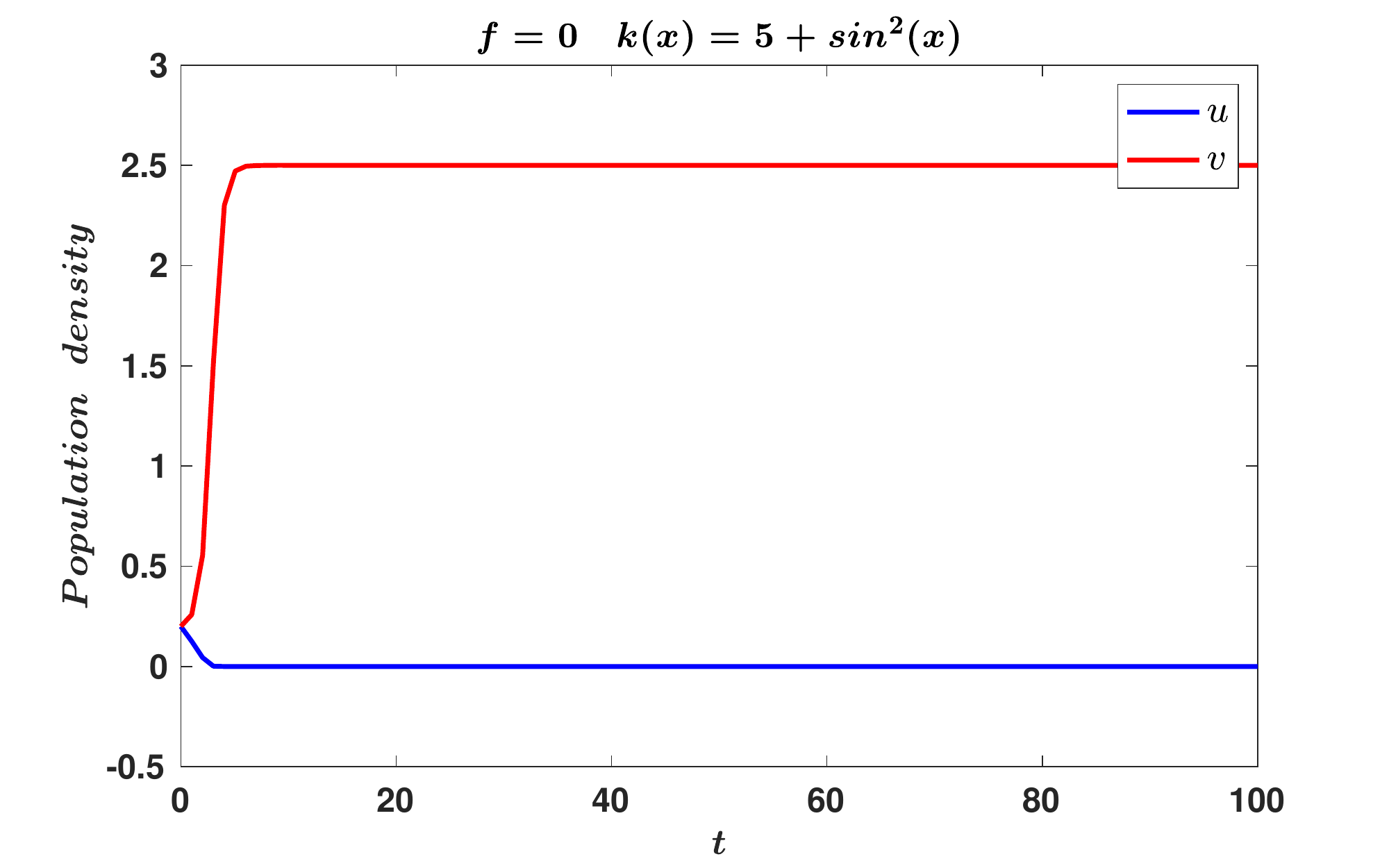}
		\caption{$[u_0,v_0]=[0.2,0.2].$ }
	\end{subfigure}
	\hfill
	\begin{subfigure}[b]{.475\linewidth}
		\includegraphics[width=\linewidth,height=1.8in]{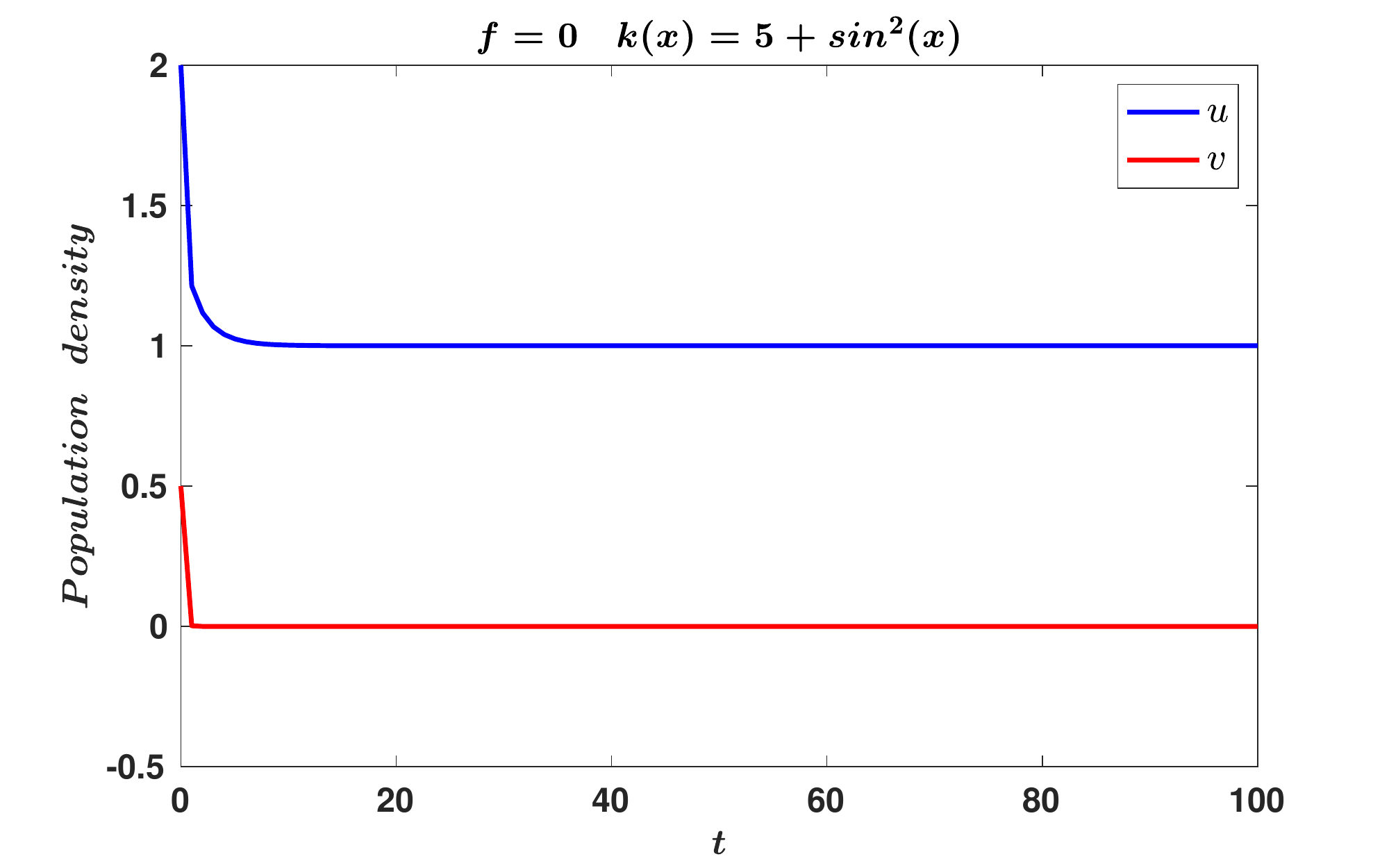}
		\caption{$[u_0,v_0]=[2,0.5].$}
	\end{subfigure}
	\newline
	\begin{subfigure}[b]{.475\linewidth}
		\includegraphics[width=\linewidth,height=1.8in]{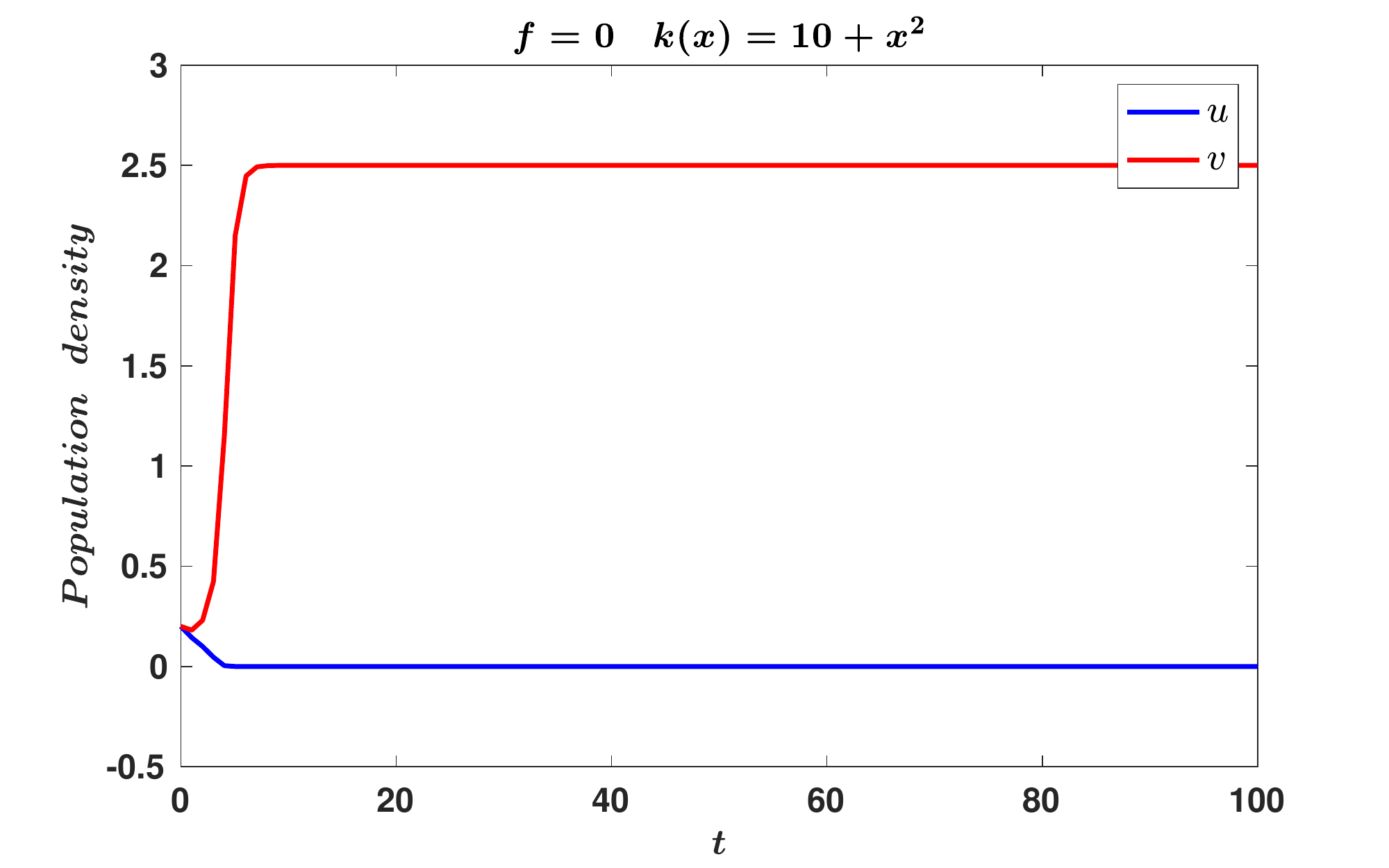}
		\caption{$[u_0,v_0]=[0.2,0.2].$ }
	\end{subfigure}
	\hfill
	\begin{subfigure}[b]{.475\linewidth}
		\includegraphics[width=\linewidth,height=1.8in]{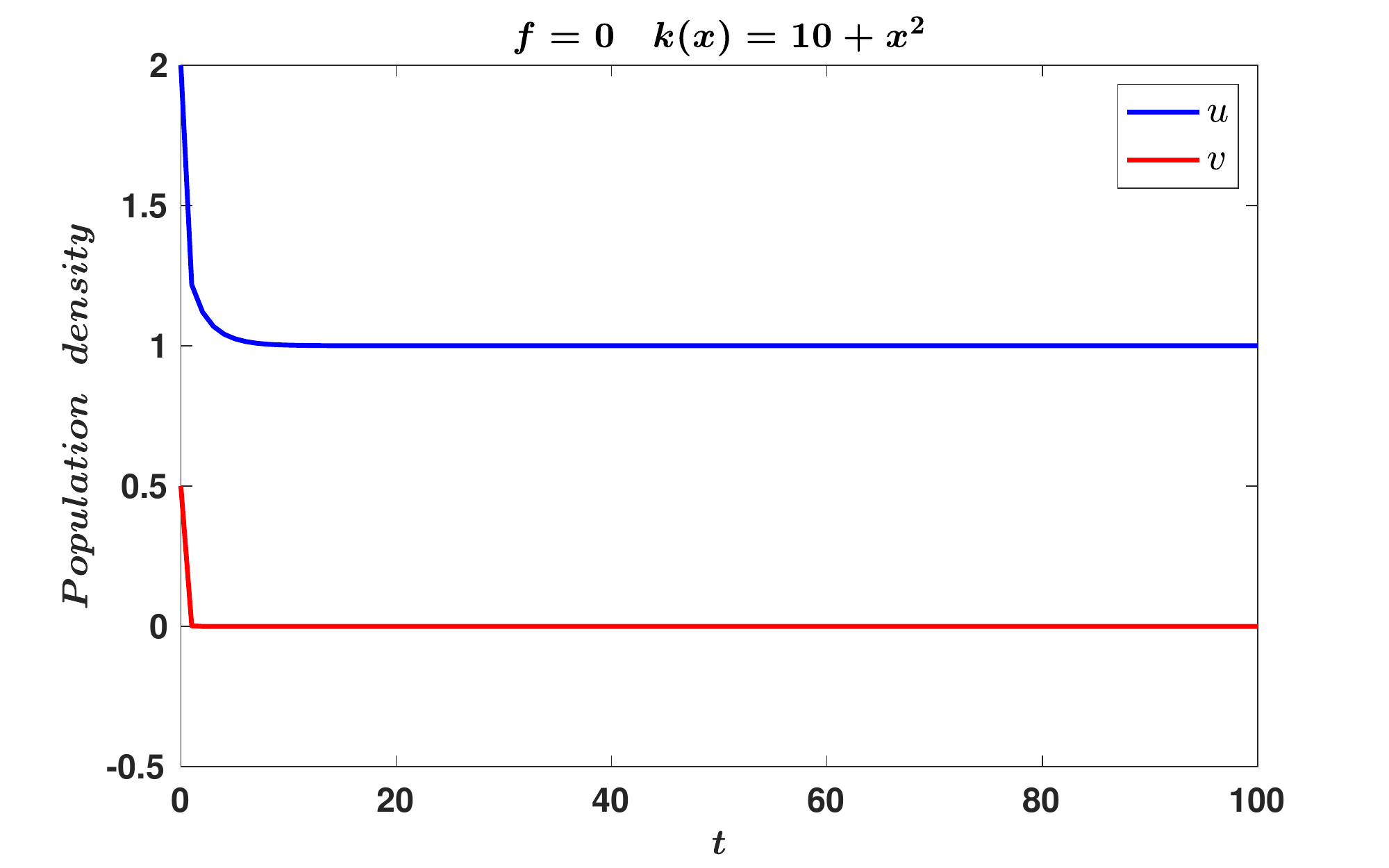}
		\caption{$[u_0,v_0]=[2,0.5].$}
	\end{subfigure}
	\caption{Numerical simulation of $(\ref{eq:PDE})$ for the case of strong-competition in $\Omega=[0,1]$. The parameters are chosen as $d_1=1,d_2=1,a_1=0.5,a_2=2,b_1=0.5,b_2=0.8$ and $c_1=c_2=4$.}
	\label{fig:pde_sc1_pde1}
\end{figure}

\begin{figure}[h]
	\begin{subfigure}[b]{.475\linewidth}
		\includegraphics[width=\linewidth,height=1.8in]{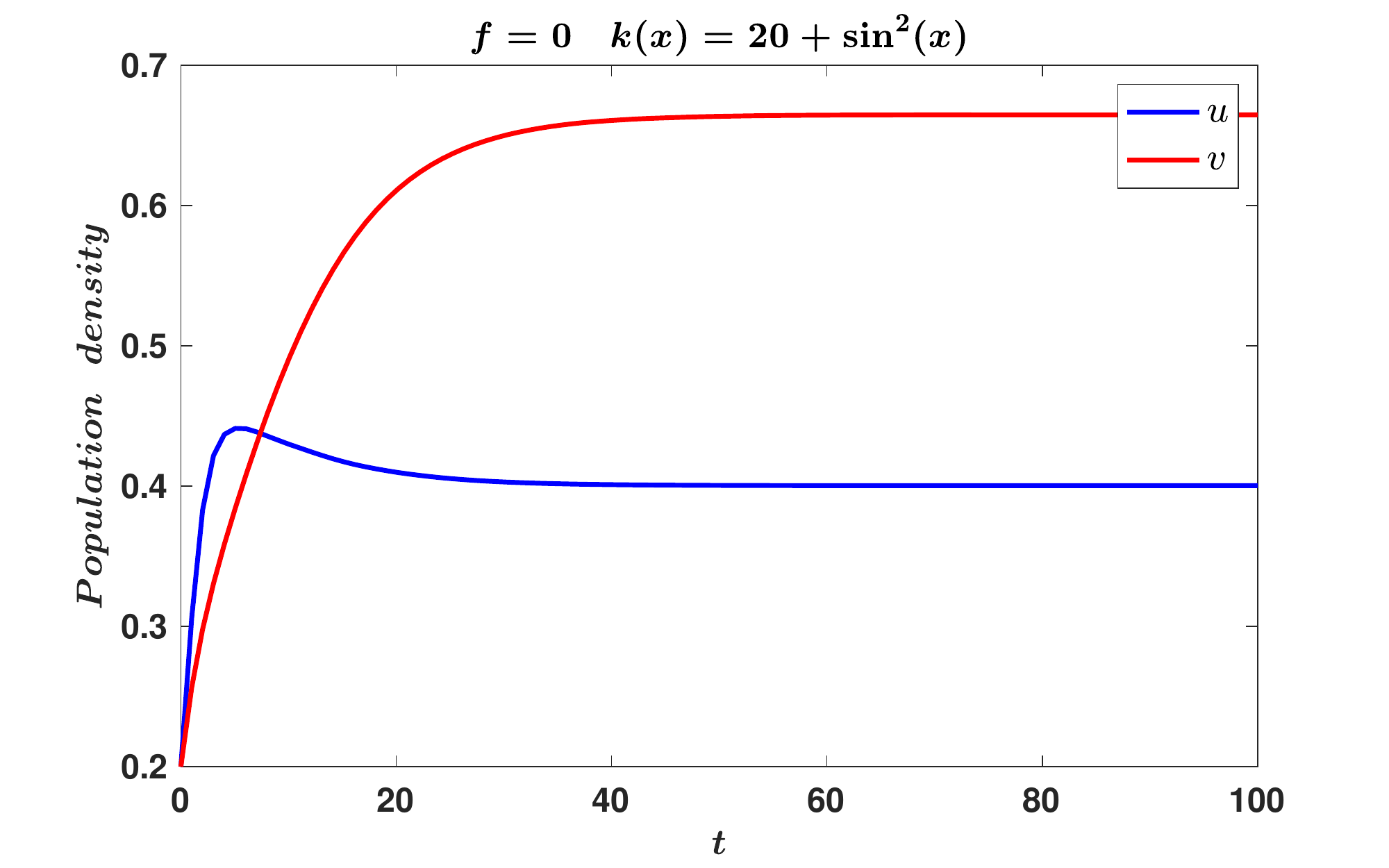}
		\caption{$[u_0,v_0]=[0.2,0.2].$ }
	\end{subfigure}
	\hfill
	\begin{subfigure}[b]{.475\linewidth}
		\includegraphics[width=\linewidth,height=1.8in]{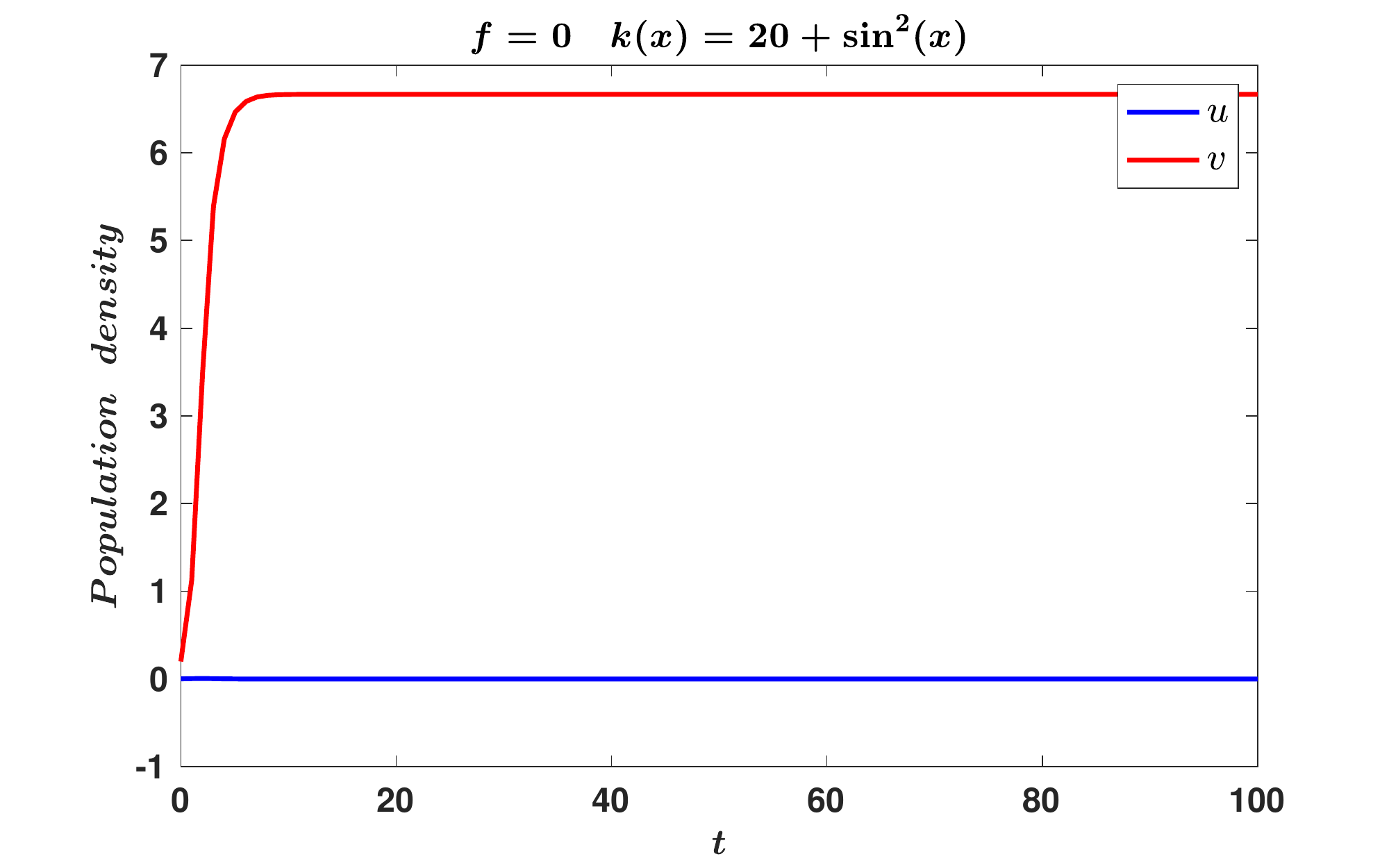}
		\caption{$[u_0,v_0]=[0.002,0.2].$}
	\end{subfigure}
	\caption{Numerical simulation of $(\ref{eq:PDE})$ for the case of two-positive interior equilibria in $\Omega=[0,1]$. The parameters are chosen as $d_1=1,d_2=1,a_1=1,a_2=2,b_1=2,b_2=0.3$ and $c_1=0.3,c_2=0.05$. }
	\label{fig:pde_two_post_pde1}
\end{figure}

\begin{figure}[h]
	\begin{subfigure}[b]{.475\linewidth}
		\includegraphics[width=\linewidth,height=1.8in]{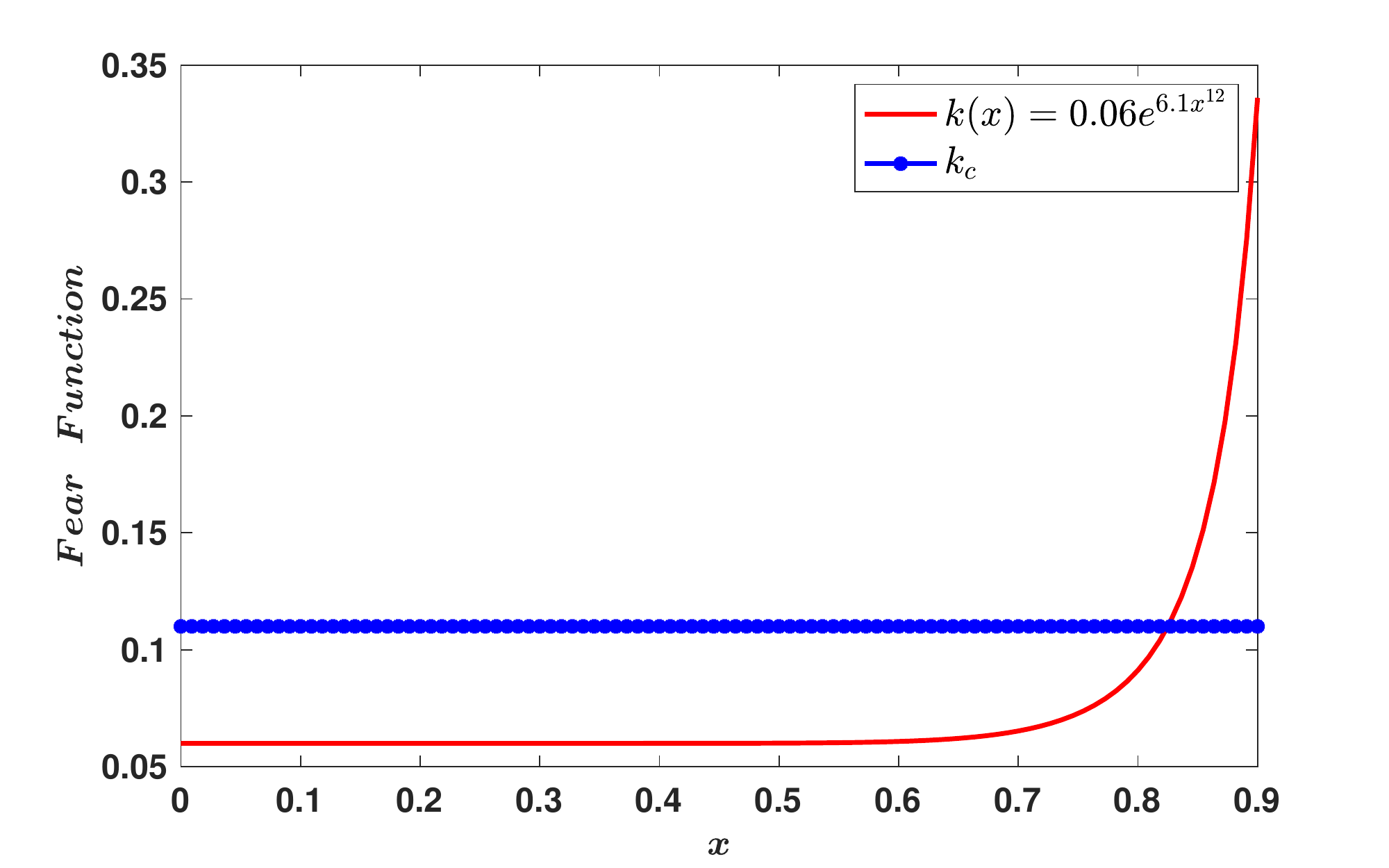}
	\end{subfigure}
	\hfill
	\begin{subfigure}[b]{.475\linewidth}
		\includegraphics[width=\linewidth,height=1.8in]{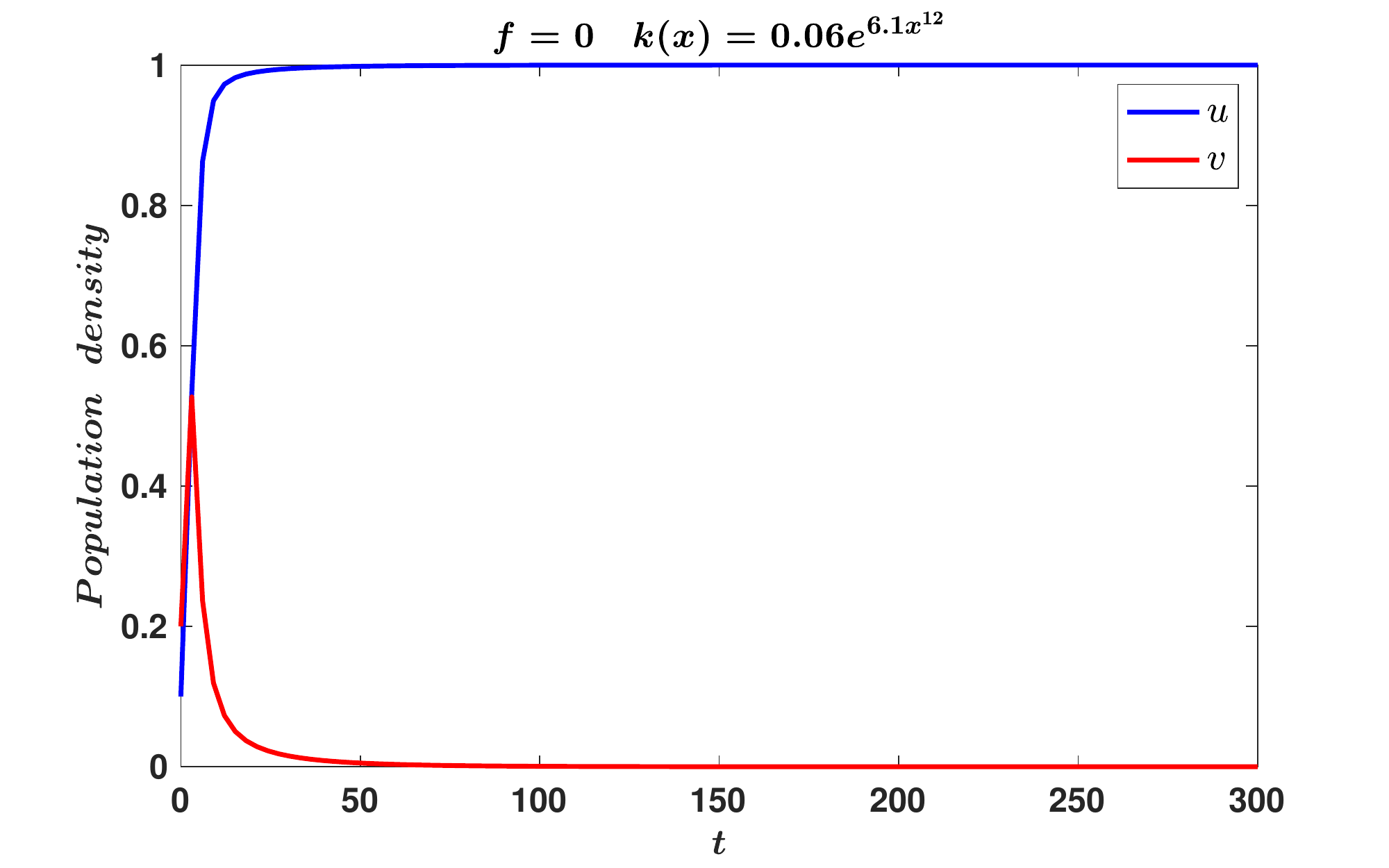}
	\end{subfigure}
	\caption{Numerical simulation for the case  $(v^*<u^*)$. The parameters are chosen as $[u_0,v_0]=[0.1,0.2],d_1=1,d_2=1,a_1=1,a_2=2,b_1=1,b_2=2,c_1=0.3$ and $c_2=1.8$. The interior equilibrium $[u^*,v^*]=[0.958904,0.136986]$ and the fear threshold $k_c=\frac{ a_2b_1^2-c_2 a_1 b_1}{a_1^2c_2}=0.11$ and $C_{1}=1$.}
	\label{fig:remark_5}
\end{figure}

\begin{figure}[h]
	\begin{subfigure}[b]{.475\linewidth}
		\includegraphics[width=\linewidth,height=1.8in]{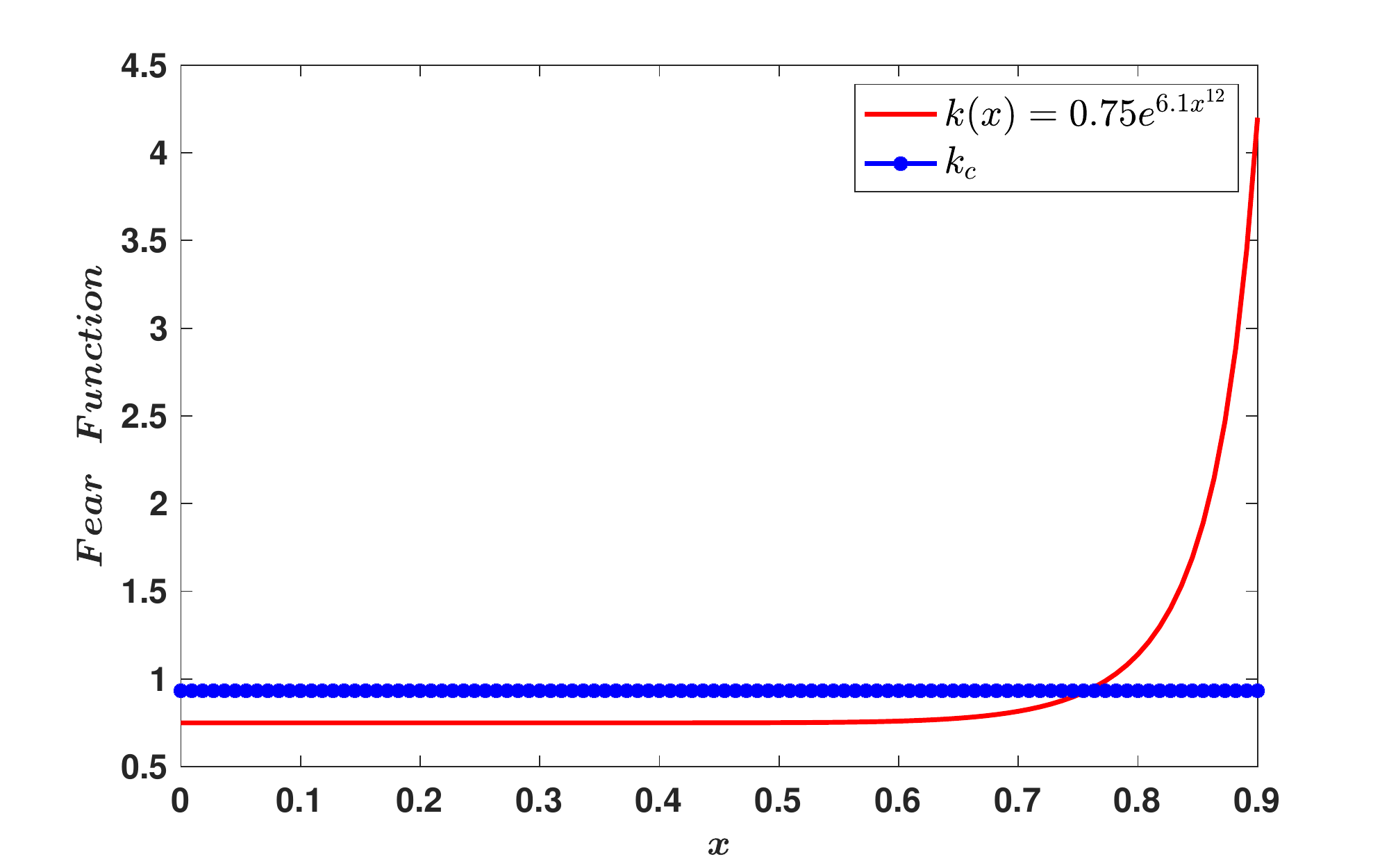}
	\end{subfigure}
	\hfill
	\begin{subfigure}[b]{.475\linewidth}
		\includegraphics[width=\linewidth,height=1.8in]{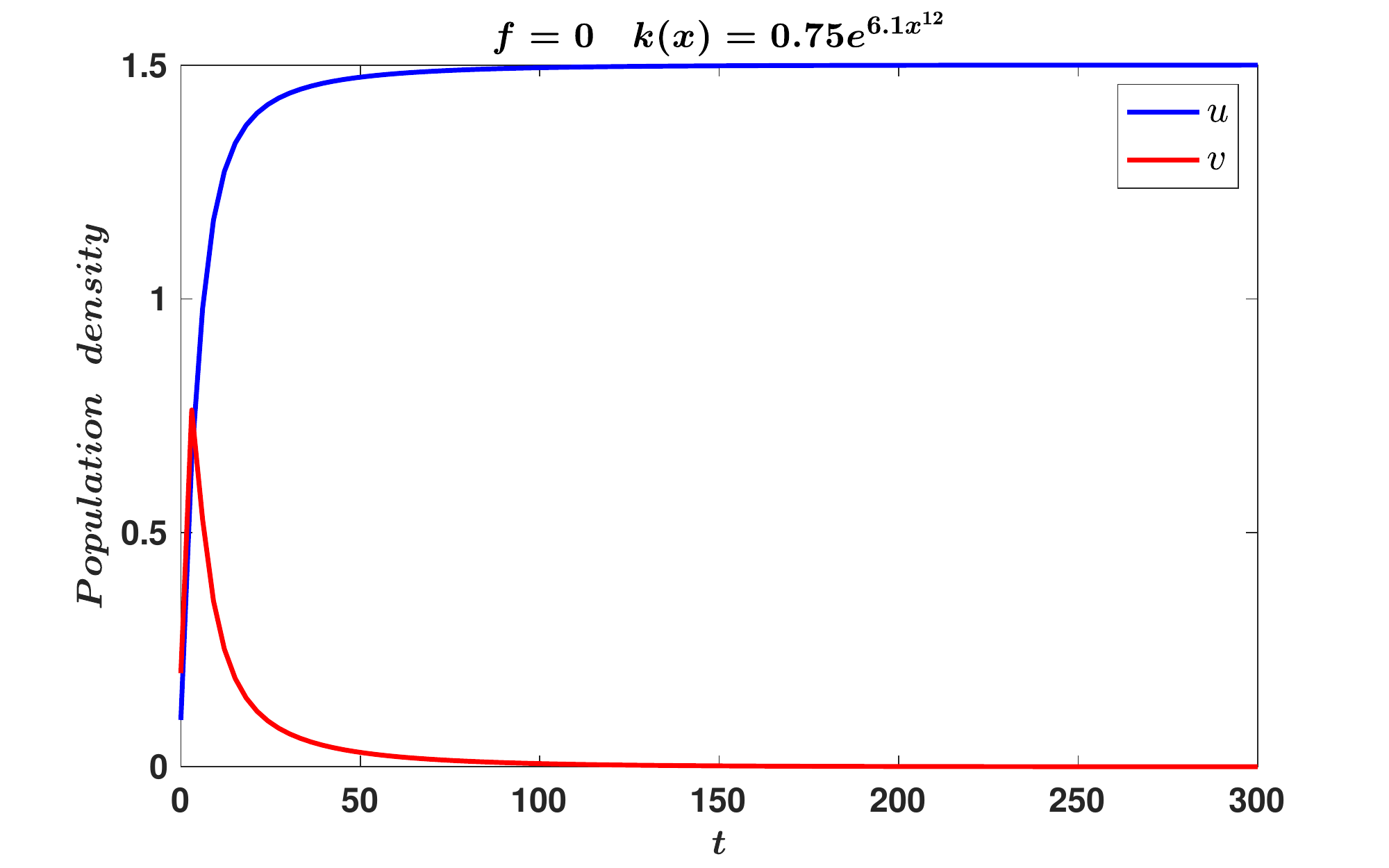}
	\end{subfigure}
	\caption{Numerical simulation for the case $(u^*<v^*)$. The parameters are chosen as $[u_0,v_0]=[0.1,0.2],d_1=1,d_2=1,a_1=1.5,a_2=1.8,b_1=1,b_2=1,c_1=0.8$ and $c_2=0.5$. The interior equilibrium $[u^*,v^*]=[0.1,1.75]$ and the fear threshold $k_c=\frac{ a_2b_1^2-c_2 a_1 b_1}{a_1^2c_2}=0.933$ and $C_{1}=0.8$.}
	\label{fig:remark_51}
\end{figure}

\begin{figure}[h]
	\begin{subfigure}[b]{.475\linewidth}
		\includegraphics[width=\linewidth,height=1.8in]{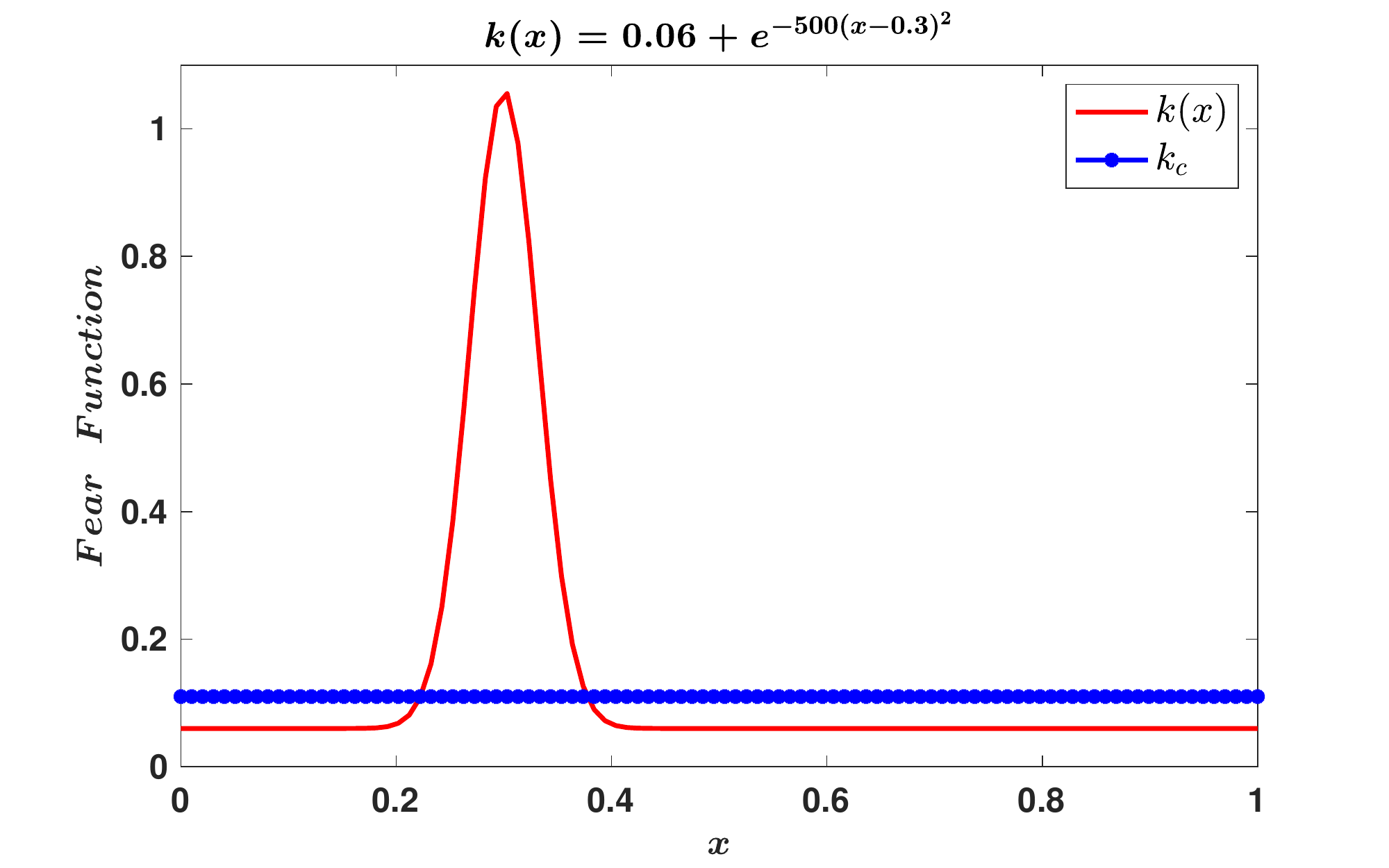}
	\end{subfigure}
	\hfill
	\begin{subfigure}[b]{.475\linewidth}
		\includegraphics[width=\linewidth,height=1.8in]{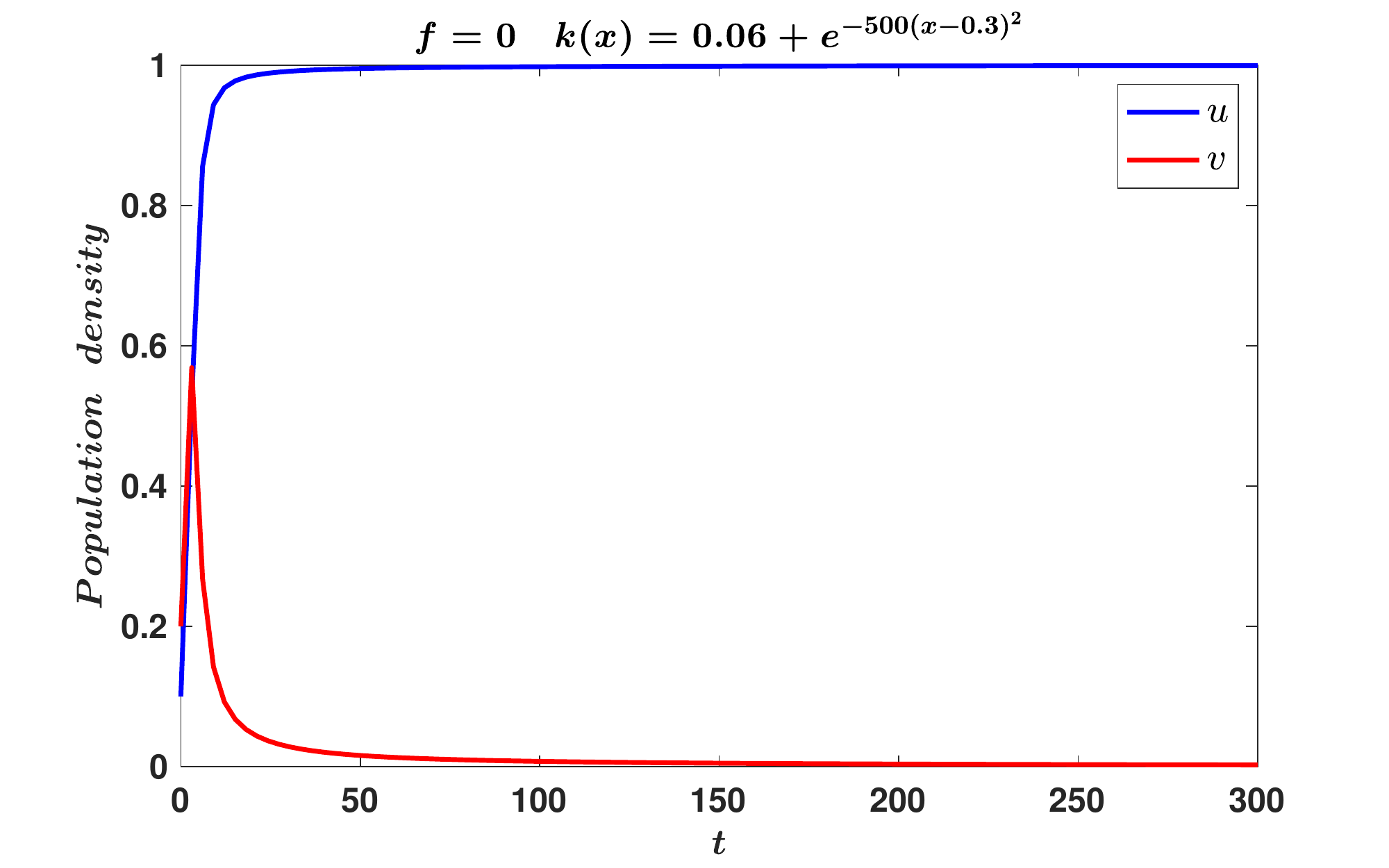}
	\end{subfigure}
	\newline
	\begin{subfigure}[b]{.475\linewidth}
		\includegraphics[width=\linewidth,height=1.8in]{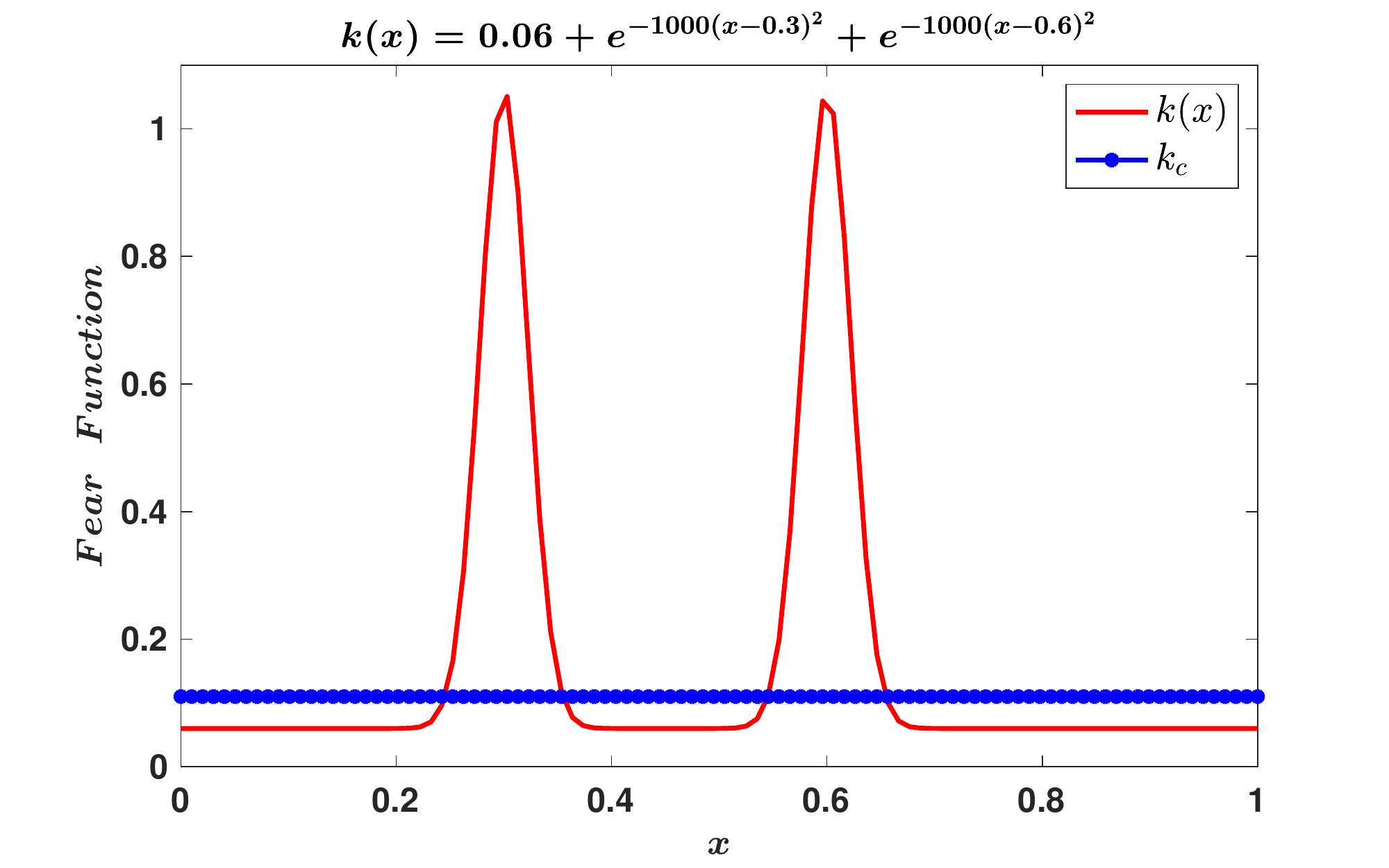}
	\end{subfigure}
	\hfill
	\begin{subfigure}[b]{.475\linewidth}
		\includegraphics[width=\linewidth,height=1.8in]{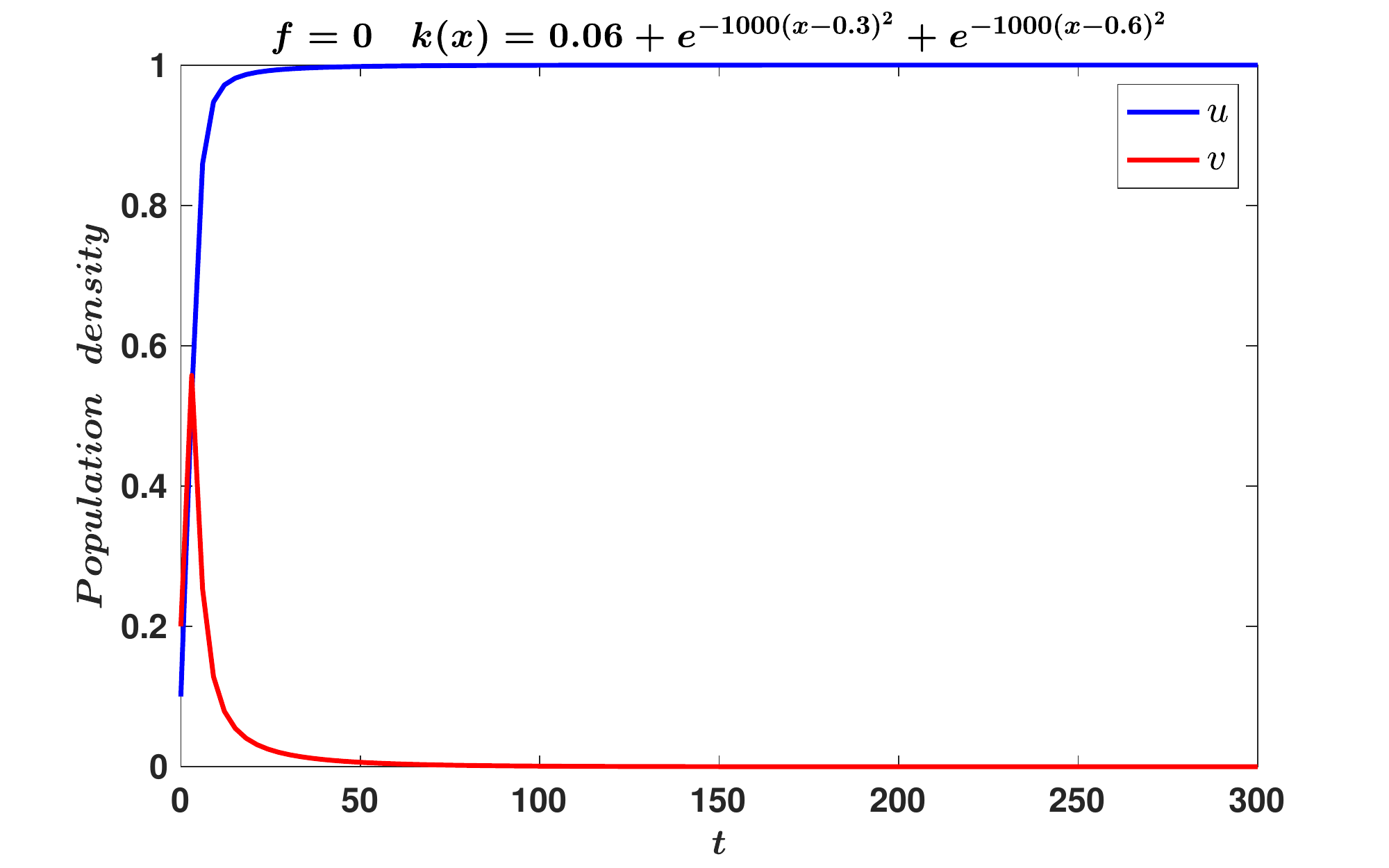}
	\end{subfigure}
	\caption{The parameters are chosen as $[u_0,v_0]=[0.1,0.2],d_1=1,d_2=1,a_1=1,a_2=2,b_1=1,b_2=2,c_1=0.3$ and $c_2=1.8$. The interior equilibrium $[u^*,v^*]=[0.958904,0.136986]$ and the fear threshold $k_c=\frac{ a_2b_1^2-c_2 a_1 b_1}{a_1^2c_2}=0.11$ and $C_{1}=1$.}
	\label{fig:remark_52}
\end{figure}

From the observations in the above numerical simulations, we can state a conjecture concerning positive interior equilibrium,

\begin{conjecture}
\label{conj:pde_st1}
	Consider the system $\eqref{eq:PDE}$. For any non-negative fear function $k(x)$ such that
	\begin{align}
		\Big( \dfrac{b_2b_1}{c_1} -c_2\Big) <\frac{k a_2b_1^2}{(b_1+ka_1)^2}
	\end{align}
	and parametric restrictions given by \eqref{eq:as1} and Theorem~\ref{thm:exist} hold true, for $\mathbf{k}=\mathbf{\widetilde{k}} ,\mathbf{\hat{k}}$. There exists some data $[u_0,v_0]$, for which the solution $(u,v) \to (u^*,v^*)$, and for some choice of data $[u_1,v_1]$, the solution converges to the boundary equilibrium $(0,v^*)$.
\end{conjecture}


\section{Discussion}
\label{disc}
In this section we discuss several aspects of the lemmas and theorems described in the current manuscript.

\begin{figure}[h]
	\begin{subfigure}[b]{.49\linewidth}
		\includegraphics[width=\linewidth,height=1.9in]{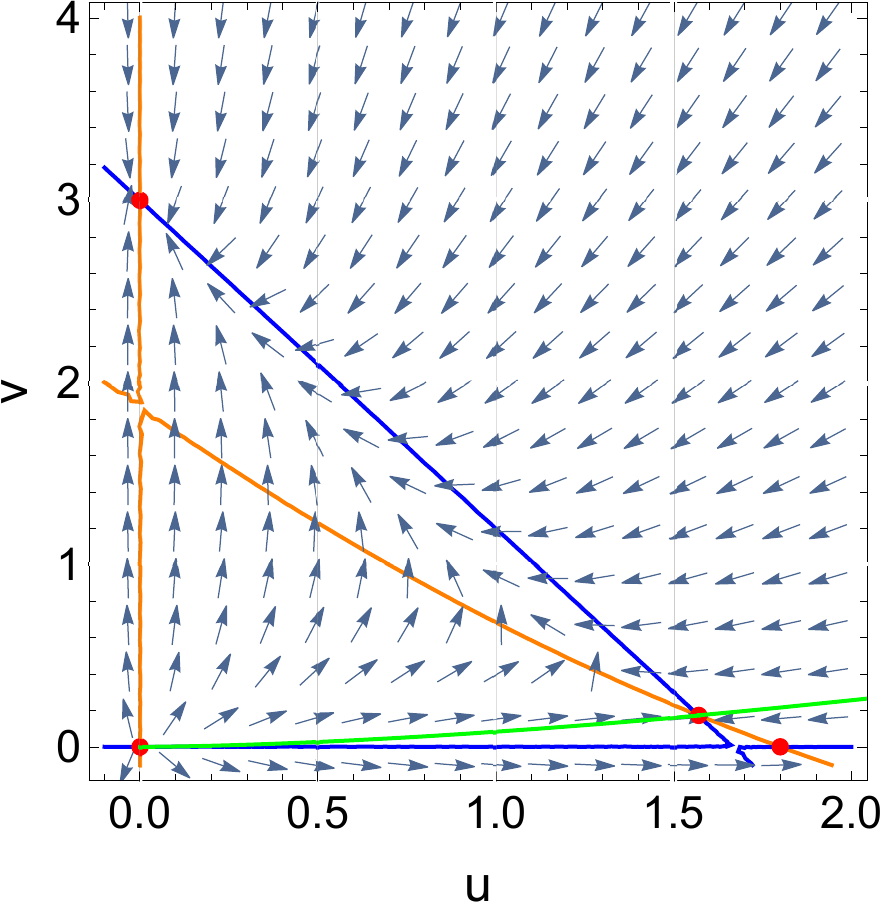}
		\caption{$f=0.5$}
		\label{fig:dis_st}
		\end{subfigure}
	\hfill
	\begin{subfigure}[b]{.5\linewidth}
		\includegraphics[width=\linewidth,height=2in]{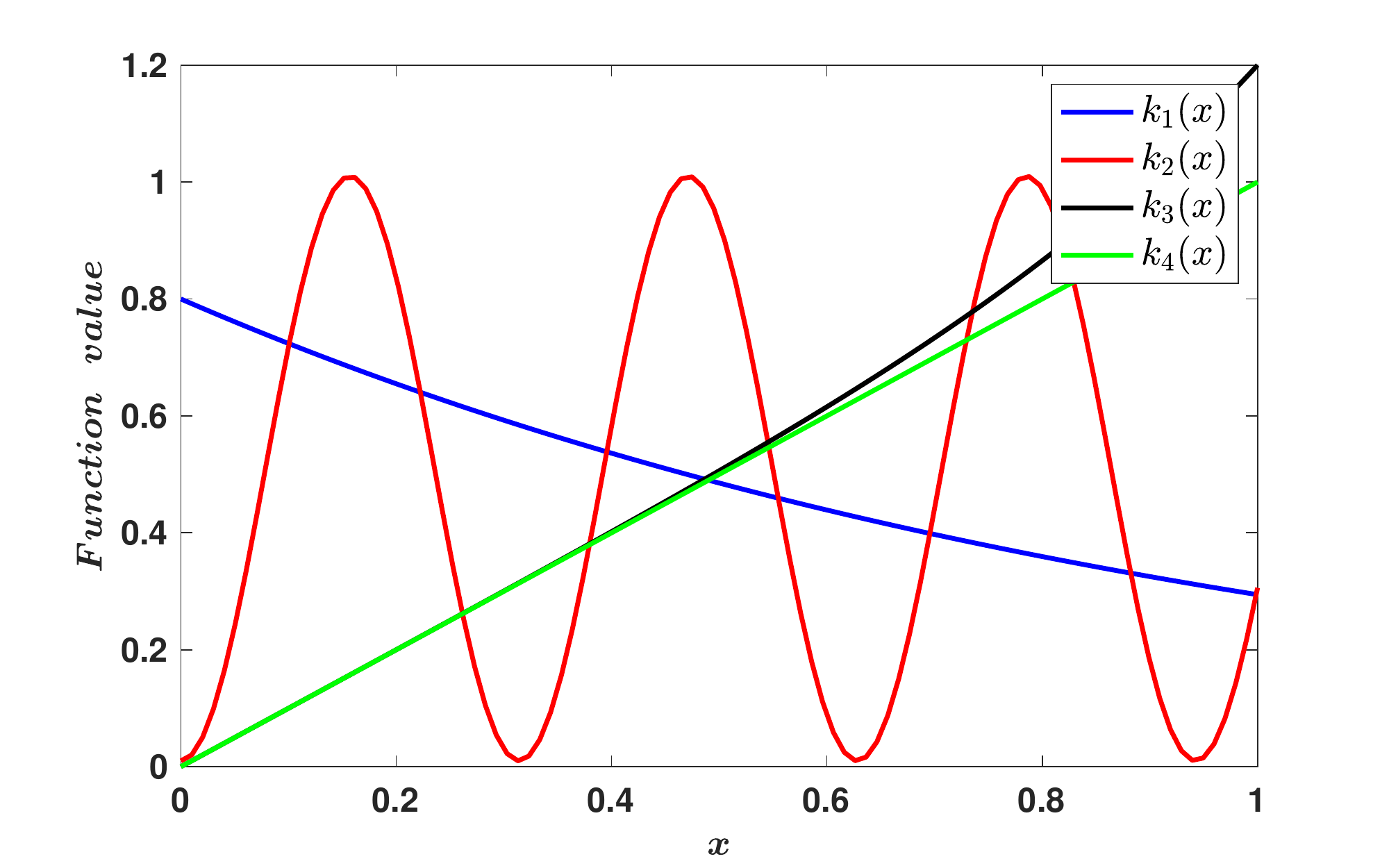}
		\caption{Various heterogeneous fear functions}
		\label{fig:fear_plot}
	\end{subfigure}
	\caption{(A) Phase plots showing dynamics under strong competition parametric restriction with $k=0$. Parameters used are $a_1=1, a_2=1,b_1=1, b_2=1, c_1=2, c_2=2$ and $f=0.5$ (B) Illustration of various heterogeneous fear functions. $k_1(x)=0.8 e^{-x},k_2 (x)=0.01+\sin^2 (10x),k_3 (x)=x+0.2x^5$ and $k_4(x)=x$}
\end{figure}	

We see that in the competitive exclusion case, where $(u^{*},0)$ is globally stable, a small amount of fear in $u$, can bring about a strong competition type situation - via a transcritical bifurcation. Where after bifurcation there appears an interior equilibrium, which is a saddle, and one has initial condition dependent attraction to $(u^{*},0)$ or $(0,v^{*})$. See Figs.~[\ref{fig:ode_two_postive},\ref{fig:Fig1},\ref{fig:trans}]. However, a slightly larger level of fear, can bring down the separatrix drastically see Fig.~\ref{fig:dis_st}. Furthermore, here (depending on parametric restrictions) only a very small quantity of fear can create two interior equilibrium, see Figs.~[\ref{fig:Fig1},\ref{fig:trans}], where there are both weak and strong competition dynamics at play, where most initial conditions are attracted to an interior equilibrium, but certain initial conditions $(u_{0},v_{0})$ (where $u_{0} << 1, v_{0} >> 1$), are attracted to $(0,v^{*})$. It is interesting to think about this in an invasion setting. If $u$ were an invasive species, and $v$ a resident species, then in the absence of fear in $u$, it can invade and thus exclude $v$. However, if the resident $v$ can instill just a small amount of fear in the invader $u$, coexistence for most initial conditions is possible. If one plays devil's advocate, and switches the role of $u$ and $v$, ($u$ resident, $v$ invader), then we see in the absence of fear in the resident, it will exclude the invader, but if it can instill some fear, then coexistence is possible, as earlier alluded to. However, having too much fear in either case, yields only one interior equilibrium, which is a saddle and we are in a strong competition type setting, with initial condition dependent attraction to either $(u^{*},0)$ or $(0,v^{*})$. Thus if coexistence (for most initial conditions) is sought after, it is advantageous to induce a little fear  - but not to much fear. Dynamically, this is seen because in the former case we have a saddle-node bifurcation occurring first, followed by a transcritical bifurcation, see Fig.~\ref{fig:transB}, as opposed to the later case, where only a transcritical bifurcation occurs, see Fig.~\ref{fig:transA}.

This phenomenon becomes even more interesting in sociological or political settings, particularly when we are in the spatially explicit scenario. Consider the case of law enforcement trying to control crime in high crime areas, where traditionally control efforts have failed, indicative of the competitive exclusion scenario, or a coexistence scenario, where the criminal groups have not been managed so coexist at high levels, and have not been eradicated. The question becomes, what is the optimal design of the fear function $k(x)$, that could now yield a co-existence scenario - or one that can yield a competitive exclusion scenario. We see via Theorem \ref{thm:cotoce_1_pde}, that this function could be very small in some areas of the spatial domain, and large enough in others, so that it would change dynamics - this is also seen numerically, where ``hotspot" type fear functions have been utilised to generate a competitive exclusion type scenario, see Figs.~[\ref{fig:remark_5},\ref{fig:remark_51},\ref{fig:remark_52}]. Similar ideas for control and policing activities have been explored in \cite{Rod21, Ber13}. Essentially, the fear function can be modeled  through various functional types, that also preserve spatial features of the underlying domain. For instance, consider competition between two warring drug cartels, where the weaker cartel has certain territorial strongholds. This situation can be modeled through oscillatory functions like $\sin^2(nx)$, where $n\in \mathbb{N}$. Moreover, for complex ecological models, the result of these theorems also holds, where we need to model the fear responses by a sequence of functions or in the form of piece-wise functions. We display some of these functions in Fig.~\ref{fig:fear_plot}. Many of the relevant theorem have been tested by choosing such functions (See Figs.~\ref{fig:pde_ce1_pde1},\ref{fig:pde_co1_pde1},\ref{fig:pde_sc1_pde1},\ref{fig:pde_two_post_pde1}). Note, all of our PDE simulations for the case of a spatially heterogeneous fear function were performed in \verb|MATLAB| R2021b, using the \verb|MATLAB| inbuilt function pdepe, which is used in solving $1-$D parabolic and elliptic PDEs. The simulations were run on an $8-$core CPU, Apple $M1$ pro-chip-based workstation. In this configuration, a typical simulation takes between $5-7$ seconds when the spatial domain is taken to be $[0,1]$, and is partitioned into $1000$ sub-intervals.

In the weak competition case without fear, once there is a critical level of fear, in either $u$ or $v$, a competitive exclusion type phenomenon will occur, see Lemma \ref{lem:ce_ode_2}. This occurs again via a transcritical bifurcation, see Theorem \ref{thm:tc} and Fig.~[\ref{fig:Fig1},\ref{fig:trans}]. This again is interesting in the spatially explicit setting, as via Theorem \ref{thm:ce_1_pde}, where the fear function need not be above $k_{c}$ so as to induce competitive exclusion. We have explored certain types of fear functions herein, see Fig.[\ref{fig:pde_ce1_pde1},\ref{fig:remark_5},\ref{fig:remark_51},\ref{fig:remark_52}], but the effect of functions such as in Fig. \ref{fig:fear_plot}, remain to be explored. Here one could look at the effect of fear in fragmented domains vs domains that are not, via deriving or enforcing conditions relating the fear function $k(x)$ to the resource function $m(x)$, see similar ideas explored in \cite{DeAngelis2016}. It would be interesting to explore applications of these sorts to competing political parties, where a nationally weaker party with smaller voter bank, could have several local strongholds. What levels of fear need to be instilled in those strongholds so as to enforce competitive advantage, is an apt question. Herein estimates on the $\mathbb{L}^{1}(\Omega)$ norm of $k(x)$, would be interesting to derive. Currently Lemma \ref{lem:ff1}, only gives lower estimates - but sharp upper estimates or even just upper estimates are unknown, and would make for interesting future work. The effect of fear on the strong competition setting is perhaps the least interesting dynamically. Herein, having fear in $u$ or $v$, only shifts the interior saddle equilibrium - but qualitatively the dynamics remain the same, as in a strong competition type scenario persists, no matter what level of fear ($f$ or $k$) is chosen. Also it would be of interest to rigorously prove Conjecture \ref{conj:co_1_pde3} and Conjecture \ref{conj:pde_st1}. Our current numerical evidence clearly motivates their validity.

Also, in \cite{Wang16}, it is found that various monotonically decreasing fear functions essentially yield the same dynamics. This has not been tested in the current work, in the case of competitive systems. Neither have we investigated rigorosly the case of both competitors being fearful of each other (see \ref{app} for some preliminaries), with possible different fear functions, indicative of different LOF for each competitor. Here again, upper estimates on the fear functions would be useful and we could allude to methods and techniques explored in \cite{Mazari2020a} and in \cite{Zhang20}. Here, one could consider a spatially dependent growth function as well, such as in \cite{Zhang20}, and attempt to derive conditions relating the function describing the resources to the fear function. Furthermore, it would make for interesting future work if certain choices of (density) dependence or fear functions which lead to degenerate dynamics, can cause periodic orbits. In the current scenario Lemma \ref{lem:dc1} and Lemma \ref{lem:dc2}, do not allow limit cycle dynamics. All in all, we hope these questions lead to future investigations of the fear effect in competitive systems, as a host of rich applications exist.  
	
\section{Acknowledgement}	
VS and RP acknowledge valuable summer support from the National Science Foundation via DMS 1715044.

\section{Appendix}
\label{app}

\subsection{The case of $u$ fearing $v$}

\begin{lemma}\label{3eq1}
	The trivial steady state $E_1$ is locally unstable.
\end{lemma}

\begin{proof}
	Evaluating $J^*$ at $E_1$ yields the following characteristic equation: $$\left(a_1 - \lambda \right) \left(a_2 -\lambda \right)=0.$$ Clearly, $\lambda_1=a_1>0$ and $\lambda_2=a_2>0$ and hence $E_0$  is locally unstable.
\end{proof}

\begin{lemma}\label{3eq2}
	The boundary equilibrium point $E_2$ is locally stable iff $a_2b_1<c_2a_1$.
\end{lemma}
\begin{proof}
	On evaluating the $J^*$ at $E_2$ , we have
	\begin{equation*}
		\widehat{J}^*(E_2)=
		\left(
		\begin{array}{cc}
			-a_1 & -\dfrac{a_1^2 f}{b_1} - \dfrac{c_1 a_1}{b_1} \\
			0&  a_2 - \dfrac{c_2 a_1}{b_1} \\
		\end{array}
		\right).
	\end{equation*}
	Being a triangular matrix, we know that the above matrix has two eigenvalues, $\lambda_1=-a_1$ and $\lambda_2=a_2 - \dfrac{c_2a_1}{b_1}$. Under the assumed parametric restriction 
	\[ a_2b_1<c_2a_1  \iff \lambda_2:=a_2 - \dfrac{c_2 a_1}{b_1}<0.\]
	Hence, the equilibrium point $E_2$ is locally stable
	
\end{proof}

\begin{lemma}\label{3eq3}
	The boundary equilibrium point $E_3$ is locally stable iff $f>\dfrac{a_1 b_2^2 - a_2 b_2 c_1}{a_2^2 c_1}$.
\end{lemma}

\begin{proof}
	Let $f>\dfrac{a_1 b_2^2 - a_2 b_2 c_1}{a_2^2 c_1}$. A similar evaluation of $J^*$ at $E_3$ gives
	\begin{equation}
		J^*(E_3)=
		\left(
		\begin{array}{cc}
			\dfrac{a_1 b_2}{a_2 f+b_2}-\dfrac{a_2 c_1}{b_2} & 0 \\
			-\dfrac{a_2 \left(c_2\right)}{b_2} & -a_2 \\
		\end{array}
		\right)
	\end{equation}
	and its corresponding characteristic equation is
	\[\dfrac{\left(a_2+\lambda \right) \left(a_2 b_2 c_1-a_1 b_2^2+a_2^2 c_1 f+\lambda (b_2^2  +a_2 b_2 f) \right)}{b_2 \left(a_2 f+b_2\right)}=0.\]
	The associated eigenvalues are $\lambda_1=-a_2<0$ and $\lambda_2=\dfrac{-a_2 b_2 c_1+a_1 b_2^2-a_2^2 c_1 f}{b_2 \left(a_2 f+b_2\right)}$. Since 
	\[f>\dfrac{a_1 b_2^2 - a_2 b_2 c_1}{a_2^2 c_1} \iff \lambda_2<0,\]
	and hence $E_3$ is locally stable.
\end{proof}

\begin{lemma}\label{3eq4}
	The interior equilibrium $E_4$ exists and is locally stable  if  
	\[ f < \frac{1}{a_1} \Big( \frac{b_1b_2}{c_2} - c_1 \Big).\]
\end{lemma}

\begin{proof}
	On evaluating  $J^*$ at $E_4$ , we have
	\begin{equation*}
		\widehat{J}^*(E_4)=
		\left(
		\begin{array}{cc}
			-b_1 u^* & -\dfrac{a_1 f u^*}{(f v^* +1)^2} - c_1 u^* \\
			-c_2 v^*&  -b_2 v^*\\
			
		\end{array}
		\right).
	\end{equation*}
	To claim local stability of interior equilibrium $E_4$, we need to show $Trace(J^*(E_4))<0$ and $Det(J^*(E_4))>0$. Simple calculation yields
	\[ Trace(J^*(E_4)) =  -b_1 u^* -b_2 v^*<0\]
	and 
	\begin{align*}
		\begin{split}
			Det(J^*(E_4)) &= u^* v^* \Big[b_1 b_2 - c_2 \Big( \dfrac{a_1 f}{(f v^* +1)^2} + c_1 \Big)\Big].
		\end{split}
	\end{align*}
	Note,
	\[ a_1 f + c_1 > \dfrac{a_1 f}{(f v^* +1)^2} + c_1.\]
	Hence, under the assumption $f < \frac{1}{a_1} \Big( \frac{b_1b_2}{c_2} - c_1 \Big)$, we have
	\[ c_2 \Big( \dfrac{a_1 f}{(f v^* +1)^2} + c_1 \Big) < c_2 \Big( a_1 f + c_1 \Big) < b_1 b_2 \implies Det(J^*(E_4))>0,\]
	and the result follows.
\end{proof}

\begin{lemma}\label{3eq5}
	The interior equilibrium $E_4$ exits and is saddle  if  
	\[\Big( \dfrac{b_1 b_2}{c_2} -c_1 \Big)< \dfrac{a_1 c_2^2 f}{(f a_2 +c_2)^2}.\]
\end{lemma}

\begin{proof}
	To claim local stability of interior equilibrium $E_4$, we need to show $Trace(J^*(E_4))<0$ and $Det(J^*(E_4))<0$. Consider,
	\[ Trace(J^*(E_4)) =  -b_1 u^* -b_2 v^*<0\]
	whereas
	\begin{align*}
		\begin{split}
			Det(J^*(E_4)) &= u^* v^* \Big[ b_1 b_2 - c_2 \Big( \dfrac{a_1 f u^*}{(f v^* +1)^2} + c_1 \Big)\Big].
		\end{split}
	\end{align*}
	From the density of reals and nullclines of $v$, we have
	\[ \dfrac{a_1 c_2^2 f}{(f a_2  + c_2)^2} < \dfrac{a_1 f}{\Big[ f (\frac{a_2}{c_2} - c_2 u^*) +1\Big]^2} = \dfrac{a_1 f}{(f v^*+1)^2}.\]
	Moreover, under the assumption, we have
	\[ \dfrac{b_1 b_2}{c_2} - c_1 < \dfrac{a_1 c_2^2 f}{(fa_2  + c_2)^2} < \dfrac{a_1 f}{(f v^* +1)^2}.\]
	On further rearrangement,
	\[ b_1b_2 - c_2 \Big( \dfrac{a_1 f}{(f v^* +1)^2} + c_1 \Big) <0 \implies  Det(J^*(E_4))<0,\]
	and the result follows.
\end{proof}

\subsection{Case of both species $u$ and $v$ fearing each other}
We consider the case of the both the competitor $v$ and $u$ are fearful to each other. Thus in the classical model \eqref{eq:GeneralEquation}, we model the fear effect as in \cite{Wang16}, where the growth rate of both the competitor $v$ and $u$, is not constant but rather density dependent. Essentially, the growth rate is decreased by a factor $\approx \frac{1}{1+k u}$, and $\approx \frac{1}{1+f v}$, where $k,f \geq 0$ is a fear coefficient. When $k,f=0$, the assumption is there is no fear, and one recovers the classical model \eqref{eq:GeneralEquation}. If fear is present, we obtain the following ODE model for two competing species $u$ and $v$: 


\begin{equation}\label{eq: Fearmodel1}
	\begin{split}
		\dfrac{du}{dt} &= \dfrac{a_1 u}{1+f v}-b_1 u^2 -c_1 uv, \\  
		\dfrac{dv}{dt}  &=  \dfrac{a_2 v}{1+k u}-b_2 v^2 -c_2 uv. \\
	\end{split}
\end{equation}

\textit{Equilibria:} In the event that there is no fear present, or $f=k=0$, \eqref{eq: Fearmodel1} reduces to the classical competition model \eqref{eq:GeneralEquation}. The dynamical analysis of \eqref{eq: Fearmodel1} when $f,k>0$, leads to 4$^{th}$ order polynomial analysis $\eqref{eq: vpoly}.$ The system (\ref{eq: Fearmodel1}) possesses the following biologically feasible equilibria. These are 
\begin{itemize}
\item $\mathbf{E}_1=(0,0)$,
\item $\mathbf{E}_2=\left(\dfrac{a_1}{b_1},0 \right)$,
\item $\mathbf{E}_3=\left(0,\dfrac{a_2}{b_2} \right)$,
\item $\mathbf{E}_4=\left(u^*,v^*\right)$,
\end{itemize}
where 
\begin{align}\label{u_star}
	u^*=\dfrac{1}{b_1}\left(\dfrac{a_1}{1+fv^*}-c_1 v^* \right)
\end{align}
and $v^*$ is a positive root of the following fourth-order polynomial
\begin{equation}
\label{eq: vpoly}
A(v^*)^4+B(v^*)^3+C(v^*)^2+D(v^*)+E=0
\end{equation}
where  

\begin{equation}
\begin{split}
A&=c_1 f^2 k \left(b_1 b_2-c_1 c_2\right),\\
B&=-f \left(b_1 b_2-c_1 c_2\right) \left(b_1 f-2 c_1 k\right),\\
C&=b_2 b_1 \left(-a_1 f k-2 b_1 f+c_1 k\right)+c_1 c_2 \left(2 f \left(a_1 k+b_1\right)-c_1 k\right)+a_2 b_1^2 f^2,\\
D&=c_2 \left(c_1 \left(2 a_1 k+b_1\right)-a_1 b_1 f\right)+2 a_2 b_1^2 f-b_2 b_1 \left(a_1 k+b_1\right),\\
E&=a_2 b_1^2-a_1 c_2 \left(a_1 k+b_1\right).
\end{split}
\end{equation}

The Jacobian matrix of system $(\ref{eq: Fearmodel1})$ is given by

\begin{equation}\label{jacob_main}
	\mathbf{J}(u^*,v^*)=
	\left(
	\begin{array}{cc}
		\dfrac{a_1}{fv^*+1}-2b_1u^*-c_1v^* &-\dfrac{a_1 f u^*}{(fv^*+1)^2} -c_1 u^* \\
		-\dfrac{a_2 k v^*}{(k u^*+1)^2}-c_2 v^* & \dfrac{a_2}{k u^*+1}-2 b_2 v^*-c_2 u^* \\
	\end{array}
	\right).
\end{equation}

\subsection{Transcritical bifurcation}

\begin{theorem}
The model \eqref{eq:ODE2} undergoes a transcritical bifurcation around $E_3^*$ when $a_1=a_1^*=\dfrac{c_1a_2}{b_2}$ and $c_1 \neq \dfrac{b_1b_2}{ka_2+c_2}$.
\end{theorem}

\begin{proof}
The Jacobian matrix for system \eqref{eq:ODE2} evaluated at $E_3$ with $a_1^*=\dfrac{c_1a_2}{b_2}$ is given as 
\begin{equation}\label{trans}
JJ^*=
\left(
\begin{array}{cc}
     0& 0 \\
 -\dfrac{a_2}{b_2} \left(a_2 k+c_2\right) & -a_2 \\
\end{array}
\right).
\end{equation}
The corresponding eigenvalues to the Jacobian of  \eqref{eq:ODE2} in Eq. (\ref{trans}) are $\lambda_1=0$ and $\lambda_2=-a_2.$ Clearly, there is a zero eigenvalue at  $a_1=a_1^*=\dfrac{c_1a_2}{b_2}$. Next, we let $W=(w_1,w_2)^T$ and $Z=(z_1,z_2)^T$ represent the eigenvectors  related to the zero eigenvalue of the matrices $JJ^{*}$ and $JJ^{*T}$ respectively.

We obtain $W=\left(-\dfrac{b_2}{ka_2+c_2},1 \right)^T$ and $Z=\left(1,0 \right)^T$.  Now, let $R=(R_1,R_2)^T$ where 
	\begin{equation*}
		\begin{split}
			R_1&= a_1u-b_1 u^2-c_1uv,\\
			R_2&= \dfrac{a_2 v}{1+ku}-b_2v^2-c_2 uv.
		\end{split}
	\end{equation*}
	Presently, we validate the transversality conditions using the Sotomayor's theorem \cite{perko2013differential}. Now,
	$$Z^TR_{a_1^*}(E_3^*,a_1)=\left(1,0 \right) \left(0,0 \right)^T= 0.$$
	Also,
	\begin{align*}
Z^{T}\left[DR_{a_1}\left(E_3,a_1^* \right)W\right] &= \left(
\begin{array}{cc}
 1 & 0 \\
\end{array}
\right) 
\left(
\begin{array}{ccc}
 1& 0  \\
  0 & 0 \\
\end{array}
\right)
\left(
\begin{array}{ccc}
 w_1 \\
  w_2 \\
\end{array}
\right) \\
&=-\dfrac{b_2}{ka_2+c_2}\neq 0.
\end{align*}

and 
\begin{equation*}
\begin{split}
Z^{T}\left[D^2 R\left(E_3,a_1^* \right)(W,W)\right] &= \left(
\begin{array}{cc}
 1 & 0 \\
\end{array}
\right) 
\left(
\begin{array}{ccc}
 \dfrac{2b_2}{ka_2+c_2}\left(c_1-\dfrac{b_1 b_2}{ka_2+c_2} \right)  \\
  \dfrac{2a_2^2 k^2 b_2}{\left(ka_2+c_2 \right)^2} \\
\end{array}
\right) \\
& =\dfrac{2b_2}{ka_2+c_2}\left(c_1-\dfrac{b_1 b_2}{ka_2+c_2} \right) \neq 0.
\end{split}
\end{equation*}

Therefore by the Sotomayor's theorem system (\ref{eq:ODE2}) experiences a transcritical bifurcation at $a_1=a_1^*=\dfrac{c_1a_2}{b_2}$ around $E_3^*$.
\end{proof}

\end{document}